\documentclass[11pt]{article}
\usepackage[margin=1in]{geometry}
\usepackage{amsmath}
\usepackage{graphicx}
\usepackage{enumerate}
\usepackage{natbib}
\usepackage{url} % not crucial - just used below for the URL 
\usepackage{mathtools}
\usepackage{empheq}
\usepackage{amsfonts}
\usepackage{amsgen,amsmath,amstext,amsbsy,amsopn,amssymb,amsthm,stmaryrd,bbm,subcaption}
\usepackage{cprotect}
\usepackage{natbib}
\usepackage[dvipsnames]{xcolor}
\usepackage{comment}
\usepackage{enumitem}
\usepackage{multirow,booktabs,makecell}
\usepackage[affil-it]{authblk}
\usepackage[ruled,vlined,longend]{algorithm2e}
\newtheorem{thm}{Theorem}
\newtheorem{lem}{Lemma}

\newtheorem{cor}{Corollary}

\newtheorem{claim}{Claim}
\newtheorem{assump}{Assumption}

\newcommand{\ones}{{{\mathbbm{1}}}}
\newcommand{\ind}[1]{\ones\{#1\}}
\newcommand{\bR}{\mathbb{R}}
\newcommand{\bO}{\mathbb{O}}
\newcommand{\cO}{\mathcal{O}}
\DeclareMathOperator*{\argmax}{arg\,max}
\DeclareMathOperator*{\argmin}{arg\,min}
\newcommand{\h}[1]{\widehat{#1}}
\newcommand{\til}[1]{\widetilde{#1}}
\newtheorem{remark}{Remark}
\newtheorem{theorem}{Theorem}
\def\bSig\mathbf{\Sigma}

\newcommand{\cG}{\mathcal{G}}

\newcommand{\bbR}{\mathbb{R}}
\newcommand{\bbE}{\mathbb{E}}
\newcommand{\bbP}{\mathbb{P}}
\newcommand{\bbO}{\mathbb{O}}

\DeclareMathSymbol{\sms}{\mathbin}{AMSa}{"39}
\makeatletter
\newcommand*{\rom}[1]{\expandafter\@slowromancap\romannumeral #1@}
\providecommand{\keywords}[1]
{ 
	\small	
	\textbf{\textit{Keywords---}} #1
}

\begin{document}
\title{\bf Cluster Quilting: Spectral Clustering for Patchwork Learning}
\author[1]{Lili Zheng$^\dagger$\thanks{Corresponding author: lilizheng.stat@gmail.com}}
\author[2]{Andersen Chang$^\dagger$}
\author[3]{Genevera I. Allen}
		
\affil[1]{Department of Statistics, University of Illinois Urbana-Champaign}
\affil[2]{Department of Neuroscience, Baylor College of Medicine}
\affil[3]{Department of Statistics, Columbia University}

\date{}
\maketitle
\begin{abstract}
		Patchwork learning arises as a new and challenging data collection paradigm where both samples and features are observed in fragmented subsets. Due to technological limits, measurement expense, or multimodal data integration, such patchwork data structures are frequently seen in neuroscience, healthcare, and genomics, among others. Instead of analyzing each data patch separately, it is highly desirable to extract comprehensive knowledge from the whole data set. In this work, we focus on the clustering problem in patchwork learning, aiming at discovering clusters amongst all samples even when some are never jointly observed for any feature. We propose a novel spectral clustering method called Cluster Quilting, consisting of (i) patch ordering that exploits the overlapping structure amongst all patches, (ii) patchwise SVD, (iii) sequential linear mapping of top singular vectors for patch overlaps, followed by (iv) k-means on the combined and weighted singular vectors. Under a sub-Gaussian mixture model, we establish theoretical guarantees via a non-asymptotic misclustering rate bound that reflects both properties of the patch-wise observation regime as well as the clustering signal and noise dependencies. We also validate our Cluster Quilting algorithm through extensive empirical studies on both simulated and real data sets in neuroscience and genomics, where it discovers more accurate and scientifically more plausible clusters than other approaches.
	\end{abstract}
	\noindent
\footnotetext{\hspace{-0.5em}$\dagger$: Equal contribution.}

	\section{Introduction}
	\label{s:intro}
	
	Clustering is a fundamental statistical tool for discovering prominent patterns from large-scale data sets, such as finding gene pathways \citep{dalton2009clustering}, tumor cell subtypes \citep{shen2009integrative}, text topics \citep{aggarwal2012survey}, communities in social networks \citep{mishra2007clustering}, and many others. Various clustering techniques were proposed in the literature and widely applied, e.g., k-means \citep{macqueen1967some}, hierarchical clustering \citep{bridges1966hierarchical}, convex clustering \citep{chen2015convex}, and spectral clustering \citep{donath1972algorithms}. Although clustering for fully observed data has been extensively studied, in this paper, we consider a novel setting where one only has access to a collection of data patches, i.e., partial observations of both features and samples (see e.g., Fig. \ref{fig:illustration_PL}). This is motivated by many biomedical applications. For example, in neuroscience, genomics, and proteomics, simultaneous measurement of all observations of interest can be extremely expensive or impossible due to technological constraints; thus, in these situations, data may be collected such that only a subset of features and observations are recorded within a single experimental session \citep{microns2021functional, browning2008missing,baghfalaki2016missing}. Another example arises in healthcare applications, where data collected on patients from different sites or hospitals may only cover distinct subsets of all possible modalities due to clinician discretion or availability of facilities and thus lead to a lack of measurements for entire blocks of observations and features \citep{cai2019survey,kline2022multimodal}. This particular data missingness problem was introduced as the new ``patchwork learning paradigm'' in \cite{rajendran2023patchwork}; there, the authors also outline numerous opportunities and challenges in this area, one of which is the fragmented observational structure where missingness is non-random. Given such patchwise observations, naively applying existing clustering techniques to each patch fails to give a holistic view of the whole population. Instead, we aim to effectively utilize all data patches and perform a comprehensive clustering for all samples available. We term this as the ``cluster quilting" problem, following a recent work on ``graph quilting" \citep{vinci2019graph} which aims to learn graph structure from a similar type of data.

	\subsection{Problem Set-up}
	To fix ideas, we start by introducing the commonly studied mixture model for data generation and the patchwork data observation setting. Suppose there are $K>1$ clusters with $p$-dimensional centroids $\theta^*_1,\dots,\theta^*_K$. For each sample $1\leq i\leq n$, let $z_i\in [K]$ be its cluster label, and its associated full feature vector $X_i=\theta_{z_i}+\epsilon_i\in \mathbb{R}^p$ with mean-zero random noise $\epsilon_i$. Let $X=(X_1,\cdots,X_n)\in \bR^{p\times n}$ be the data matrix, $\Theta^* = (\theta_1,\cdots,\theta_K)\in \bR^{p\times K}$ be the centroid matrix, $E = (\epsilon_1,\dots,\epsilon_n)\in \bR^{p\times K}$ be the noise matrix, and $F^*\in \bR^{n\times K}$ encode the cluster membership: $F^*_{i,k} = \ind{z_i=k}$. The model can be written in the following matrix format:
	\begin{equation}\label{eq:model}
	X = X^* + E = \Theta^*F^{*\top}+E\in \bR^{p\times n},
	\end{equation}
	where $X^* = \bbE X$ is the expectation of observational data. Our goal is then to assign the $n$ samples into $K$ clusters, or equivalently, to estimate the cluster labels $z_1,\dots,z_n$ up to permutation on $[K]$. However, in the patchwork observation setting, we do not have access to the full data matrix $X$ but rather only subsets of the rows and columns structured as patches. More formally, we suppose here that all $p$ features can be divided into $M$ disjoint subsets $T_1,\,\dots,\,T_M$, $\cup_{m=1}^M T_m= [p]$ and that each feature subset $T_m$ of size $p_m$ is associated with a corresponding observation subset $I_m\subset[n]$ of size $n_m$; the observed data thus consists of submatrices of the full data: $\{X_{T_1,I_1},\dots,X_{T_M,I_M}\}$. Following the prior work on patchwork learning \citep{rajendran2023patchwork}, we also use a graph to characterize the observational relationship between different feature blocks (e.g., modalities). In particular, consider a graph $\cG=(V, E)$ where the node set $V=[M]$ are the $M$ patches, and $(j,k)\in E$ if and only if $I_j\cap I_k \neq \emptyset$ (see Fig. \ref{fig:illustration_PL} for illustration). Throughout the paper, we assume the graph $\cG$ to be connected so that the signals can be shared across different patches.
	
	\begin{figure}[t]
		\centering
		\begin{subfigure}[t]{0.45\linewidth}
			\centering
			\includegraphics[width=.95\linewidth]{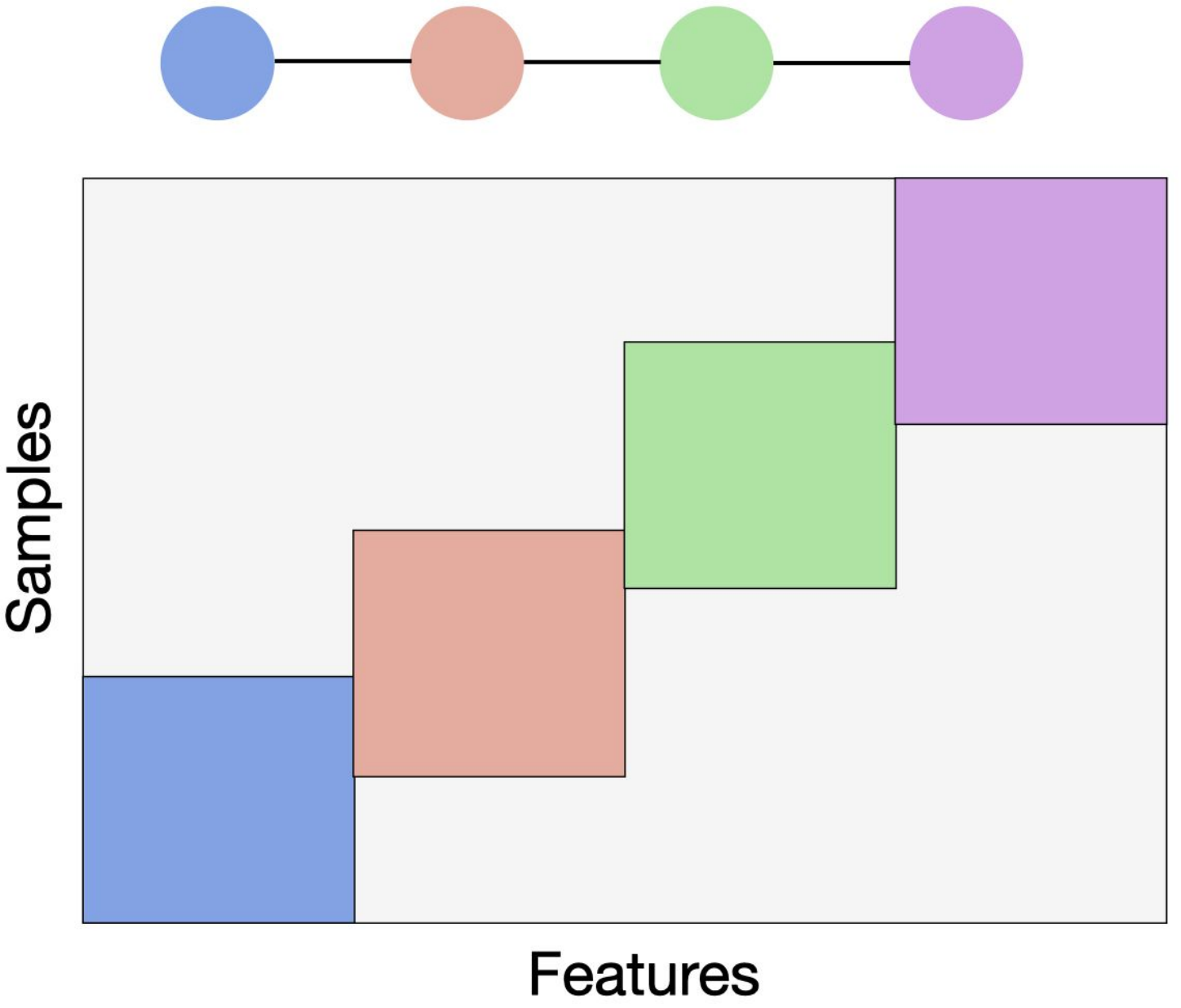}
			\caption{Sequential patches motivated by neuronal functional recording via calcium imaging, where feature subsets are distinct time periods and samples are neurons.}
			\label{fig:visualize_sequential}
		\end{subfigure}
		\begin{subfigure}[t]{0.45\linewidth}
			\centering
			\includegraphics[width=.95\linewidth]{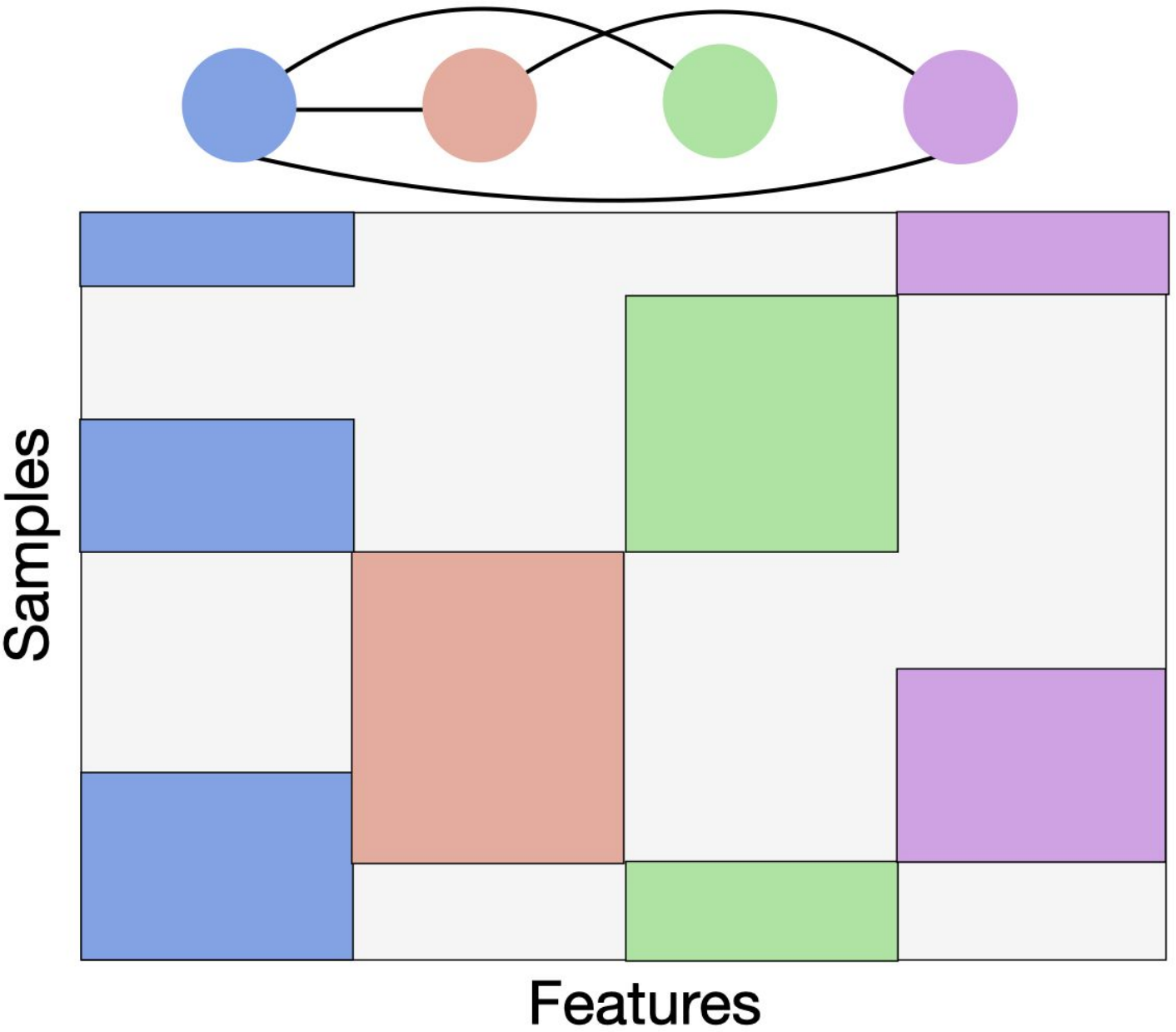}
			\caption{Mosaic patches motivated by healthcare and genomics, where disjoint feature subsets cover different modalities.}
			\label{fig:visualize_multiview}
		\end{subfigure}
		\caption{Two examples of patchwork observational patterns. The graphs above represent the overlapping relationships between disjoint feature blocks.}
		\label{fig:illustration_PL}
	\end{figure}

	\subsection{Related Literature}
	Many different approaches have previously been proposed for clustering with missing data. Classically, two-step approaches have been applied in which missing cases are either deleted or imputed \citep{inman2015case,aktas2019data}; the former can not be applied in the patchwork learning scenario because there do not exist any fully observed rows or columns, while the performance of latter does not come with any statistical guarantees about clustering performance. More recently, several papers have proposed general likelihood-based optimization approaches for estimating clusters with missing data. Examples of these include a constrained k-means objective estimated by majorization-minimization \citep{chi2016k}, an $\ell_0$ fusion-penalty based optimization method \cite{poddar2019clustering}, a likelihood-based approach to estimate mixture models parameters from missing data \cite{hunt2003mixture,lin2006fast}, a modified fuzzy c-means approach with an adjusted objective to account for data missingness \cite{hathaway2001fuzzy}, and a Bayesian model-based random labeling point process approach \cite{boluki2019optimal}. 
	
	However, one key difference between the setting used in many prior papers and the one we consider in this paper is in the structure of the missing entries. Cluster quilting is particularly a non-trivial task because of the imposed block observation scheme, as the missingness pattern is highly non-random and very dense. Additionally, many pairs of samples are only observed on disjoint sets of features, meaning that there is no direct way to compute similarities between many of the pairs of samples in the data. Although imputation comes up as a natural idea, how to accurately impute the data or similarity matrix from non-random observational patterns is still a challenging and active research area \citep{chatterjee2020deterministic,yan2024entrywise,choi2023matrix}, where most prior works consider missing patterns different from ours; it is also unclear whether the imputed data suffices for the downstream clustering task in the patchwork observation setting. Therefore, previous results on theoretical guarantees and empirical performance of the likelihood-based approaches outlined above may not be applicable here. In particular, many of the previous works lack provable theoretical computational convergence, while others only study the case of random missingness in the data matrix. Some prior existing methods formulate cluster quilting as various optimization problems \citep{wang2019k, wen2022survey}, but these methods lack statistical guarantees of convergence, relying on non-convex optimization with no assurance of global algorithmic convergence.

	\subsection{Contribution and Organization}
	In this paper, we propose a new spectral-based method for clustering incomplete data called the Cluster Quilting algorithm; this is introduced in Section \ref{sec:model}. Broadly, our Cluster Quilting procedure consists of the following steps: (i) order the patches based on their overlapping signals; (ii) learn the top singular vectors of each data patch; (iii) perform sequential matching across data patches in the sample space; (iv) obtain an estimate of the full top singular spaces and apply k-means on the singular value-weighted singular vectors. While much work has been done regarding the theoretical and empirical properties of spectral clustering under the stochastic block \citep{rohe2011spectral,jin2015fast,lei2015consistency,zhou2023heteroskedastic} and Gaussian mixture \citep{loffler2021optimality,abbe2022l} data generation models which we study in this work, we are the first to study spectral clustering for the general patchwork learning observation setting. Our method is most similar to those proposed in \cite{bishop2014deterministic}, \cite{yu2020optimal}, and \cite{chang2022lowrank}, which also utilize a low-rank assumption in order to perform matrix completion for downstream analysis and modeling in the patchwork learning context. However, our work fundamentally differs from these previous works in terms of the considered data-generative model and the clustering objective, the general patchwork observations that go beyond sequential observational pattern \citep{chang2022lowrank}, as well as our new method that starts with a theory-motivated, data-driven patch ordering step. In particular, we provide novel theory for our cluster quilting algorithm in Section \ref{sec:theory}, characterizing the conditions and misclustering rate in terms of separation distance or minimum singular values in the data patches; the result presents a much tighter bound compared to the previous literature on matrix completion \citep{bishop2014deterministic}. We also prove that the data-driven selection of patch ordering guarantees a similar performance to the oracle ordering, with only a constant factor inflation in the SNR condition and misclustering rate. Lastly, we validate our method through extensive empirical studies on simulated and real data sets in Sections \ref{sec:gsim}-\ref{sec:real}, in which we show that Cluster Quilting improves upon previous methods proposed for incomplete spectral clustering in the task of recovering true cluster labels in the patchwork observation setting.

	\section{Cluster Quilting Algorithm}
	\label{sec:model}
	Inspired by the success of spectral clustering on sub-Gaussian mixture models \citep{loffler2021optimality}, we propose a spectral algorithm called \emph{Cluster Quilting}, where the key steps are learning the top right singular vectors of $X^*$. Due to the blockwise missingness, the full right singular vectors of $X^*$ cannot be recovered from any data patch alone; but instead, we need to combine all pieces together and match them using linear transformations. Denote by $r>0$ the rank of the ground truth matrix $X^*$. We define the SVD of $X^*$ as $X^* = U^*\Lambda^*V^{*\top}$ for orthogonal matrices $U\in \bO^{p\times r}$, $V\in \bR^{n\times r}$, and diagonal matrix $\Lambda\in \bR^{r\times r}$. Since we do not assume the cluster centroid matrix $\Theta^*$ to be of full rank, the rank $r\leq K$ does not necessarily equal $K$. For now, we assume the data patches are appropriately ordered and each sample subset $I_m\cap (\cup_{m'<m}I_{m'})\neq \emptyset$; in fact, such an order always exists under our assumption that the observational graph $\cG$ is connected. We will introduce several theory-inspired strategies to choose the ordering in Section \ref{sec:practicality}. 
	
	Given a fixed ordering, our algorithm consists of three main steps. First, we perform singular value decomposition for each patch $X_{T_m,I_m}$ and take its top $r$ singular vectors. We then find the best invertible matrix to match the eigenvectors of the $m$th patch with the previous $m-1$ patches, for $2\leq m\leq M$. Lastly, we perform k-means clustering on the full estimate of top singular vectors. The full cluster-quilting algorithm is summarized in Algorithm \ref{alg:bsvdcq}. We note here that the blockwise SVD and consecutive block matching ideas also appeared in an earlier work \citep{bishop2014deterministic} for PSD matrix completion from block observations. However, we consider a different model, i.e. a mixture model, with a different objective, i.e. clustering, and hence our matching step looks for the best linear transformation between patches through least squares instead of finding the best rotation matrix as in \cite{bishop2014deterministic}. 

		{\LinesNumberedHidden
			\begin{algorithm}[t] 
				\caption{Cluster Quilting} \label{alg:bsvdcq}
				\SetAlgoNoLine
				\SetKwFunction{Union}{Union}\SetKwFunction{FindCompress}{FindCompress}
				\SetKwInOut{Input}{Input}\SetKwInOut{Output}{Output}
				\begin{small}
					\textbf{Input:} $\{X_{T_m,I_m}\in \mathbb{R}^{|T_m|\times |I_m|},m \in 1, \hdots M\}$, $r > 0$\\
					\textbf{Initialize:} $\widetilde{H} = \boldsymbol{0}_{p \times r}$, $\widetilde{V} = \boldsymbol{0}_{n \times r}$.\\
					\textbf{Find:} low-rank solution for first patch: \\
					\hspace{0.5cm} (1) Let $\widehat{U}_1\widehat{\Lambda}_1 \widehat{V}_1^{\top}$ be the top-$r$ SVD of $X_{T_1,I_1}$\\
					\hspace{0.5cm} (2) Set $\widetilde{V}_{I_1,\, :} = \widehat{V}_1$, $\widetilde{H}_{T_1,\, :} = \widehat{U}_1\widehat{\Lambda}_1$ \\
					\For{$m = 2, \hdots M$}{
						\vspace{7pt}
						\hspace{0.1cm}(1) Find low-rank solution for the $m$-th patch: \\ 
						\hspace{0.4cm}(a) Let $\widehat{U}_m\widehat{\Lambda}_m \widehat{V}_m^{\top}$ be the top-$r$ SVD of $X_{T_m,I_m}$\\
						\hspace{0.4cm}(c) Let $\widehat{H}_m = \widehat{U}_m\widehat{\Lambda}_m$.\\
						\hspace{0.1cm}(2) Merge with previous patches: \\ 
						\hspace{0.4cm}(a) Find overlaps $J_m^{(1)} = (\bigcup_{j = 1}^{m-1} I_j )\cap I_m, \, J_m^{(2)} = \{j: I_m[j] \in J^{(1)}\}$. \\
						\hspace{0.4cm}(b) Find the best linear transform $G_m=((\widehat{V}_m)_{J_m^{(2)},:}^\top (\widehat{V}_m)_{J_m^{(2)},:})^{\dagger}(\widehat{V}_m)_{J_m^{(2)},:}^\top \widetilde{V}_{J_m^{(1)},:}$\;
						\hspace{0.4cm}(c) Set $\widetilde{V}_{I_m\backslash J_m^{(1)},\, :} = (\widehat{V}_m)_{\backslash J_m^{(2)},\, :}G_m$, $\widetilde{H}_{T_m} = \widehat{H}_m(G_m^\top)^{-1}$.
					}
					\textbf{Post-processing of $\widetilde{H}\widetilde{V}^\top$}: Let the SVD of $\widetilde{H}\widetilde{V}^\top$ be $\widehat{U}\widehat{\Lambda}\widehat{V}^\top$.\\
					\textbf{Perform $k$-means on the columns of $\widehat{\Lambda}\widehat{V}^\top$}:
					$$
					(\widehat{z}, \{\widehat{c}_j\}_{j=1}^K)=\arg\min_{z\in [K]^n, \{\widehat{c}_j\}_{j=1}^K\in \mathbb{R}^r}\sum_{i=1}^n\|(\widehat{\Lambda}\widehat{V}^\top)_{:,\,i}-\widehat{c}_{z_i}\|_2^2
					$$
					\Return{Clustering labels $\widehat{z}\in [K]^n$, cluster means $\{\widehat{U}\widehat{c}_j\}_{j=1}^K$}
				\end{small}
				
				\vspace{2pt}
		\end{algorithm}}

	\subsection{Practical Considerations}\label{sec:practicality}
	
		{\LinesNumberedHidden
			\begin{algorithm}[t] 
				\caption{Forward Search for Block Ordering} \label{alg:orderFor}
				\SetAlgoNoLine
				\SetKwFunction{Union}{Union}\SetKwFunction{FindCompress}{FindCompress}
				\SetKwInOut{Input}{Input}\SetKwInOut{Output}{Output}
				\begin{small}
					\textbf{Input:} $\{X_{T_m,I_m}\in \mathbb{R}^{|T_m|\times |I_m|},m \in 1, \hdots M\}$, score function $s(\cdot,\cdot):[M]\times 2^{[n]}\rightarrow [0,\infty)$\\
					\textbf{Initialize:} 
					Let $\hat{\pi}(1)\neq \hat{\pi}(2)\in [M]$ be an ordered pair of patch indices such that 
					$(\hat{\pi}(1),\,\hat{\pi}(2)) = \argmax_{k_1\neq k_2\in [M]}s(k_2,I_{k_1})$. 
					\For{$3\leq k\leq M$}{
						\begin{enumerate}
							\item Let $L = \{\hat{\pi}(1),\,\dots,\,\hat{\pi}(k - 1)\}$.
							\item Let $\hat{\pi}(k) = \argmax_{j: j\notin L}s(j,\cup_{l\in L}I_l)$.
						\end{enumerate}
					}
					\Return{Ordering function $\hat{\pi}(\cdot):[M]\rightarrow[M]$.}
				\end{small}
				\vspace{2pt}
		\end{algorithm}}

	The application of the Cluster Quilting algorithm involves several practical choices. Firstly, the ordering of the patches is critical, as for the first step of Algorithm \ref{alg:bsvdcq} to be applicable, we need to ensure that $I_m\cap (\cup_{l<m}I_l)\neq \emptyset$; (ii) for accurate matching in each step (computation of $G_m$), the overlapping signal within $J^{(1)}_m = I_m\cap (\cup_{l<m}I_l)$ also needs to be sufficient. Motivated by this, we propose to first order the patches appropriately with the following ordering function $\hat{\pi}: [M]\rightarrow [M]$:
	\begin{equation}\label{eq:order_maximizer}
	\hat{\pi} = \argmax_{\pi:[M]\rightarrow[M]}\prod_{m=2}^Ms(\pi(m), \cup_{l<m} I_{\pi(l)}),
	\end{equation}
	where $s(\pi(m), \cup_{l<m} I_{\pi(l)})$ is a generic score function that measures the overlapping signal between the $\pi(m)$'s patch and prior patches $\cup_{l<m} I_{\pi(l)}$ in the sample space, which can also be interpreted as the connection strength in the patch graph $\mathcal{G}$. Our choice to maximize the multiplication of all overlapping scores is motivated by our theoretical results in Section \ref{sec:theory}, and this minimizes the error propagation of sequential patch matching in Alg. \ref{alg:bsvdcq}. The generic score functions can also take different forms: one simple idea is to take $s(k, I) = |I_k\cap I|$ as the overlapping size; we will also introduce a theory-motivated score function in Section \ref{sec:theory} and show theoretical guarantees with this data-driven patch ordering. Given a score function $s(k,I)$, one can solve \eqref{eq:order_maximizer} by searching over the whole space of feasible orderings such that $I_m\cap (\cup_{l<m}I_l)\neq \emptyset$. Since the number of patches $M$ is often not large (under $10$) in our motivating applications, e.g., multiomics data sets or calcium imaging data, this combinatorial optimization problem is still computationally feasible. Given $\hat{\pi}(\cdot)$, we can then apply Alg. \ref{alg:bsvdcq} on $\{X_{T_{\hat{\pi}(m)},I_{\hat{\pi}(m)}}\}_{m=1}^M$. 
	% Suggest: combinatorial search for small number of views, forward search for large number.
	For computational speed up when $M$ is large, we also propose a greedy forward search algorithm for the patch-ordering, summarized in Alg. \ref{alg:orderFor}. In this schema, Alg. \ref{alg:orderFor} first chooses the $\hat{\pi}(1),\,\hat{\pi}(2)$ such that $s(\hat{\pi}(2),I_{\hat{\pi}(1)})$ is maximized over all ordered patch pairs; then it sequentially selects $\hat{\pi}(k)$ that maximizes $s(\hat{\pi}(k),\cup_{l<k}I_{\hat{\pi}(l)})$. 
	
	The Cluster Quilting algorithm also requires the selection of both the number of clusters as well as the rank of the spectral clustering step. To select these hyperparameters, we follow the previously proposed prediction validation approach \citep{tibshirani2005cluster} across a grid of potential hyperparameter values selected a priori. For this approach, the data are split into a training and a test set and clustered separately using each potential combination of rank and number of clusters. A random forest model is then trained on the former to predict the cluster assignments of the observations in the latter. The final rank and number of clusters are chosen by selecting the setting which provides the highest agreement between the labels found by the clustering algorithm and the predictions from the supervised model.

	\section{Theoretical Guarantees}\label{sec:theory}
	
	In this section, we provide theoretical characterization for the clustering performance of our Alg. \ref{alg:bsvdcq}. We first assume that Alg. \ref{alg:bsvdcq} is applied on the original data patch ordering $1,\,\dots,\,M$, without making probabilistic assumptions on the noise distribution; Section \ref{sec:theory_subGauss} gives the consequences under sub-Gaussian noise; Section \ref{sec:theory_datadriven} then provides theory for the data-driven selection of the patch ordering. 
	
	Here, we first introduce some key technical quantities we will use throughout Section \ref{sec:theory}. Define the misclustering rate as: $\ell(\widehat{z},z):=\min_{\phi\in \Phi}\frac{1}{n}\sum_{i\in [n]}\ind{\phi(\widehat{z}_i)\neq z_i},$ where $\Phi=\{\phi: \phi$ is a bijection from $[K]$ to $[K]\}$. We will show upper bounds for this misclustering rate $\ell(\widehat{z},z)$ that hold with high probability. For $1\leq j\leq K$, let $n_j = \sum_{i=1}^n\ind{z_i=j}$ be the sample size in each cluster $j$. Let $\Delta = \min_{j\neq j'}\|\theta_j-\theta_{j'}\|_2$ be the minimum separation distance between clusters, and $\Delta_m = \min_{j\neq j'}\|(\theta_j)_{T_m}-(\theta_{j'})_{T_m}\|_2$ be the separation distance within block $m$. Given a patch ordering $\pi:[M]\rightarrow [M]$, for $m\geq 2$, let $J_m^{(1)}(\pi) = (\cup_{j=1}^{m-1}I_{\pi(j)})\cap I_{\pi(m)}$ be the overlap of the $m$th block with prior blocks. We use $0<\gamma_m(\pi)<1$ to denote the signal proportion of the overlapping part $J_m^{(1)}(\pi)$ in the whole block $I_{\pi(m)}$: $\gamma_m(\pi) = \frac{\sigma_r(X^*_{T_{\pi(m)},\,J_m^{(1)}(\pi)})}{\|X_{T_{\pi(m)},\,I_{\pi(m)}}^*\|}$, and $\gamma(\pi) = \min_{2\leq m\leq K}\gamma_m(\pi)$. When the original patch ordering is in use, e.g., $\pi$ is the identity mapping, we abbreviate these notations to $J_m^{(1)}$, $\gamma_m$, and $\gamma$. Let $\kappa_m$ be the condition number of the $m$th block. Consider the population block SVD $X_{T_m,\,I_m}^* = U_m^*\Lambda_m^*V_m^{*\top}$, and define the matching matrix for block $m$ as $B_m^* = \Lambda^*U_{T_m,:}^*U_m^*\Lambda_m^{*-1}$, which transforms the whole matrix SVD into the block SVD results: $V_m^* = V_{I_m,\,:}^* B_m^*$. We also define $\beta_{\max}$ and $\beta_{\min}$ to quantify the heterogeneity across blocks: $\beta_{\max} = \max_{1\leq m,\,m'\leq M}\|B_m^{*-1}B_{m'}^*\|$, $\beta_{\min} = \min_{1\leq m,\,m'\leq M}\sigma_r(B_m^{*-1}B_{m'}^*)$. Another measure of the heterogeneity of signals across blocks is defined as $\alpha_{m,\,m'} = \frac{\|X^*_{T_m,I_m}\|}{\sigma_r(X_{T_{m'},I_{m'}}^*)}$, and $\alpha_{\max} = \max_{m,\,m'}\alpha_{m,\,m'}$. The following regularity and block-wise SNR conditions are needed for our general theory with the original patch ordering. %$1,\dots,\,M$.
	\begin{assump}[Regularity conditions on centroids]\label{assump:centroid_norm_bnd}
		The clustering centroids $\theta_1,\dots,\theta_K\in \mathbb{R}^{p}$ satisfies $\max_j\|\theta_j\|_2\leq C\Delta = C\min_{j\neq j'}\|\theta_j-\theta_{j'}\|_2$ for some universal constant $C>0$
	\end{assump}
	Assumption \ref{assump:centroid_norm_bnd} is a regularity condition on clustering centroids, which also appears in prior works on spectral clustering \cite{abbe2022l}. It is a mild condition as long as the data matrix is appropriately centered as a preprocessing step.
	
	\begin{assump}[SNR condition]\label{assump:block_eigmin}
		For some universal constant $C>0$, the SNR of each block $m$ satisfies
		$$
		\frac{\sigma_r(X^*_{T_m,I_m})}{\|E_{T_m,I_m}\|} \geq \frac{C\alpha_{\max}(\beta_{\max}+1)}{\beta_{\min}^2}\prod_{l=2}^M\left(\frac{1.1}{\gamma_l}+1\right)\sqrt{\frac{rK\max_j n_j}{\min_j n_j}}.
		$$
	\end{assump} 
	Assumption \ref{assump:block_eigmin} suggests that the clustering separation signal in each block well exceeds the noise level. When the number of blocks is larger (i.e. larger $M$), a stronger SNR condition is also required. This is because our Cluster Quilting algorithm (Alg. \ref{alg:bsvdcq}) requires a matching step that sequentially transforms the SVD results of each block to match its previous block. It takes $M-1$ sequential transformations to obtain the full singular space estimate, each of which leads to an additional error factor $\frac{1.1}{\gamma_l}+1$ depending on the signal ratio between each block and its overlapping part. The sequential error propagation leads to the factor $\prod_{l=2}^M\left(\frac{1.1}{\gamma_l}+1\right)$ in Assumption \ref{assump:block_eigmin}. 
	
	\begin{theorem}\label{thm:main}
		Consider the model set-up described in Section \ref{sec:model} and Alg. \ref{alg:bsvdcq}. Suppose that Assumptions \ref{assump:centroid_norm_bnd}, \ref{assump:block_eigmin} hold.
		Then the outputs of Alg. \ref{alg:bsvdcq} satisfy
		\begin{equation}\label{eq:cluster_err}
		\begin{split}
		\ell(\widehat{z},z)\leq &
		B^2(X^*, \mathcal{O})\frac{rK\max_j n_j}{n}\max_m\frac{\|E_{T_m,I_m}\|^2}{\sigma_r^2(X^*_{T_m,I_m})},\\
		\min_{\phi\in \Phi}\max_{1\leq j\leq K}\|\widehat{U}\widehat{c}_j - \theta_{\phi(j)}\|_2 \leq &B(X^*, \mathcal{O})\sqrt{\frac{rK\max_j n_j}{\min_j n_j}}\max_m\frac{\|E_{T_m,I_m}\|}{\sigma_r(X^*_{T_m,I_m})}\max_j\|\theta_j\|_2.
		\end{split}
		\end{equation}
		where $B(X^*, \mathcal{O}) = \frac{C\alpha_{\max}(\beta_{\max}+1)}{\beta_{\min}^2}\prod_{l=2}^M\left(\frac{1.1}{\gamma_l}+1\right)$ depends on the ground truth matrix $X^*$ and the measurement pattern $\mathcal{O}$. 
	\end{theorem} Theorem \ref{thm:main} suggests that as long as the SNR within each block $\frac{\sigma_r(X^*_{T_m,I_m})}{\|E_{T_m,I_m}\|}$ is sufficiently large, accuracy of both clustering centroid estimation are well controlled by a function of the SNR, rank $r$, number of clusters $K$, and number of blocks $M$.
	\begin{remark}[Consistent clustering]
		Recall that parameters $\alpha_{\max},\,\beta_{\max}\geq 1$, $\beta_{\min}\leq 1$ are defined in the beginning of Section \ref{sec:theory}, and they characterize the heterogeneity of different blocks. When all blocks $X^*_{T_m,I_m}$ are well conditioned and the signal in different blocks are similar, we may safely assume $\frac{\alpha_{\max}(\beta_{\max}+1)}{\beta_{\min}^2}$ to be bounded. If $K\max_j n_j\leq Cn$ (different clusters are balanced) also holds, we know that the misclustering rate $\ell(\widehat{z},z)=o(1)$ as long as $\sigma_r(X^*_{T_m,I_m})\gg \sqrt{r}\prod_{l=2}^M\left(\frac{1.1}{\gamma_l}+1\right)\|E_{T_m,I_m}\|$. 
	\end{remark}
	The patchwise observational pattern plays a critical role in the difficulty of the problem. As illustrated in Assumption \ref{assump:block_eigmin}, Theorem \ref{thm:main}, when $\prod_{l=2}^M \left(\frac{1.1}{\gamma_l}+1\right)$ is larger, the problem gets harder as we need a stronger SNR condition and have a larger misclustering rate. Therefore, we ideally hope for larger $\gamma_l$'s, which are the fraction of the overlapping signal between the $l$th block with previous blocks, as a proportion of the whole signal in the $l$th block: $\gamma_l = \frac{\sigma_r(X^*_{T_l, I_l\cap(\cup_{k<l}I_k)})}{\|X^*_{T_l,I_l}\|}$. Different from prior works for PSD matrix or covariance completion \citep{bishop2014deterministic}, our theory hinges on a much smaller merging factor $(1.1\gamma_l^{-1}+1)$; we use spectral norm instead of Frobenious norm to characterize noise level; and we require no condition on the gap between non-zero singular values. This more refined result is based on a careful induction proof, and it motivates a data-driven, theory-motivated ordering method for the patches in Section \ref{sec:theory_datadriven}.

	In Theorem \ref{thm:main}, the signal strength is captured by the minimum singular value of the ground truth matrix within each block. In fact, we can also express the signal strength in terms of the separation distances between clusters, which might be more natural in the clustering context. In the following, we present a corollary that illustrates the effect of seperation distances on the clustering performance of Alg. \ref{alg:bsvdcq}. 
	\begin{assump}[Separation condition]\label{assump:block_separation}
		Within each block $m$, the minimum separation between clusters $\Delta_m = \min_{j\neq j'}\|\theta_{T_m,j}-\theta_{T_m,j'}\|_2$ satisfies
		$$
		\frac{\sqrt{\min_j n_{j,m}}\Delta_m}{\|E_{T_m,I_m}\|}\geq \frac{C\alpha_{\max}(\beta_{\max}+1)\kappa_r(\Theta^*_{T_m,:})}{\beta_{\min}^2}\prod_{l=2}^M\left(\frac{1.1}{\gamma_l}+1\right)\sqrt{\frac{rK\max_j n_j}{\min_j n_j}},
		$$
		for some universal constant $C>0$, where $\kappa_r(\Theta^*_{T_m,:}) = \frac{\sigma_1(\Theta^*_{T_m,:})}{\sigma_r(\Theta^*_{T_m,:})}$, the condition number of the $m$th block's centroid matrix; $n_{j,m} = \sum_{i\in I_m}\ind{z_i=j}$ is the number of samples within $m$th block and $j$th cluster.
	\end{assump}
	Here in Assumption \ref{assump:block_separation}, the separation SNR $\frac{\sqrt{\min_j n_{j,m}}\Delta_m}{\|E_{T_m,I_m}\|}$ can be approximately viewed as the ratio between the $\ell_2$ separation distance amongst clusters and the noise level for each sample within block $m$. In particular, in the balanced setting ($\max_{1\leq j\leq K}n_{j,m}\leq C\min_{1\leq j\leq K} n_{j,m}$), $$\frac{\|E_{T_m,I_m}\|}{\sqrt{\min_j n_{j,m}}}\leq C\frac{\|E_{T_m,I_m}\|}{\sqrt{n_m}}\leq C\sqrt{\frac{1}{n_m}\sum_{i\in I_m}\|E_{T_m,i}\|_2^2},$$ which is the average $\ell_2$ norm of the noise associated with each sample within block $m$.

	\begin{cor}[SNR in the form of separation distance]\label{cor:main_delta}
		Consider the model set-up described in Section \ref{sec:model} and Alg. \ref{alg:bsvdcq}. Suppose that Assumptions \ref{assump:centroid_norm_bnd} and \ref{assump:block_separation} hold. Then the outputs of Alg. \ref{alg:bsvdcq} satisfy
		\begin{equation}\label{eq:cluster_err_separation}
		\begin{split}
		\ell(\widehat{z},z)\leq &
		2B^2(X^*, \mathcal{O})\frac{rK\max_j n_j}{n}\max_m\frac{\kappa_r^2(\Theta^*_{T_m,:})\|E_{T_m,I_m}\|^2}{\min_{1\leq j\leq K} n_{j,m}\Delta_m^2},
		\end{split}
		\end{equation}
	\end{cor}
	Corollary \ref{cor:main_delta} suggests that in a balanced ($\frac{K\max_jn_j}{n}\leq C$), well-conditioned setting with bounded heterogeneity across blocks, as long as the separation distance within each block satisfies $\Delta_m \gg \sqrt{r}\prod_{l=2}^M\left(\frac{1.1}{\gamma_l}+1\right)$ $\frac{\|E_{T_m,I_m}\|}{\min_{1\leq j\leq K}\sqrt{n_{j,m}}}$, our cluster quilting algorithm achieves consistent clustering with misclustering rate $\ell(\widehat{z},z)=o(1)$. We will provide more comparisons with prior spectral clustering theory in Section \ref{sec:theory_subGauss}, under sub-Gaussian noise assumptions.
 
	\subsection{Special Case: Sub-Gaussian Noise and Dependent Features}\label{sec:theory_subGauss}
	Theorem \ref{thm:main} is a general deterministic result that holds for all noise types. To make it more concrete and comparable with prior works, here we consider an example of mean-zero sub-Gaussian noise. Specifically, suppose that there exist $n'\geq n,\,p'\geq p$, matrices $A_r\in \bR^{n\times n'}$ and $A_c\in \bR^{p\times p'}$, such that 
	\begin{equation}\label{eq:noise_distr}
	E=A_rZA_c^\top,\quad Z\in \bR^{n'\times p'},\quad Z_{j,k} \overset{\text{ind.}}{\sim}\text{sub-Gaussian}(\sigma)\text{ with zero mean}.
	\end{equation} 
	A special case of \eqref{eq:noise_distr} is when $n'=n$, $p'=p$, and $A_r$, $A_c$ are identity matrices; this is equivalent to write $E_{j,k}\overset{\text{ind.}}{\sim}\text{sub-Gaussian}(\sigma)$. By considering general matrices $A_r$ and $A_c$, we also allow linear dependency across features and samples.
	
	In the following, we present the consequence of Theorem \ref{thm:main} under the noise model \eqref{eq:noise_distr}; For presentation simplicity, we consider the following regularity assumption. 
	\begin{assump}[Bounded patch heterogeneity, condition number, balanced clusters]\label{assump:BndHeteroBalanced}
		%Suppose that 
		\begin{equation}\label{eq:BndHeteroBalanced}
		\begin{split}
		&\alpha_{\max},\beta_{\max},\beta_{\min}^{-1}, \frac{\max_{1\leq j\leq K}n_j}{\min_{1\leq j\leq K}n_j}\leq C,\\
		&\frac{\max_{1\leq j\leq K}n_{j,m}}{\min_{1\leq j\leq K}n_{j,m}}, \kappa_{r}(\Theta^*_{T_m,:})\leq C,\quad 1\leq m\leq M.
		\end{split}
		\end{equation} 
	\end{assump}
	Here, $\alpha_{\max},\beta_{\max},\beta_{\min}^{-1}\geq 1$ quantify the heterogeneity of the signal levels in different blocks; $\frac{\max_jn_j}{\min_j n_j}$, $\frac{\max_j n_{j,m}}{\max_jn_{j,m}}$ quantify how unbalanced that different clusters are; $\kappa_{r}(\Theta^*_{T_m,:})$ gives the condition number of the centroid matrix within each block.
	\begin{theorem}\label{thm:main2}
		Consider the model set-up described in Section \ref{sec:model}. Suppose that the noise $E$ satisfies \eqref{eq:noise_distr} and Assumptions \ref{assump:centroid_norm_bnd}, \ref{assump:BndHeteroBalanced} hold.
		If the minimum separation distance within each block $1\leq m\leq M$ satisfies
		\begin{equation}\label{eq:SNR_subGaussian}
		\Delta_m\geq C\sigma\|A_r\|\|A_c\|\sqrt{r}K\prod_{l=2}^M\left(\frac{1.1}{\gamma_l}+1\right)\left(1+\sqrt{\frac{p_m}{n_m}}\right),
		\end{equation} 
		then with probability at least $1-C\sum_{m=1}^M\exp\{-c(n_m+p_m)\}$, the misclustering rate is bounded as follows:
		$\ell(\widehat{z},z)\leq C\sigma^2\|A_r\|^2\|A_c\|^2rK\prod_{l=2}^M\left(\frac{1.1}{\gamma_l}+1\right)^{2}\max_m\left(1+\frac{p_m}{n_m}\right)\Delta_m^{-2}.$
	\end{theorem}
	Theorem \ref{thm:main2} suggests that we can obtain consistent clustering as long as $$\Delta_m\gg \sigma\|A_r\|\|A_c\|\sqrt{r}K^{\frac{3}{2}}\prod_{l=2}^M\left(\frac{1.1}{\gamma_l}+1\right)\left(1+\sqrt{\frac{p_m}{n_m}}\right),$$ where $\sigma\|A_r\|\|A_c\|$ represents the noise level; $\prod_{l=2}^M\left(\frac{1.1}{\gamma_l}+1\right)$ quantifies the effect of matching $M$ semi-overlapping blocks; $\left(1+\sqrt{\frac{p_m}{n_m}}\right)$ reflects the effect of dimension and sample size within each block. Similar to our discussion for Assumption \ref{assump:block_eigmin}, the exponential dependence on $M$ is due to the nature of the observational pattern: the first and last blocks are only connected through the $M-1$ sequential overlaps between consecutive blocks.
	\begin{remark}[Comparison with prior works]
		If the noise matrix $E$ has independent sub-Gaussian-$\sigma$ entries ($A_r=I_{n\times n},\,A_c = I_{p\times p}$) in our setting, then the minimum separation condition for consistent clustering ($\ell(\widehat{z}, z) = o(1)$) would be $\Delta_m\gg \sigma\sqrt{r}K\prod_{l=2}^M\left(\frac{1.1}{\gamma_l}+1\right)\left(1+\sqrt{\frac{p_m}{n_m}}\right)$. As a comparison, in prior literature \citep{loffler2021optimality} on spectral clustering for Gaussian mixtures from the full data matrix, they require $\Delta\gg \sigma\sqrt{r}(1+\sqrt{\frac{p}{n}})$. The additional factor we require $K\prod_{l=2}^M\left(\frac{1.1}{\gamma_l}+1\right)$ is due to the cost of matching $M$ semi-overlapping blocks. One might note that our explicit bound for the misclustering rate only scales as $\Delta^{-2}$ instead of the exponential rate $\exp\{-c\Delta\}$ as in \cite{loffler2021optimality}. This is due to that our blockwise SVD and sequential matching steps in Alg. \ref{alg:bsvdcq} creates complicated dependency structures between our algorithm output and all samples, and hence the proof techniques for deriving exponential error bounds are not applicable. 
	\end{remark}
	One key quantity in Theorem \ref{thm:main2} is $\frac{p_m}{n}$, the ratio between number of features in each block and the sample size. In the high-dimensional setting where $\frac{p_m}{n_m}$ tends to infinity, Theorem \ref{thm:main2} can still be useful as long as the coordinate-wise separation distances are lower bounded by a constant; in this case, the separation distance $\Delta_m = \min_{k_1\neq k_2}\|(\theta^*_{k_1})_{T_m} - (\theta^*_{k_2})_{T_m}\|_2$ will also grow with the rate $\sqrt{p_m}$. On the other hand, when the centroid difference is sparse, a different clustering method that exploits such sparsity might be more appropriate. 
	\begin{remark}[Effect of dependency]
		We note that the effect of dependency is reflected by $\|A_c\|$ and $\|A_r\|$ in the separation condition and misclustering rate. Motivated by the neuronal functional data where firing activities are collected over time, let's consider the autoregressive process as an example. Suppose that the noise $E_{:,i}$ comes from an AR(1) process with coefficient $\rho\in (-1,1)$: $E_{k,i} = \rho E_{k-1,i} + \varepsilon_{k,i}$, where $\varepsilon_{k,i}$ and $E_{1,i}$ are independent $\sigma_{\varepsilon}$-sub-Gaussian, then we can let $A_c$ be the identity matrix, and $A_r$ be a lower triangular matrix with $(A_r)_{i,j}=\begin{cases}
		1,&i=j;\\\rho,&i=j+1
		\end{cases}.$
		In this setting, $\|A_c\|=1$, $\|A_r\|\leq 1+|\rho|$; a larger autoregressive coefficient $|\rho|$ leads to a stronger SNR condition, and a higher misclustering rate. 
	\end{remark}
 
	\subsection{Data-driven Selection of Patch Ordering}\label{sec:theory_datadriven}
	Recall that in our theoretical results in Section \ref{sec:theory}, the SNR conditions and final misclustering error hinge on one key factor: $\prod_{m=2}^M(1.1\gamma_m^{-1}(\pi)+1)$, where $\gamma_m(\pi) = \frac{\sigma_r(X^*_{T_{\pi(m)}, J_{m}^{(1)}(\pi)})}{\|X^*_{T_{\pi(m)},I_{\pi(m)}}\|}$ is the proportion of the overlapping signals in the $m$th patch. Let $\pi^* = \argmin_{\pi:[M]\rightarrow[M]}\prod_{m=2}^M(1.1\gamma_m^{-1}(\pi)+1)$ be the minimizer of this factor, i.e., the oracle patch ordering. Since each $\gamma_m(\pi)$ depends on the unknown ground truth matrix $X^*$, we propose to estimate it by its empirical counterpart $\hat{\gamma}_m(\pi) = \frac{\sigma_r(X_{T_{\pi(m)},J_m^{(1)}(\pi)})}{\|X_{T_{\pi(m)},I_{\pi(m)}}\|}$; and then find the ordering $\pi$ that minimizes $\prod_{m=2}^M(1.1\hat{\gamma}_m^{-1}(\pi)+1)$. This is equivalent to setting the score function in \eqref{eq:order_maximizer} as follows:
	\begin{equation}\label{eq:score}
	s(k, I) = \left(\frac{1.1\|X_{T_{k},I_{k}}\|}{\sigma_r(X_{T_{k},I_k\cap I})}+1\right)^{-1}.
	\end{equation}
	The following theorem gives the clustering guarantee with this data-driven selection of the patch ordering.
	\begin{theorem}\label{thm:datadriven}
		Suppose that the factor $\prod_{m=2}^M(1.1\gamma_m^{-1}+1)$ is substituted by $e^2\prod_{m=2}^M(1.1\gamma_m^{-1}(\pi^*) + 1)$ in Assumptions \ref{assump:block_eigmin}, \ref{assump:block_separation}, Theorems \ref{thm:main}, \ref{thm:main2}, and Corollary \ref{cor:main_delta}. Then Theorems \ref{thm:main}, \ref{thm:main2}, and Corollary \ref{cor:main_delta} still hold for Alg. \ref{alg:bsvdcq} applied on $\{X_{T_{\hat{\pi}(m)}, I_{\hat{\pi}(m)}}\}_{m=1}^M$, where the patch ordering $\hat{\pi}$ satisfies \eqref{eq:order_maximizer} with the score function \eqref{eq:score}. 
	\end{theorem}
	Theorem \ref{thm:datadriven} suggests that the data-driven selection of the patch ordering leads to almost the same performance as the oracle ordering; the slight constant factor inflation ($e^2$) is due to the finite sample approximation of the score function. The proof of Theorem \ref{thm:datadriven} is included in Section \ref{sec:datadrivenproof} of the Supplementary material, and builds on careful analysis of the accumulated finite sample errors across $M$ patches.
	\section{Simulation Studies} \label{sec:gsim}
	
	In this section, we study the performance of the proposed Cluster Quilting method in the mosaic patch data observation contexts using data simulated from a Gaussian mixture model. We compare our Cluster Quilting approach to several of the incomplete spectral clustering methods: IMSC-AGL \citep{wen2018incomplete}, IMG \citep{zhao2016incomplete}, DAIMC \citep{hu2019doubly}, and OPIMC \citep{hu2019onepass}. We also compare our procedure to the application of standard spectral clustering after imputing the missing portion of the data matrix using nuclear-norm regularized least squares \citep{chang2022lowrank}. The accuracy of each of the methods is quantified using the adjusted Rand index \citep{santos2009use} with respect to true underlying cluster labels in the generative models. In all simulation studies, we show results for oracle hyperparameter tuning. We tune other hyperparameters specific to each comparison method using a grid search and selection by best adjusted Rand Index score.

	\subsection{Data Generation Process} \label{sub:datagen}
	
	For the simulation studies below, we generate the underlying full data from Gaussian mixture models with low-rank cluster centroid matrices. We first create the true cluster labels by randomly grouping $n$ samples into $K$ clusters, encoded by matrix $F^*\in \bR^{n\times K}$ where 
$$F^*_{ik} = 
		\begin{cases} 
		1, & \text{sample } i \text{ belongs to cluster } k; \\
		0, & \text{otherwise}.
		\end{cases}$$ 

	To construct a low-rank cluster centroid matrix, we generate two rank $r$ matrices with $r < K$: a random orthogonal matrix $\boldsymbol{Z} \in \mathbb{R}^{p \times r}$ and a matrix $\boldsymbol{W} \in \mathbb{R}^{K\times r}$ whose entries are randomly selected from $\{-d,0,d\}$ for a constant $d$. For the latter matrix, we check that the rows are not linearly dependent and re-simulate new entries if this requirement is not met. The cluster centroid matrix $\Theta^* = (\theta_1,\cdots,\theta_K)^{\top}\in \bR^{K\times p}$ is formulated as the product of these two rank $r$ matrices, i.e. $\Theta^* = \boldsymbol{W}\boldsymbol{Z}^\top$. This generation ensures that $\Theta^*$ has rank $r$, with expected average separation distance $\mathbb{E}\|(\Theta^*)_{i,:} - (\Theta^*)_{j,:}\|_2^2=\frac{4}{3}rd^2$. Data are then simulated from a multivariate Gaussian distribution with mean matrix $B = F^*\Theta^* \in \mathbb{R}^{n \times p}$ and covariance matrix $I_p$.

	\subsection{Simulation Study} \label{sub:pbmvgmm}

	We now analyze the effectiveness of the Cluster Quilting procedure in the mosaic patch observation setting, where multi-view data is simulated. For this simulation study, instead of creating a single cluster centroid matrix as we do for the simulations in Section \ref{sub:bmgmm}, we generate separate cluster centroid matrices for each view in the data, applying the procedure outlined in Section \ref{sub:datagen} for each view. To create the missing observation pattern, we first randomly segment the rows of the full data matrix into $M$ blocks of observations. Following this, for each observation block, we randomly select $h$ views that are fully included in the set of observed entries while masking all other columns for the corresponding rows in the observation block. Cluster Quilting and the other incomplete spectral clustering methods are then applied to the masked data and compared based on how well they recover the true cluster labels. For this simulation study, we use default parameters of 840 observations partitioned into 3 underlying clusters, and 12 views of 50 features each for a total of 600 columns in the full data matrix. The cluster centroid matrix for each view is generated as a rank 2 matrix with $d = 7.5$. After data generation, we randomly divide the rows into 4 blocks of equal size. To create the mosaic patch observation structure, for each observation block, we randomly select 6 views to be part of the set of observed entries. We ensure in this step that the set of views selected for each observation block are unique, that each feature is part of the set of observed entries at least once, and that the graph of connected observation blocks with respect to overlapping views can be fully traversed. We analyze the performance of Clustering Quilting and other incomplete spectral clustering methods when varying the number of unmasked views per observation block, the number of observation blocks, the number of total features in the data, the number of underlying clusters, the distance between clusters, and the correlation between features.
	\begin{figure}[t]
		\begin{center}
			\begin{subfigure}[t]{0.37\linewidth}
				\centering
				\includegraphics[width=\linewidth]{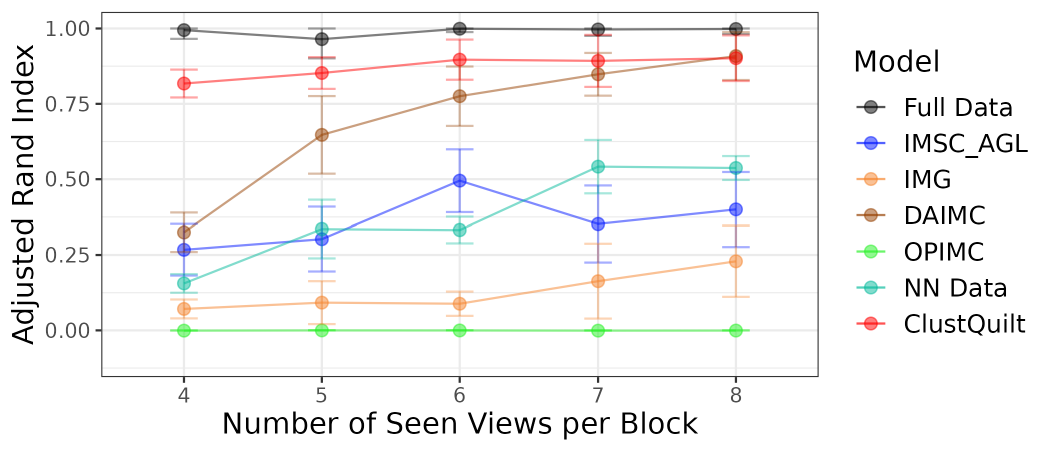}
				\caption{Views observed per block.}
				\label{fig:pmv_h}
			\end{subfigure} \hspace{0.5cm}%
			\begin{subfigure}[t]{0.37\linewidth}
				\centering
				\includegraphics[width=\linewidth]{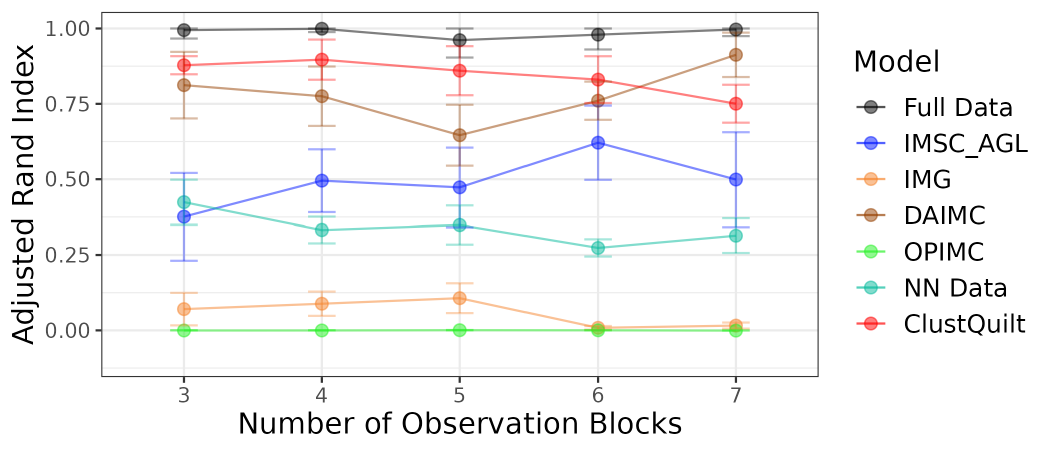}
				\caption{Number of blocks.}
				\label{fig:pmv_b}
			\end{subfigure}
			\begin{subfigure}[t]{0.37\linewidth}
				\centering
				\includegraphics[width=\linewidth]{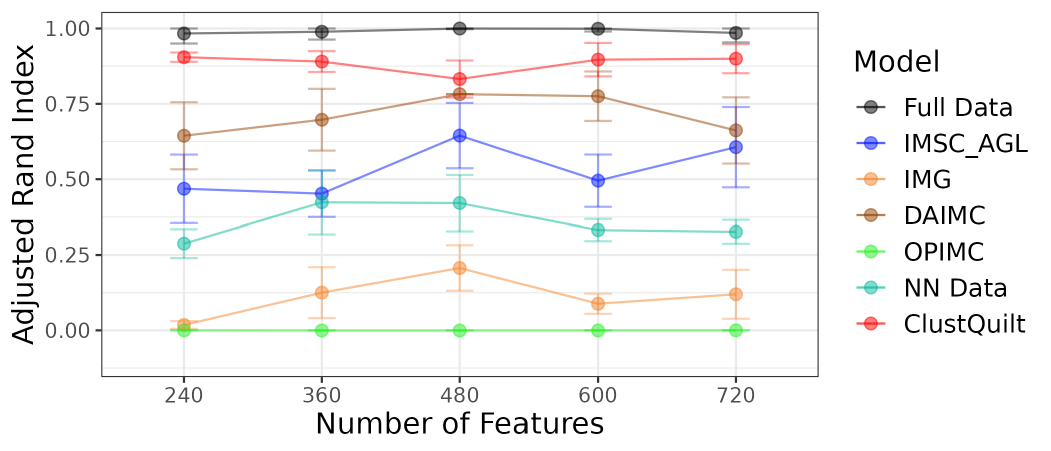}
				\caption{Number of total features.}
				\label{fig:pmv_p}
			\end{subfigure}\hspace{0.5cm}%
			\begin{subfigure}[t]{0.37\linewidth}
				\centering
				\includegraphics[width=\linewidth]{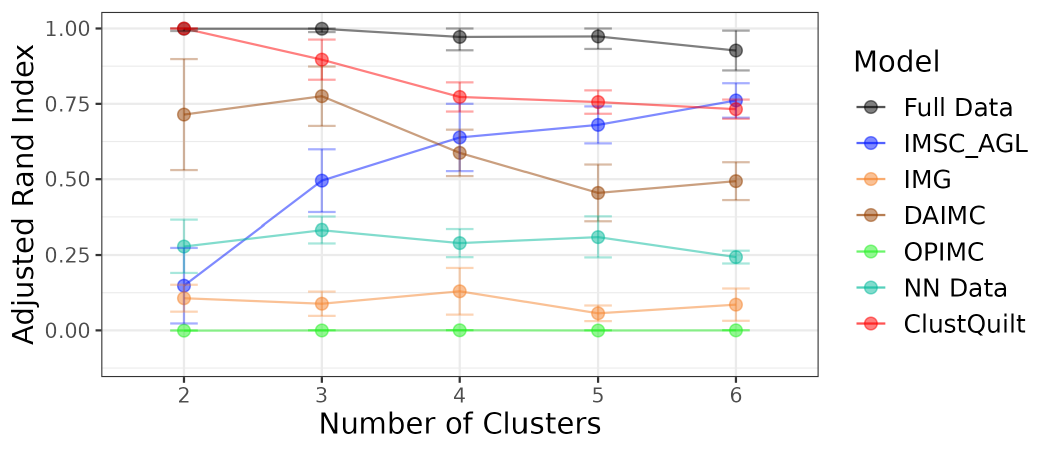}
				\caption{Number of true clusters.}
				\label{fig:pmv_k}
			\end{subfigure}
			\begin{subfigure}[t]{0.37\linewidth}
				\centering
				\includegraphics[width=\linewidth]{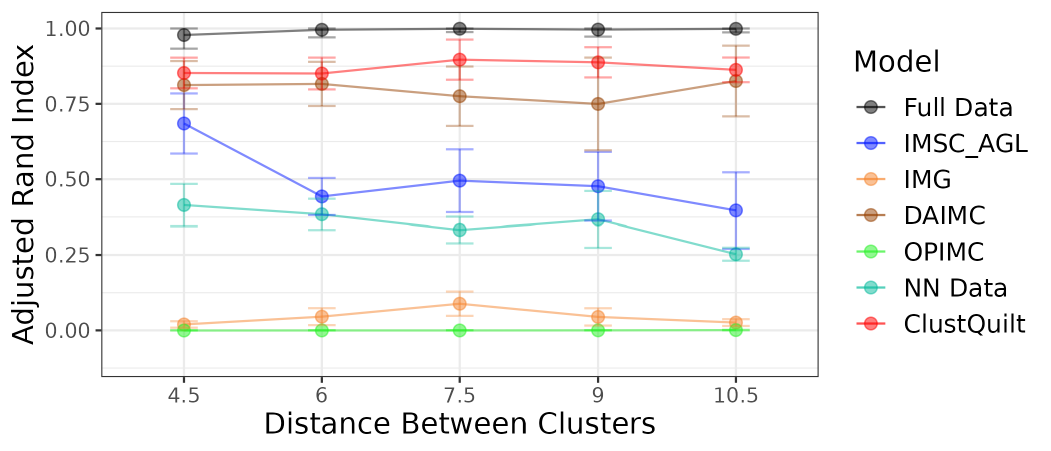}
				\caption{Cluster distance.}
				\label{fig:pmv_d}
			\end{subfigure}\hspace{0.5cm}%
			\begin{subfigure}[t]{0.37\linewidth}
				\centering
				\includegraphics[width=\linewidth]{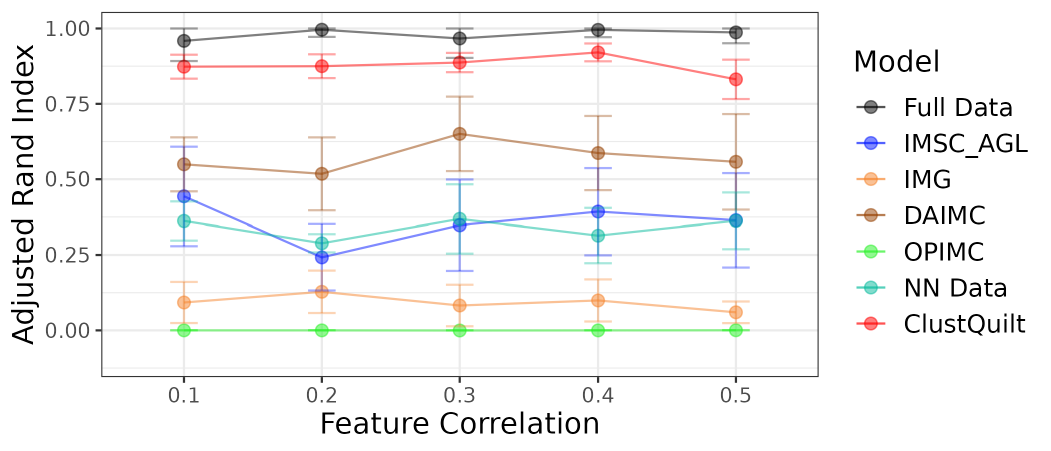}
				\caption{Feature autocorrelation.}
				\label{fig:pmv_rho}
			\end{subfigure}
		\end{center}
		\caption{Performance of Cluster Quilting and comparison cluster imputation methods with oracle tuning on Gaussian mixture model data in the mosaic patch setting for different simulation parameters, measured by adjusted Rand index.}
		\label{fig:pmv}
	\end{figure}
	We show the results of the mosaic patch observation simulation study in Figure \ref{fig:pmv}. The Cluster Quilting method outperforms the comparative incomplete spectral clustering methods in the vast majority of the simulation settings; we see that Cluster Quilting performs particularly better when the number of views seen per block (Figure \ref{fig:pmv_h}) or the number of features (Figure \ref{fig:pmv_p}) is relatively small. Additionally, the performance of the Cluster Quilting is better when the number of observation blocks is smaller but tends to degrade more than other methods with an increase in the number of blocks (Figure \ref{fig:pmv_b} ), which matches with the theoretical results of Section \ref{sec:theory}. We see that an increase in the number of true underlying clusters tends to decrease the performance of all methods consistently except for the IMSC-AGL algorithm (Figure \ref{fig:pmv_k}), while the performance of methods is relatively consistent with changing the distance between clusters (Figure \ref{fig:pmv_d}). Also, in Figure \ref{fig:pmv_rho}, we find that increasing the correlation between features increases the accuracy of cluster estimates.

	In the Supporting Information, we include several additional simulation studies. We first show results for a parallel empirical study as above in the sequential patch observation setting. For both mosaic and sequential patch observations, show results for data-driven hyperparameter tuning, for Gaussian mixture model data in the high-dimensional case, as well as when the data are generated from non-Gaussian copula mixture models. The results of these additional studies paint a similar picture, i.e. that the Cluster Quilting method performs similarly well, relative to the comparative incomplete spectral clustering methods, for recovering cluster labels.

	\section{Multiomics Case Study} \label{sec:real}

	Below, we investigate the performance of the Cluster Quilting method in the mosaic patch setting for identifying cancer types from multi-omics data. In this case studies, we take complete data sets and create synthetic patch missingness in the data matrix. As in Section \ref{sec:gsim}, we compare our Cluster Quilting approach to several other incomplete spectral clustering methods, namely IMSC-AGL \citep{wen2018incomplete}, IMG \citep{zhao2016incomplete}, DAIMC \citep{hu2019doubly}, and OPIMC \citep{hu2019onepass}. The data analyzed below come from The Cancer Genome Atlas (TCGA) \citep{weinstein2013cancer}, which contains genetic, transcriptomic, and proteomic data for a wide-ranging sample of tumors. For this case study, we amalgamate tumor data from 3 different primary sites, namely kidney, breast, and lung, and we consider features collected on each tumor from 5 different views: DNA methylation, miRNA expression, RNAseq gene expression, and reverse phase protein array. In accordance with the procedures in \cite{lock2013bayesian} and \cite{wang2021integrative}, we preprocess the data set before analysis by performing a square root transformation of the DNA methylation features and a log transformation of the miRNA features, as well as filtering out columns in the gene expression data with standard deviation less than 1.5. Additionally, we scale the data column-wise before creating the partially observed data set.
	
	\begin{figure}[t]
		\begin{center}
			\begin{subfigure}[t]{0.49\linewidth}
				\includegraphics[width=\linewidth]{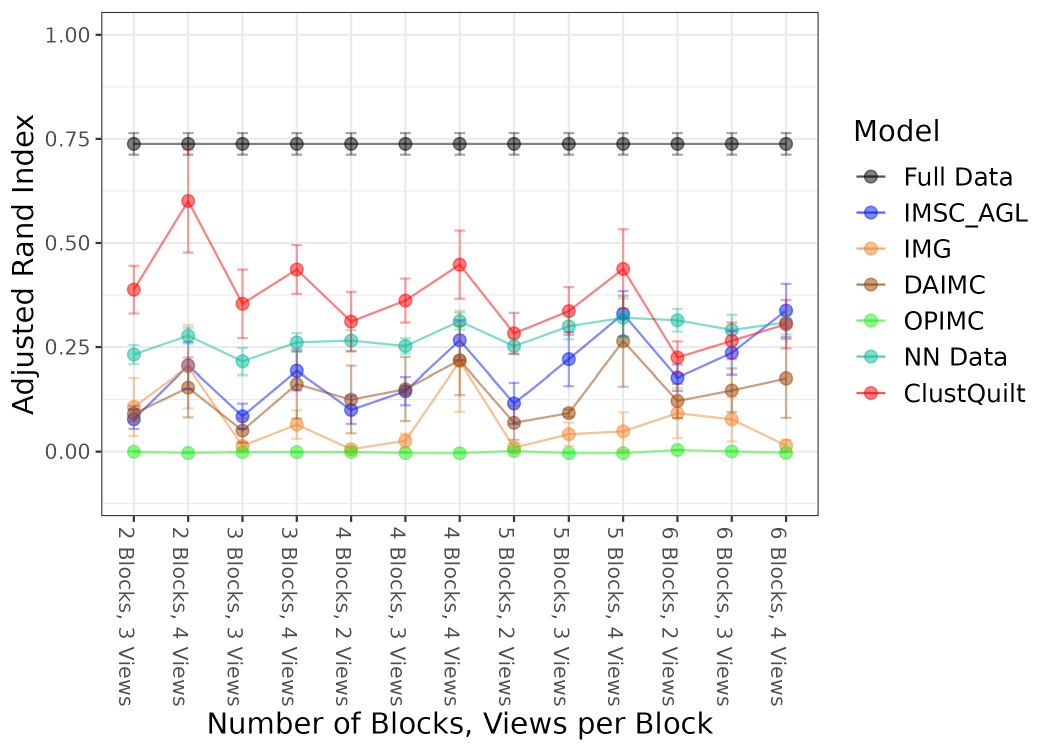}
				\caption{Oracle tuning.}
				\label{fig:tcga_o}
			\end{subfigure}%
			\begin{subfigure}[t]{0.49\linewidth}
				\includegraphics[width=\linewidth]{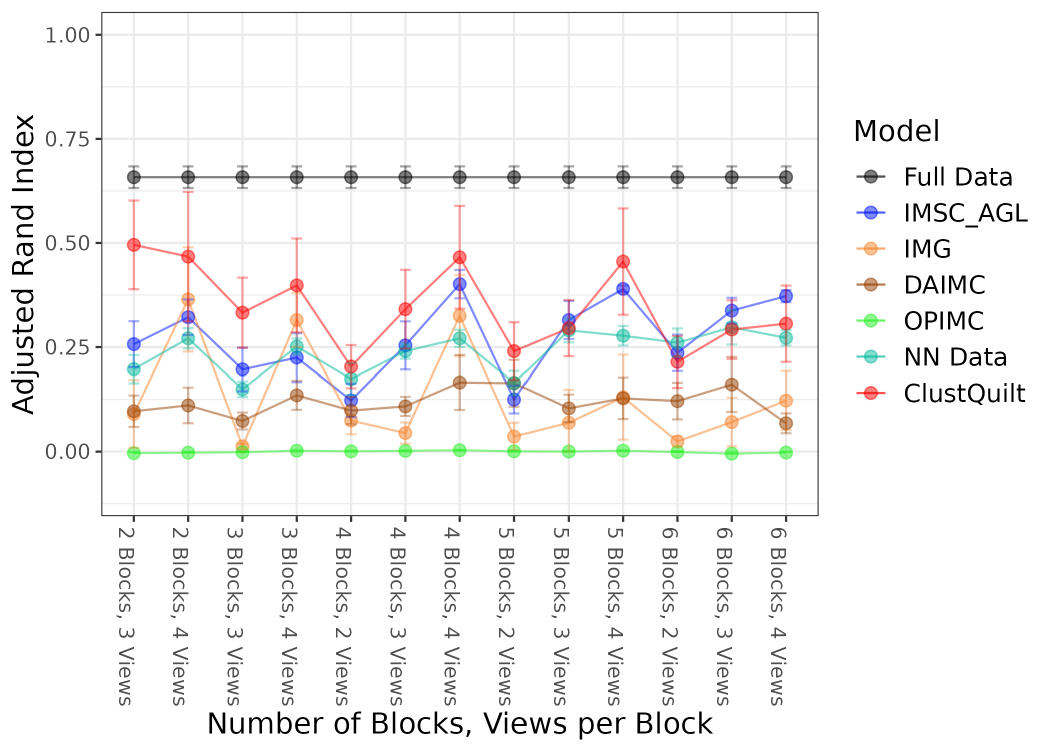}
				\caption{Data driven tuning.}
				\label{fig:tcga_dd}
			\end{subfigure}
		\end{center}
		\caption{Adjusted Rand index of cluster estimates from each incomplete spectral clustering methods with respect to true primary sites of tumors in the TCGA dataset. Results are shown across varying number of observation blocks and number of views unmasked per block.}
		\label{fig:tcga_full}
	\end{figure}

	From the complete data set, we create synthetic missing data with a mosaic patch observation pattern by randomly partitioning the rows of the data into distinct observation blocks, then selecting a subset of views for each block to include in the set of observed entries. We then apply Cluster Quilting and the comparison incomplete spectral clustering methods to the synthetic partially observed data set and measure the performance of each method using adjusted Rand index with respect to the true primary site labels. The efficacy of the Cluster Quilting and other incomplete spectral clustering methods are shown across varying numbers of observation blocks and views seen per block; for each combination of the two, we perform 50 replications and report the mean and standard deviation of the adjusted Rand index for each method. Additionally, we present results for both oracle and data-driven tuning of the number of clusters and rank of the data matrix; for the latter, we use the prediction validation method proposed in \cite{tibshirani2005cluster}. Here, we set the oracle rank to be the rank selected by data-driven tuning on spectral clustering when applied to the fully observed data.

	The results of our study are presented in Figure \ref{fig:tcga_full}, with oracle hyperparameter tuning shown in Figure \ref{fig:tcga_o} and data-driven tuning in Figure \ref{fig:tcga_dd}. Across the different numbers of views observed per patch, we observe that the Cluster Quilting algorithm does best when the number of observation blocks is relatively small, while several methods, namely the likelihood imputation and IMSC-AGL methods, outperform Cluster Quilting, when the number of observation patches in the data is larger. The degredation in accuracy of the Cluster Quilting method aligns with the theoretical results presented in Section \ref{sec:theory}. However, even when the number of blocks is larger, the Cluster Quilting method still achieves close to the highest accuracy out of all methods. 
	
	\begin{figure}[tbhp]
		\begin{center}
			\includegraphics[width=0.55\linewidth]{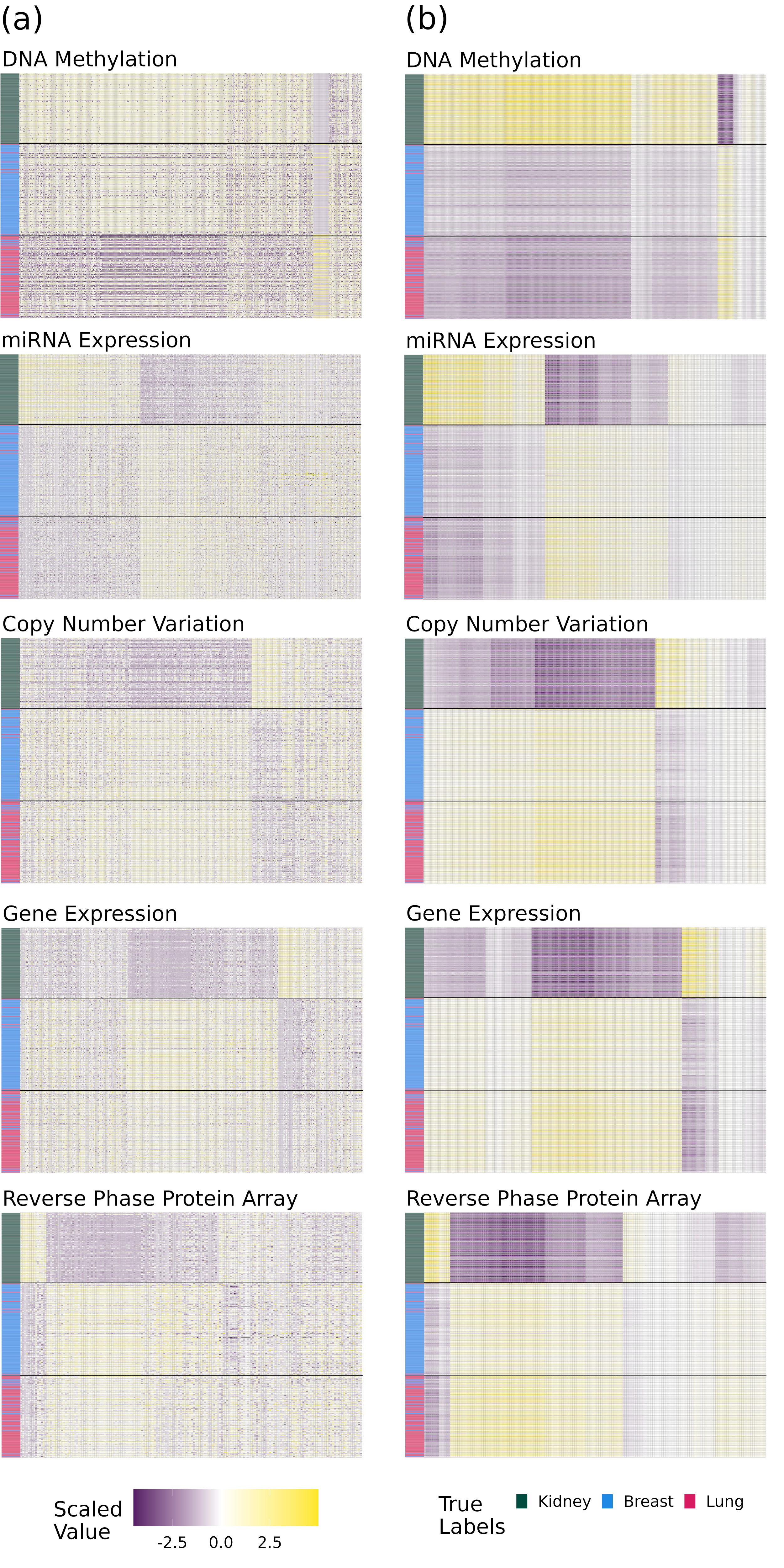}
		\end{center}
		\caption{Cluster heatmap for each individual view for the \textbf{(a)} true data matrix and \textbf{(b)} imputed low-rank data matrix estimated from Cluster Quilting. Cells are colored by scaled value of the matrix. Rows are ordered by true primary site, with different underlying groupings partitioned; color labels on the left bar show the estimated cluster labels returned by the Cluster Quilting procedure. Columns are ordered by similarities from hierarchical clustering.}
		\label{fig:tcga_heat}
	\end{figure}
	
	We further investigate the cluster label estimates from the Cluster Quilting procedure in Figure \ref{fig:tcga_heat}. Here, we show cluster heatmaps for each view included in our TCGA empirical study, with the heatmap of the fully observed data set in column (a) and the heatmap of the data matrix as imputed 
	via the Cluster Quilting algorithm ($\widehat{U}\widehat{\Lambda}\widehat{V}^\top$ in Algorithm \ref{alg:bsvdcq}) for the corresponding views in column (b). The rows in the heatmap are sorted by the estimated cluster labels returned by Cluster Quilting, while the true primary sites of each tumor are shown in the colored bar on the left side of the figure. The Cluster Quilting method is able to correctly cluster together all of the tumors with kidney as the primary site, but misclassifies several of the tumors from the breast and lung primary sites as the other. This is reflected in the cluster heatmaps for both the true data matrix and the imputed data matrix estimated via the Cluster Quilting algorithm, in which we see that the transcriptomic and genetic information in the kidney primary site tumors are more easily distinguishable compared those from the other two primary sites. The misclassification between the lung and breast cancer tumors can likely be explained by the relative similarity between the two, which has been shown in previous literature \citep{gemma2001genomic, schwartz1999familial}; in particular, metastasis between the lung and breast have been studied in great depth \citep{montel2005expression, jin2018breast, landemaine2008six}. In the Supplementary Material, we perform a real-world data study in the sequential patch observation setting, in which we analyze functional neuronal clustering from calcium imaging data.

	\section{Discussion} \label{s:discuss}
	In this work, we study the problem of clustering in the patchwork learning setting. We propose the Cluster Quilting procedure, the first spectral clustering method on patchwork data with statistical guarantees. We have shown finite sample misclustering rates based on the separation distance between clusters and signal-to-noise ratio in the data patches, and we have demonstrated through empirical studies that the Cluster Quilting method outperforms other incomplete spectral clustering methods for recovering true underlying cluster labels in both simulated and real-world data. While other works in the literature have studied the general problem of cluster estimation under the presence of missing data, including many techniques proposed for partial multi-view spectral clustering, we are the first to propose a method and study its performance from a statistical perspective for the patchwork learning setting. Additionally, our work significantly differs from while also complements the matrix completion literature with deterministic missing patterns: prior works in this area mostly consider a more specific observational pattern, e.g., causal panel data, sequential observational blocks, etc., and have no clear implications in the clustering context. 
	
	While the Cluster Quilting method has shown promising results for clustering in the patchwork learning setting, there are several open problems and potential extensions which require further investigation. Our current empirical studies have focused on the case where each cluster is well-represented in each observation patch and where the amount of overlap between patches is relatively evenly distributed and sufficiently large. Violations of these assumptions, however, may occur in real-world applications, and it is unclear how well Cluster Quilting will perform in these situations. We have also restricted our work to the case where patch merging is done sequentially, i.e. where we only consider iteratively merging a new patch to ones that have already been previously consolidated. This does not have to be the case algorithmically, and in fact it may be beneficial to merge non-sequentially in order to maximize the total overlapping signal; however, this scheme presents extra challenges for quantifying theoretical performance. From a broader perspective, ideas from the Cluster Quilting algorithm could potentially be extended for clustering patchwork-observed data in the context of multimodal data integration; however, further work will be required to study whether our method can be used to meet challenges specific to this setting, such as mixed data types or for highly erosely measured modalities. Beyond data clustering, our methods and ideas may also be leveraged for other problems in patchwork learning, such as PCA, ranking, and community detection in networks.

	\section*{Acknowledgements}
	The authors gratefully acknowledge support from NSF NeuroNex-1707400, NIH 1R01GM140468, and NSF DMS-2210837.
 \newpage
	\appendix
 \section{Proofs}
In this section, we present the complete proofs for all our theoretical results.
\begin{proof}[Proof of Theorem \ref{thm:main}]
We first present a lemma that shows that spectral error bound for $X^*$ implies clustering consistency:
\begin{lem}\label{lem:spectral_cluster_err}
    Suppose that $\|\widetilde{H}\widetilde{V}^\top - X^*\|\leq \varepsilon$. Let $\Phi$ denote the set of all permutations on $[n]$. Then as long as $\varepsilon<\sqrt{\frac{\min_{1\leq j\leq K} n_j}{64r}}\Delta$, we have
    $$
    \ell(\widehat{z},z)\leq \frac{32r\varepsilon^2}{n\Delta^2},\quad \min_{\phi\in \Phi}\max_{1\leq j\leq K}\|\widehat{U}\widehat{c}_j - \theta^*_{\phi(j)}\|_2^2 \leq \frac{16r\varepsilon^2}{\min_{1\leq j\leq K} n_j}.
    $$
\end{lem}
The proof of Lemma \ref{lem:spectral_cluster_err} follows similar ideas in \cite{loffler2021optimality}. The main difficulty of the full proof of Theorem \ref{thm:main} lies in showing the spectral error bound for the output of Algorithm \ref{alg:bsvdcq} so that we can apply Lemma \ref{lem:spectral_cluster_err}. In the following, we present a theorem that bounds $\|\widetilde{H}\widetilde{V}^\top - X^*\|$ with high probability: 
\begin{thm}\label{thm:spectral_err}
    Consider the model set-up described in Section \ref{sec:model} and Algorithm \ref{alg:bsvdcq}. If Assumption \ref{assump:block_eigmin} holds, we have%with probability at least $1-C\sum_{m=1}^M\exp\{-cn_m\}$ for some universal constants $c,\,C>0$, we have
    $$
    \|\widetilde{H}\widetilde{V} - X^*\|\leq 
    \frac{C\alpha_{\max}(\beta_{\max}+1)}{\beta_{\min}^{2}}\prod_{l=2}^M\left(\frac{1.1}{\gamma_l}+1\right)\max_m\frac{\|E_{T_m,I_m}\|}{\sigma_r(X^*_{T_m,I_m})}\|X^*\|.
    $$
\end{thm}
\noindent Combining Theorem \ref{thm:spectral_err} and Assumption \ref{assump:block_eigmin}, we have
$$
    \|\widetilde{H}\widetilde{V} - X^*\|\leq C\|X^*\|\sqrt{\frac{\min_j n_j}{rK \max_j n_j}}.
    $$
By Assumption \ref{assump:centroid_norm_bnd}, $\|X^*\|\leq \|X^*\|_F\leq \max_j\|\theta_j\|_2\sqrt{n}\leq C\Delta\sqrt{K\max_jn_j}$, and hence by choosing the constant $C$ appropriately in Assumptions \ref{assump:block_eigmin}, we have $\|\widetilde{H}\widetilde{V} - X^*\|<\sqrt{\frac{\min_j n_j}{64r}}\Delta$. Now we can apply Lemma \ref{lem:spectral_cluster_err}, plug in the error bound in Theorem \ref{thm:spectral_err} into Lemma \ref{lem:spectral_cluster_err} to complete the proof of Theorem \ref{thm:main}.
\end{proof}
The following lemma for bounding the sum of sequential products will be useful in many steps of our main proof.
\begin{lem}\label{lem:sum_prod}
    For any sequence $\{a_1,a_2,a_n\}$, if $a_l\geq a_{\min}$ for all $1\leq l\leq n$ and some $a_{\min} >1$, then we have
    \begin{equation*}
        \sum_{i=1}^n\prod_{l=1}^i a_l \leq \frac{\prod_{l=1}^na_la_{\min} - a_1}{a_{\min} - 1}.
    \end{equation*}
\end{lem}
\begin{proof}[Proof of Lemma \ref{lem:sum_prod}]
    Let $B = \sum_{i=1}^n\prod_{l=1}^i a_l$. We can then write
    \begin{equation*}
        \begin{split}
            Ba_{\min} =& \sum_{i=1}^n\prod_{l=1}^i a_la_{\min}\\
            \leq&\sum_{i=1}^{n-1}\prod_{l=1}^{i+1} a_l + \prod_{l=1}^n a_l a_{\min}\\
            = &\sum_{i=1}^{n}\prod_{l=1}^{i} a_l + \prod_{l=1}^n a_l a_{\min} - a_1\\
            = &B + \prod_{l=1}^n a_l a_{\min} - a_1.
        \end{split}
    \end{equation*}
    Therefore, $B(a_{\min} - 1) \leq \prod_{l=1}^n a_l a_{\min} - a_1$. The proof is now complete.
\end{proof}

\subsection{Proof of Theorem \ref{thm:spectral_err}}
The proof consists of three main steps: (i) we first upper bound the spectral errors for the SVD of each block; (ii) then we show that the matching matrices for consecutive blocks are accurately estimated, given that each block has small perturbation bound shown in the first step; (iii) combining the first two results together, we obtain final error bounds for the full matrix factors $\til{V}\in \bbR^{n\times r}$ and $\til{H}\in \bbR^{p\times r}$. 

\paragraph{Bounding spectral errors for each block}
Consider the block-wise SVD for the population quantity $X^*_{T_m,I_m}=U_m^*\Lambda_m^*V_m^{*\top}$ with $U_m^*\in \bbO^{p_m\times r}$, $\Lambda_m\in \bbR^{r\times r}$, and $V_m\in \bbO^{n_m\times r}$. Then we would like to show spectral error bounds for $\widehat{H}_m$ and $\h{V}_m$ as estimates for $U_m^*\Lambda_m^*$ and $V_m^*$ up to rotation. By Weyl's inequality, $\sigma_{r+1}(X_{T_m,I_m}) \leq \|E_{T_m,I_m}\|$. As long the universal constant $C>0$ in Assumption \ref{assump:block_eigmin} is taken to be larger than $2$, we have $\sigma_{r+1}(X_{T_m,I_m}) \leq \|E_{T_m,I_m}\|\leq \frac{1}{2}\sigma_r(X^*_{T_m,I_m})$. Applying the Davis-Kahan's theorem \cite[see, e.g., Theorem 2.7 in][]{chen2021spectral}, we know that 
\begin{equation}\label{eq:block_V_err}
\begin{split}
    \min_{W\in \cO^{r\times r}}\|\h{V}_mW - V_m^*\|\leq &\frac{\sqrt{2}\|E_{T_m,I_m}\|}{\sigma_r(X^*_{T_m,I_m}) - \sigma_{r+1}(X_{T_m,I_m})}\\
    \leq &\frac{2\sqrt{2}\|E_{T_m,I_m}\|}{\sigma_r(X^*_{T_m,I_m})}.
\end{split}
\end{equation}
For the rest of the proof, we use $\widehat{W}_m$ to denote the best rotation matrix that achieves the minimum $\min_{W\in \mathcal{O}^{r\times r}} \|\widehat{V}_mW - V_m^*\|$. Now we turn to upper bounding $\widehat{H}_m\widehat{W}_m - U_m^*\Lambda_m^*$. Recognizing the fact that $\widehat{H}_m\widehat{W}_m = X_{T_m,I_m}\widehat{V}_m\widehat{W}_m$ and $U_m^*\Lambda_m^* = X^*_{T_m,I_m}V_m^*$, we can show the following:
\begin{equation}\label{eq:H_m_err}
\begin{split}
    \|\widehat{H}_m\widehat{W}_m - U_m^*\Lambda_m^*\|= &\|X_{T_m,I_m}\widehat{V}_m\widehat{W}_m - X^*_{T_m,I_m}V_m^*\|\\
    \leq&\|X^*_{T_m,I_m}\|\|\widehat{V}_m\widehat{W}_m-V_m^*\| + \|E_{T_m,I_m}\|\\
    \leq&4\kappa_m\|E_{T_m,I_m}\|
\end{split}
\end{equation}
where the second line utilizes \eqref{eq:block_V_err}. 
We have also used $\kappa_m = \frac{\sigma_1(X^*_{T_m,I_m})}{\sigma_r(X^*_{T_m,I_m})}$ to denote the condition number of the ground truth matrix in block $m$.

\paragraph{Oracle matching matrices for consecutive blocks} Now we would like to establish a correspondence between our block-wise SVD results and the sub-matrices of the full SVD. By establishing such a correspondence, we can also find the oracle matching matrices between each block and its prior blocks. Consider the SVD for the full ground truth matrix $X^*=U^*\Lambda^*V^{*\top}$. Then for each block $m$, note that $U_{T_m,:}^*\Lambda^*V_{I_m,:}^{*\top} = U_m^*\Lambda_m^*V_m^{*\top}$. Hence we define $B_m^* = \Lambda^* U_{T_m,:}^{*\top}U_m^*\Lambda_m^{*-1}\in \mathbb{R}^{r\times r}$, and one can show that $B_m^*$ can transform the $m$th block of the full SVD matrix $V^*_{I_m,:}$ to $V_m^*$:
    \begin{equation}\label{eq:V_B}
        V_m^* = V^*_{I_m,:}B_m^*.
    \end{equation}
    We can also see that 
    \begin{equation}\label{eq:U_B}
        U_m^*\Lambda_m^*B_m^{*\top} = U_m^*U_m^{*\top}U_{T_m,:}^*\Lambda^* = U_{T_m,:}^*\Lambda^*,
    \end{equation} since $U_{T_m,:}^*$ lies within the space spanned by $U_m^*$. 
    Let $\widetilde{V}^* = V^*B_1^{*}\widehat{W}_1^\top$, and $\widetilde{H}^* = U^*\Lambda^* (B_1^{*-1})^\top\widehat{W}_1^\top$, then we also have another decomposition for the ground truth matrix $X^* = \widetilde{H}^*\widetilde{V}^{*\top}$, and they satisfy:
    \begin{equation*}
        (\til{V}^*)_{I_1,:} = V_1^*\h{W}_1^\top, \quad (\til{H})^*_{T_1,:} = U_1^*\Lambda_1^*\h{W}_1^\top.
    \end{equation*}
    Therefore, our perturbation bounds for blockwise SVD in the first step already suggest $\til{V}_{I_1,:} = \h{V}_1$ and $\til{H}_{T_1,:} = \h{H}_1$ as good estimates for $(\til{V}^*)_{I_1,:}$ and $(\til{H})^*_{T_1,:}$. On the other hand, for block $m>1$, some transformation is needed to match the first block so that they can estimate $(\til{V}^*)_{I_m,:}$ and $(\til{H}^*)_{T_m,:}$ well. Define 
    \begin{equation}\label{eq:G_m_def}
        G_m^* := \h{W}_mB_m^{*-1}B_1^*\h{W}_1^\top,
    \end{equation} 
    then some computations show that
    \begin{equation*}
        (\til{V}^*)_{I_m,:} = V_m^*\h{W}_m^\top G_m^* ,\quad (\til{H})^*_{T_m,:} = U_m^*\Lambda_m^*\h{W}_m^\top(G_m^*)^{\top -1}.
    \end{equation*}
    We will need to show that the oracle $G_m^*$ can be estimated well by Algorithm \ref{alg:bsvdcq} later in the proof, which then suggests all blocks of $\widetilde{V}$ and $\til{H}$ are accurate estimates of $\til{V}^*$ and $\til{H}^*$.
    
  Furthermore, the spectral error for the full low-rank reconstruction satisfies the following:
    \begin{equation}
        \begin{split}
            \|\widetilde{H}\widetilde{V}^\top-X^*\| = &\|\widetilde{H}\widetilde{V}^\top-\widetilde{H}^*\widetilde{V}^{*\top}\|\\
            \leq &\|\widetilde{H}\|\|\widetilde{V}-\widetilde{V}^{*}\|+\|\widetilde{H}-\widetilde{H}^*\|\|\widetilde{V}^{*}\|\\
            \leq &(\|\widetilde{H}^*\|+\|\widetilde{H}-\widetilde{H}^*\|)\|\widetilde{V}-\widetilde{V}^{*})\|+\|\widetilde{H}-\widetilde{H}^*\|\|\widetilde{V}^{*}\|.
        \end{split}
    \end{equation} 
    For each block $m$, let $J_m^{(1)} = (\bigcup_{k = 1}^{m-1} I_k )\cap I_m, \, J_m^{(2)} = \{j: I_m[j] \in J_m^{(1)}\}$ be the overlapping set of block $m$ with prior blocks. When $m=1$, $J_m^{(1)}$ and $J_m^{(2)}$ are trivially defined as empty sets. By the definition of the spectral norm, we know that the spectral error for both $\til{V}$ and $\til{H}$ can be decomposed as follows:
    \begin{equation}\label{eq:V_H_fullerr}
    \begin{split}
        \|\widetilde{V}-\widetilde{V}^*\| \leq &\sqrt{\sum_{m=1}^M\|\widetilde{V}_{I_m\backslash J_m^{(1)},:}-\widetilde{V}^*_{I_m\backslash J_m^{(1)},:}\|^2},\\
        \|\widetilde{H}-\widetilde{H}^*\| \leq &\sqrt{\sum_{m=1}^M\|\widetilde{H}_{T_m,:}-\widetilde{H}^*_{T_m,:}\|^2}.
    \end{split}
    \end{equation}
    We will now bound the spectral norm error $\|\widetilde{V}_{I_m\backslash J_m^{(1)},:}-\widetilde{V}^*_{I_m\backslash J_m^{(1)},:}\|$ and $\|\widetilde{H}_{T_m,:}-\widetilde{H}^*_{T_m,:}\|$ for each block $m$ in the following step.

\paragraph{Bounding the matching-induced errors} For the first block, by construction, $\|\widetilde{V}_{I_1,:}-\widetilde{V}^*_{I_1,:}\|=\|\widehat{V}_1\widehat{W}_1-V_1^*\|$, $\|\widetilde{H}_{T_1,:}-\widetilde{H}^*_{T_1,:}\|=\|\widehat{H}_1\widehat{W}_1 -U_1^*\Lambda_1^*\|$.
For the $m$th block, let $\varepsilon_m^{(2)} := \big\|\widetilde{V}_{I_m\backslash J_m^{(1)},:}-\widetilde{V}^*_{I_m\backslash J_m^{(1)},:}\big\|$. Then one can show that
    \begin{equation}\label{eq:induction1}
    \begin{split}
        \varepsilon_m^{(2)} \leq &\|\widehat{V}_mG_m-\widetilde{V}^*_{I_m,:}\|\\
        =&\|\widehat{V}_mG_m -V_m^*B_m^*B_1^{*-1}\widehat{W}_1^\top\|\\
        =&\|\widehat{V}_mG_m -V_m^*\widehat{W}_m^\top G_m^*\|\\
        \leq&\|G_m-G_m^*\|+\|\widehat{V}_m\widehat{W}_m -V_m^*\| \|G_m^*\|\\
        =:&\varepsilon_m^{(3)} +\varepsilon_m^{(1)}\|G_m^*\|,
        \end{split}
    \end{equation} 
    where $G_m^*$ was defined earlier in \eqref{eq:G_m_def} as the oracle matching matrix that aligns block $m$ with the first block, and $\varepsilon_m^{(3)} := \|G_m-G_m^*\|$, $\varepsilon_m^{(1)} := \|\widehat{V}_m\widehat{W}_m -V_m^*\|$. %The key next step is to bound $\|G_m-G_m^*\|$.

While for $\|\widetilde{H}_{T_m,:}-\widetilde{H}^*_{T_m,:}\|$, we have 
    \begin{equation}\label{eq:H_block_err}
    \begin{split}
        \|\widetilde{H}_{T_m,:}-\widetilde{H}^*_{T_m,:}\|&=\|\widehat{H}_mG_m^{\top -1}-U_m^*\Lambda_m^*B_m^{*\top}(B_1^{*-1})^{\top}\widehat{W}_1^\top\|_2\\
        &=\|\widehat{H}_mG_m^{\top -1}-U_m^*\Lambda_m^*\widehat{W}_m^\top (G_m^{*\top})^{-1}\|\\
        &\leq\|X_{T_m,I_m}^*\|\|G_m^{-1}-G_m^{*-1}\|\\ 
        &\quad+ \|\widehat{H}_m\widehat{W}_m-U_m^*\Lambda_m^*\|(\|G_m^{*-1}\|+\|G_m^{-1}-G_m^{*-1}\|).
    \end{split}  
    \end{equation} 
    This would require an error bound for $\|G_m^{-1}-G_m^{*-1}\|$.

In the following, we first show bounds for $G_m-G_m^*$, and then use them to derive error bounds for $\widetilde{V}_{I_m\backslash J_m^{(1)},:}-\widetilde{V}^*_{I_m\backslash J_m^{(1)},:}$ and $\widetilde{H}_{T_m,:}-\widetilde{H}^*_{T_m,:}$. First we note that $G_m^*$ satisfies
    \begin{equation*}
    \begin{split}
        (V_m^*)_{J_m^{(2)},:}\widehat{W}_m^\top G_m^* &= (V_m^*)_{J_m^{(2)},:}B_m^*B_1^{*-1}\widehat{W}_1^\top\\
        &=V^*_{J_m^{(1)},:}B_1^{*-1}\widehat{W}_1^\top\\
        &=\widetilde{V}^*_{J_m^{(1)},:},
    \end{split}
    \end{equation*}
    while $\h{G}_m$ is the OLS solution for transforming $(\h{V}_m)_{J_m^{(2)},:}$ to $\til{V}_{J_m^{(1)},:}$.
    The following lemma shows that the perturbation bound of the OLS solution:
    \begin{lem}\label{lem:OLS_err}
        Suppose $Y^* = Z^*G^*$ for $Z^*\in \mathbb{R}^{p\times r}$, $G^*\in \mathbb{R}^{r\times r}$, and $G = (Z^\top Z)^{\dagger} Z^\top Y$ with $Z$ and $Y$ of the same dimension as $Z^*$ and $Y^*$. If $\|Z- Z^*\|\leq \varepsilon_1$, $\|Y-Y^*\|\leq \varepsilon_2$ and $\sigma_r(Z^*)>\varepsilon_1$, then we have $$\|G-G^*\|\leq \frac{\|G^*\|\varepsilon_1+\varepsilon_2}{\sigma_r(Z^*) - \varepsilon_1}.$$
    \end{lem}
    Now we apply Lemma \ref{lem:OLS_err} with $Z^* = (V_m^*)_{J_m^{(2)},:}\widehat{W}_m^\top$, $G^* = G_m^*$, $Y^* = \widetilde{V}^*_{J_m^{(1)},:}$; $Z = (\widehat{V}_m)_{J_m^{(2),:}}$, $Y = \widetilde{V}_{J_m^{(1)},:}$. To verify the condition in Lemma \ref{lem:OLS_err}, we note that 
    \begin{equation*}
    \begin{split}
        \sigma_r((V_m^*)_{J_m^{(2)},:}\widehat{W}_m^\top) &= \sigma_r((V_m^*)_{J_m^{(1)}})\geq \frac{\sigma_r(X^*_{T_m,J_m^{(1)}})}{\|U_m^*\Lambda_m^*\|}\geq \frac{\sigma_r(X^*_{T_m,J_m^{(1)}})}{\|X^*_{T_m, I_m}\|}= \gamma_m,
    \end{split}
    \end{equation*} 
    where the last equation is due to the definition of $\gamma_m$ in Section \ref{sec:theory}.
    Meanwhile, due to our previous error bound \eqref{eq:block_V_err} for each block, the first error term in Lemma \ref{lem:OLS_err} satisfies 
    \begin{equation*}
        \begin{split}
            \varepsilon_1 = &\|(\widehat{V}_m)_{J_m^{(2),:}}-(V_m^*)_{J_m^{(2)},:}\widehat{W}_m^\top\|\\
            \leq &\|\widehat{V}_m\widehat{W}_m - V_m^*\|\\
            \leq &\frac{2\sqrt{2}\|E_{T_m,I_m}\|}{\sigma_r(X^*_{T_m,I_m})}.
        \end{split}
    \end{equation*}
    Since the number of blocks $M\geq 2$, as long as the universal constant $C$ in Assumption \ref{assump:block_eigmin} is chosen to satisfy $C\geq 20\sqrt{2}$, we have $\varepsilon_1\leq \frac{2\sqrt{2}\|E_{T_m,I_m}\|}{\sigma_r(X^*_{T_m,I_m})}\leq \frac{\gamma}{11}$,
which implies
$$
\gamma_m \geq \gamma\geq 11\varepsilon_1.
$$
Then, invoking Lemma \ref{lem:OLS_err}, we obtain the following:
\begin{equation}\label{eq:induction2}
    \begin{split}
        \varepsilon_m^{(3)}&=\|G_m-G_m^*\|\\
        &\leq \frac{1.1}{\gamma_m}(\|G_m^*\|\|\h{V}_m\h{W}_m - V_m^*\|+\|\widetilde{V}_{J_m^{(1)},:}-\widetilde{V}^*_{J_m^{(1)},:}\|)\\
        &\leq \frac{1.1}{\gamma_m}(\beta_{\max}\varepsilon_m^{(1)}+\sum_{l=1}^{m-1}\varepsilon_{l}^{(2)}),
    \end{split}
\end{equation}
where we have applied the fact that $\|G_m^*\|=\|B_m^{*-1}B_1^*\|\leq \beta_{\max}$, and $J_m^{(1)}\subset \cup_{l=1}^{m-1}(I_l\backslash J_{l}^{(1)})$. We have also utilized the definitions $\varepsilon_m^{(1)} = \|\h{V}_m\h{W}_m - V_m^*\|$, the error term for each block SVD; and $\varepsilon_{l}^{(2)} = \|\til{V}_{I_{l}\backslash J_{l}^{(1)},:} - \til{V}_{I_{l}\backslash J_{l}^{(1)},:}^*\|$, the final error of the $l$'s block after the matching step.
Putting \eqref{eq:induction1} and \eqref{eq:induction2} together, and using an induction proof, we can show the following claim.
\begin{claim}\label{claim:induction_result}
Let $\varepsilon$ denote the maximum block-SVD error $\max_{1\leq m\leq M}\varepsilon_m^{(1)}$, and $\til{\gamma}_m = \frac{1.1}{\gamma_m}$. Then $\varepsilon_1^{(2)}\leq \varepsilon$, and for $m=2,\dots, M$,
\begin{equation}\label{eq:induction_result}
    \begin{split}
    \varepsilon_m^{(2)}\leq &(2\beta_{\max}+1)\til{\gamma}_m\prod_{l=2}^{m-1}(\til{\gamma}_l+1)\varepsilon,\\
        \sum_{l=1}^m\varepsilon_m^{(2)} \leq &(2\beta_{\max}+1)\prod_{l=2}^{m}(\til{\gamma}_l+1)\varepsilon - \beta_{\max}\varepsilon.
    \end{split}
\end{equation}
\end{claim}
We will assume this claim to be true for the rest of the arguments while leaving its proof in the end. Combining Claim \ref{claim:induction_result} and \eqref{eq:induction2}, for $m\geq 2$, we have
\begin{equation}\label{eq:G_m_err}
    \begin{split}
       \|G_m-G_m^*\| = \varepsilon_m^{(3)} \leq &\til{\gamma}_m\beta_{\max}\varepsilon + \tilde{\gamma}_m\sum_{l=1}^{m-1}\varepsilon_{l}^{(2)}\\
       \leq &(2\beta_{\max}+1)\til{\gamma}_m\prod_{l=2}^{m-1}(\til{\gamma}_l+1)\varepsilon.
    \end{split}
\end{equation}
Recall \eqref{eq:V_H_fullerr}, the estimation error for the whole $\til{V}$ can be bounded as follows:
\begin{equation}\label{eq:til_V_err}
    \begin{split}
        \|\til{V}-\til{V}^*\|&\leq \sum_{m=1}^M\varepsilon_m^{(2)}\\
        &\leq (2\beta_{\max} + 1)\prod_{l=2}^{M}(\til{\gamma}_l+1)\varepsilon.
    \end{split}
\end{equation}
Now we turn to the left singular space: $\|\til{H}-\til{H}^*\|$. We still consider the contributed error $\|\til{H}_{T_m,:}-\til{H}^*_{T_m,:}\|$ of each block first, which then requires an error bound for $\|G_m^{-1}-G_m^{*-1}\|$. Recall our bound for $\|G_m-G_m^*\|$ in \eqref{eq:G_m_err}. 
We can further show that
\begin{equation*}
\begin{split}
    \|G_m-G_m^*\|&\leq 2\sqrt{2}(2\beta_{\max}+1)\prod_{l=2}^m\left(\frac{1.1}{\gamma_l}+1\right)\max_m\frac{\|E_{T_m,I_m}\|}{\sigma_r(X^*_{T_m,I_m})}\\
    &\leq  4\sqrt{2}(\beta_{\max}+1)\prod_{l=2}^M\left(\frac{1.1}{\gamma_l}+1\right)\max_m\frac{\|E_{T_m,I_m}\|}{\sigma_r(X^*_{T_m,I_m})}\\
    &\leq \frac{1}{2}\beta_{\min},
\end{split}
\end{equation*}
where the first inequality is due to \eqref{eq:G_m_err}, the definition $\varepsilon = \max_m\epsilon_m^{(1)}$ in Claim \ref{claim:induction_result}, \eqref{eq:block_V_err}, and $\til{\gamma}_m=\frac{1.1}{\gamma_m}$; the third inequality makes use of Assumption \ref{assump:block_eigmin} with a sufficiently large constant $C>0$. Here, $\beta_{\min}=\min_{1\leq m\leq M}\sigma_r(B_m^{*-1}B_1^*)
\leq \sigma_r(G_m^*)$. Then through some calculations, one can show that
\begin{equation*}
\begin{split}
    \|G_m^{-1}-G_m^{*-1}\|&=\|G_m^{-1}(G_m^*-G_m)G_m^{*-1}\|\\
    &\leq \frac{\|G_m^*-G_m\|}{\sigma_r(G_m)\sigma_r(G_m^*)}\\
    &\leq 2\beta_{\min}^{-2}\|G_m-G_m^*\|\\
    &\leq \beta_{\min}^{-1},
\end{split}
\end{equation*}
where the third inequality utilizes the fact that $\sigma_r(G_m)\geq \sigma_r(G_m^*) - \|G_m-G_m^*\|\geq \frac{1}{2}\sigma_r(G_m^*)$, and the last inequality is due to $\|G_m-G_m^*\|\leq \frac{1}{2}\beta_{\min}$. Therefore, we can plug in the error bound above into \eqref{eq:H_block_err}:
\begin{equation*}
    \begin{split}
        \|\til{H}_{T_m,:}-\til{H}_{T_m,:}^*\|
        &\leq \|U_m^*\Lambda_m^*\|\|G_m^{-1}-G_m^{*-1}\| + 2\beta_{\min}^{-1}\|\h{H}_m\h{W}_m - U_m^*\Lambda_m^*\|\\
        &\leq \frac{2\|\Lambda_m^*\|}{\beta_{\min}^2}\left(\|G_m-G_m^*\| + \frac{4\beta_{\min}\|E_{T_m,I_m}\|}{\sigma_r(X^*_{T_m,I_m})}\right)\\
        &\leq \frac{C\|\Lambda_m^*\|}{\beta_{\min}^2}(\beta_{\max}+1)\til{\gamma}_m\prod_{l=2}^{m-2}(\til{\gamma}_l+1)\max_m\frac{\|E_{T_m,I_m}\|}{\sigma_r(X^*_{T_m,I_m})},
    \end{split}
\end{equation*}
where the second inequality is due to \eqref{eq:H_m_err}; the third inequality is due to \eqref{eq:G_m_err} and the fact that $\varepsilon \leq 2\sqrt{2}\max_m\frac{\|E_{T_m,I_m}\|}{\sigma_r(X^*_{T_m,I_m})}$. Combining all the bounds for $\|\til{H}_{T_m,:}-\til{H}_{T_m,:}^*\|$ over $m=1,\dots,M$, we have
\begin{equation}\label{eq:til_H_err}
    \begin{split}
        \|\til{H}-\til{H}^*\|&\leq \frac{C\max_m\|\Lambda_m^*\|}{\beta_{\min}^2}(\beta_{\max}+1)\sqrt{\sum_{m=1}^M\prod_{l=2}^{m}(\til{\gamma}_l+1)^{2}}\max_m\frac{\|E_{T_m,I_m}\|}{\sigma_r(X^*_{T_m,I_m})}\\
        &\leq C\beta_{\min}^{-2}(\beta_{\max}+1)\prod_{l=2}^M\left(\frac{1.1}{\gamma_l}+1\right)\max_m\|\Lambda_m^*\|\max_m\frac{\|E_{T_m,I_m}\|}{\sigma_r(X^*_{T_m,I_m})}.
    \end{split}
\end{equation}
In the last line above, we have utilized the fact that 
\begin{equation*}
    \begin{split}
        &\sqrt{\sum_{m=1}^M\prod_{l=2}^{m}(\til{\gamma}_l+1)^{2}}\\
        \leq &\frac{\prod_{l=2}^M(\til{\gamma}_l + 1)^2(\min_l\til{\gamma}_l + 1)^2 - 1}{(\min_l\til{\gamma}_l+1)^2 - 1}\\
        \leq &C\prod_{l=2}^M(\til{\gamma}_l + 1)^2,
    \end{split}
\end{equation*}
where the second line is due to Lemma \ref{lem:sum_prod}, and the last line is due to $\min_l\til{\gamma}_l\geq 1.1$.  

Now we would like to combine \eqref{eq:til_V_err} and \eqref{eq:til_H_err}. Note that Assumption \ref{assump:block_eigmin} suggests the following:
\begin{equation*}
    \begin{split}
        \|\til{V}-\til{V}^*\|&\leq (2\beta_{\max}+1)\prod_{l=2}^M(\til{\gamma}_l+1)\varepsilon\\
        &\leq 2\sqrt{2}(2\beta_{\max}+1)\prod_{l=2}^M\left(\frac{1.1}{\gamma_l}+1\right)\max_m\frac{\|E_{T_m,I_m}\|}{\sigma_r(X^*_{T_m,I_m})}\\
        &\leq C,
    \end{split}
\end{equation*}
where the second inequality is due to the fact that $\varepsilon \leq 2\sqrt{2}\max_m\frac{\|E_{T_m,I_m}\|}{\sigma_r(X^*_{T_m,I_m})}$, and the last line is due to Assumption \ref{assump:block_eigmin}.
Finally, we also note that the spectral norm of $\til{V}^*$ and $\til{H}^*$ can be bounded as follows:
\begin{equation*}
\begin{split}
    \|\til{V}^*\| &= \|V^*B_1^*\|\leq \|B_1^*\|=\|\Lambda^*U_{T_1,:}^{*\top}U_1^*\Lambda_1^{*-1}\|\leq \frac{\|X^*\|}{\sigma_r(\Lambda^*_1)},\\
    \|\til{H}^*\| &=\|U^*\Lambda^*(B_1^{*-1})^\top\|\leq \|\Lambda^*\|\|B_1^{*-1}\|=\|\Lambda^*\|\|V_1^{*\top}V_{I_1,:}^*\|\leq \|X^*\|,
\end{split}
\end{equation*}
where we have applied \eqref{eq:V_B} and \eqref{eq:U_B}, which suggest $B_1^* = \Lambda^*U_{T_1,:}^{*\top}U_1^*\Lambda_1^{*-1}$, $B_1^{*-1} = V_1^{*\top}V_{I_1,:}^*$.
Therefore, we can now bound $\|\til{H}\til{V}-X^*\|$ as follows:
\begin{equation*}
    \begin{split}
        \|\til{H}\til{V}-X^*\|&\leq \|\til{H}-\til{H}^*\|(\|\til{V}^*\|+\|\til{V}-\til{V}^*\|)+\|\til{V}-\til{V}^*\|\|\til{H}^*\|\\
        &\leq C\beta_{\min}^{-2}(\beta_{\max}+1)\max_{m}\frac{\|\Lambda_m^*\|}{\sigma_r(\Lambda_1^*)}\prod_{l=2}^M\left(\frac{1.1}{\gamma_l}+1\right)\max_m\frac{\|E_{T_m,I_m}\|}{\sigma_r(X^*_{T_m,I_m})}\|X^*\|\\
        &\leq C\frac{\alpha_{\max}(\beta_{\max}+1)}{\beta_{\min}^{2}}\prod_{l=2}^M\left(\frac{1.1}{\gamma_l}+1\right)\max_m\frac{\|E_{T_m,I_m}\|}{\sigma_r(X^*_{T_m,I_m})}\|X^*\|.
    \end{split}
\end{equation*}

\begin{proof}[Proof of Claim \ref{claim:induction_result}]
We first combine \eqref{eq:induction1} and \eqref{eq:induction2} and obtain the following: for $2\leq m\leq M$,
\begin{equation}\label{eq:induction3}
    \begin{split}
        \varepsilon_m^{(2)}\leq \beta_{\max}(\til{\gamma}_m + 1)\varepsilon_m^{(1)} + \til{\gamma}_m\sum_{l=1}^{m-1}\varepsilon_{l}^{(2)}.
    \end{split}
\end{equation}
Let $\delta_m = \sum_{l=1}^m \varepsilon_l^{(2)}$. Then \eqref{eq:induction3} can be rewritten as 
\begin{equation}\label{eq:induction4}
    \delta_m \leq \beta_{\max}(\til{\gamma}_m + 1)\varepsilon + (\til{\gamma}_m+1)\delta_{m-1}.
\end{equation}
We first note that when $m=1$, $J_m^{(1)}=\emptyset$ and hence $$\varepsilon_1^{(2)} = \|\til{V}_{I_1,:}-\til{V}^*_{I_1,:}\|=\|\h{V}_1\h{W}_1 - V_1^*\|=\varepsilon^{(1)}_1.$$ Immediately, we have $\delta_1 = \varepsilon_1^{(2)}\leq \varepsilon$.

In the following, we will prove by induction that for $m\geq 2$.
\begin{equation}\label{eq:delta_m_bnd}
    \delta_m \leq \left[\sum_{k=2}^m \prod_{l=k}^m(\til{\gamma}_l + 1)\beta_{\max} + \prod_{l=2}^m(\til{\gamma}_l+1)\right]\varepsilon.
\end{equation}
When $m=2$, \eqref{eq:induction4} implies that $\delta_2\leq \beta_{\max}(\til{\gamma}_2 + 1)\varepsilon + (\til{\gamma}_2+1)\varepsilon$, which satisfies \eqref{eq:delta_m_bnd}. Suppose that \eqref{eq:delta_m_bnd} holds for $2\leq m\leq i$, then when $m=i+1$, \eqref{eq:induction4} suggests that
\begin{equation*}
\begin{split}
    \delta_{i+1}\leq &\beta_{\max}(\til{\gamma}_{i+1} + 1)\varepsilon + (\til{\gamma}_{i+1}+1)\delta_{i}\\
    \leq &\beta_{\max}(\til{\gamma}_{i+1} + 1)\varepsilon + (\til{\gamma}_{i+1}+1)\left[\sum_{k=2}^i \prod_{l=k}^i(\til{\gamma}_l + 1)\beta_{\max} + \prod_{l=2}^i(\til{\gamma}_l+1)\right]\varepsilon\\
    \leq &\left[\beta_{\max}(\til{\gamma}_{i+1} + 1) + \sum_{k=2}^i \prod_{l=k}^{i+1}(\til{\gamma}_l + 1)\beta_{\max} + \prod_{l=2}^{i+1}(\til{\gamma}_l+1)\right]\varepsilon\\
    \leq & \left[\sum_{k=2}^{i+1} \prod_{l=k}^{i+1}(\til{\gamma}_l + 1)\beta_{\max} + \prod_{l=2}^{i+1}(\til{\gamma}_l+1)\right]\varepsilon,
\end{split}
\end{equation*}
and hence \eqref{eq:delta_m_bnd} holds for all $2\leq m\leq M$.

Now we combine \eqref{eq:delta_m_bnd} and \eqref{eq:induction3} to show that
\begin{equation}\label{eq:err_m2_bnd}
        \varepsilon_m^{(2)}\leq \left[(\til{\gamma}_m + 1)\beta_{\max} + \til{\gamma}_m \sum_{k=2}^{m-1} \prod_{l=k}^{m-1}(\til{\gamma}_l + 1)\beta_{\max} + \til{\gamma}_m\prod_{l=2}^{m-1}(\til{\gamma}_l+1)\right]\varepsilon.
\end{equation}

We now apply Lemma \ref{lem:sum_prod} to show an upper bound for $\sum_{k=2}^{m} \prod_{l=k}^{m}(\til{\gamma}_l + 1)$, which can help simplify both \eqref{eq:delta_m_bnd} and \eqref{eq:err_m2_bnd}. Let $\til{\gamma}_{\min} = \min_l\til{\gamma}_l$. Lemma \ref{lem:sum_prod} then implies 
$$
\sum_{k=2}^{m} \prod_{l=k}^{m}(\til{\gamma}_l + 1)\leq \til{\gamma}_{\min}^{-1}\left[\prod_{l=2}^{m}(\til{\gamma}_l + 1)(\til{\gamma}_{\min}+1) - (\til{\gamma}_m + 1)\right].
$$
Since $0<\gamma_l\leq 1$ for all $2\leq l\leq M$, we have $\til{\gamma}_{\min} = \min_l\frac{1.1}{\gamma_l}\geq 1.1$. Therefore, \eqref{eq:delta_m_bnd} and \eqref{eq:err_m2_bnd} can be simplified as follows:
\begin{equation*}
    \begin{split}
        \delta_m \leq &(2\beta_{\max}+1)\prod_{l=2}^{m}(\til{\gamma}_l+1)\varepsilon - \beta_{\max}\varepsilon,\\
        \varepsilon_m^{(2)}\leq &(2\beta_{\max}+1)\til{\gamma}_m\prod_{l=2}^{m-1}(\til{\gamma}_l+1)\varepsilon + (\til{\gamma}_m+1)\beta_{\max}\varepsilon - \frac{(\til{\gamma}_{m-1}+1)\til{\gamma}_m}{\til{\gamma}_{\min}}\beta_{\max}\varepsilon\\
        \leq &(2\beta_{\max}+1)\til{\gamma}_m\prod_{l=2}^{m-1}(\til{\gamma}_l+1)\varepsilon,
    \end{split}
\end{equation*}
where the last line is due to the fact that $\frac{(\til{\gamma}_{m-1}+1)\til{\gamma}_m)}{\til{\gamma}_{\min}}=\frac{\til{\gamma}_{m-1}\til{\gamma}_m}{\til{\gamma}_{\min}} + \frac{\til{\gamma}_m}{\til{\gamma}_{\min}}\geq \til{\gamma}_m + 1$. Therefore, the proof of Claim \ref{claim:induction_result} is complete.
\end{proof}
\subsection{Proofs of Auxiliary Lemmas}

\begin{proof}[Proof of Lemma \ref{lem:spectral_cluster_err}]
    Our proof follows very similar arguments to \cite{loffler2021optimality}, but we still state most steps for clarity and completeness. Since $\widetilde{H}\widetilde{V}^\top$ and $X^*$ are both of rank $r$, the spectral norm error bound $\|\widetilde{H}\widetilde{V}^\top-X^*\|_2\leq \varepsilon$ implies the following Frobenious error bound:
    \begin{equation*}
        \|\widetilde{H}\widetilde{V}^\top-X^*\|_F\leq \sqrt{2r}\varepsilon.
    \end{equation*}
    Let $\widehat{X}\in \mathbb{R}^{p\times n}$ be the estimated center matrix with each column representing the estimated center mean of the corresponding sample $i$: $\widehat{X}_{:,i} = \widetilde{U}\widehat{c}_{z_i}$, where $\widehat{c}_j\in \mathbb{R}^r$ and $\widehat{z}_i\in [K]$ are the global minimizer of the k-means step in Algorithm \ref{alg:bsvdcq}. By the definition of the k-means optimization problem, one has
    \begin{equation*}
    \begin{split}
        \|\widetilde{H}\widetilde{V}^\top - \widehat{X}\|_F = \|\widetilde{U}^\top \widetilde{H}\widetilde{V}^\top - \widetilde{U}^\top \widehat{X}\|_F\leq &\|\widetilde{U}^\top \widetilde{H}\widetilde{V}^\top - \widetilde{U}^\top X^*\|_F\\ 
        = &\|\widetilde{H}\widetilde{V}^\top - X^*\|_F,
    \end{split}
    \end{equation*}
    which then implies $\|\widehat{X} - X^*\|_F\leq 2\sqrt{2r}\varepsilon$. Define the sample set $S = \{i\in[n]: \|\widehat{X}_{:,i}-X^*_{:,i}\|_2>\frac{\Delta}{2}\}$. The Frobenious norm error bound for $\widehat{X} - X^*$ then implies an upper bound for the cardinality of set $S$:
    \begin{equation*}
        |S|\leq \frac{\|\widehat{X} - X^*\|_F^2}{\Delta^2/4} \leq \frac{32r\varepsilon^2}{\Delta^2}.
    \end{equation*}
    In addition, since we have assumed $\varepsilon<\sqrt{\frac{\beta n}{rK}}\frac{\Delta}{8}$, the inequality above also implies $|S|\leq \frac{\beta n}{2K}$. Following the same arguments as in the proof of Lemma 4.2 in \cite{loffler2021optimality}, we then have $\ell(\widehat{z},z^*)\leq \frac{1}{n}|S|\leq \frac{32r\varepsilon^2}{n\Delta^2}$, and $\min_{\phi\in \Phi}\max_{1\leq j\leq K}\|\widehat{U}\widehat{c}_j - \theta^*_{\phi(j)}\|_2^2 \leq \frac{16Kr\varepsilon^2}{\beta n}$.
\end{proof}
\begin{proof}[Proof of Lemma \ref{lem:OLS_err}]
    Since we have assumed $\sigma_r(Z^*)>\varepsilon_1$, by Weyl's inequality, one has $\sigma_r(Z) \geq \sigma_r(Z^*) - \varepsilon_1 >0$, which suggests that $Z^\top Z\in \mathbb{R}^{r\times r}$ is invertible. Now we use some rearrangement of the least squares formulation of $G$:
    \begin{equation*}
    \begin{split}
        G &= (Z^\top Z)^{-1}Z^\top Y\\
        &=(Z^\top Z)^{-1}Z^\top (Z^*G^* + Y-Y^*)\\
        &= G^* +(Z^\top Z)^{-1}Z^\top [Y-Y^*-(Z-Z^*)G^*],
    \end{split}
    \end{equation*}
    Thus we only need to bound $\|(Z^\top Z)^{-1}Z^\top [(Z-Z^*)G^* + Y-Y^*]\|$. Recognizing the fact that $\|(Z^\top Z)^{-1}Z^\top\|= \sigma_r^{-1}(Z)\leq (\sigma_r(Z^*)-\varepsilon_1)^{-1}$, one can show that
    \begin{equation*}
        \begin{split}
            &\|(Z^\top Z)^{-1}Z^\top [(Z-Z^*)G^* + Y-Y^*]\|\\
            \leq&\frac{\|(Z-Z^*)G^* + Y-Y^*\|}{\sigma_r(Z^*) - \varepsilon_1}\\
            \leq&\frac{\|G^*\|\varepsilon_1+\varepsilon_2}{\sigma_r(Z^*) - \varepsilon_1}.
        \end{split}
    \end{equation*}
    The proof is now complete.
\end{proof}
\subsection{Proofs of Corollary \ref{cor:main_delta} and Theorem \ref{thm:main2}}
\begin{proof}[Proof of Corollary \ref{cor:main_delta}]
    We first establish a connection between the separation distance and the blockwise SNR in terms of minimum singular value. In particular, for each block $1\leq m\leq M$, $X_{T_m,I_m}^*=\Theta^*_{T_m,:}F^{*\top}_{I_m,:}$, where $\Theta^*_{T_m,:}\in \bR^{p_m\times K}$, $F^{*}_{I_m,:}\in \bR^{n_m\times K}$. For any $u\in \bR^{K}$ with $\|u\|_2=1$, we have $\|F^{*}_{I_m,:}u\|_2 = \sqrt{\sum_{j=1}^Ku_j^2n_{j,m}}\geq \sqrt{\min_j n_{j,m}}$, and hence $\sigma_{\min}(F^*_{I_m,:}) \geq \sqrt{\min_j n_{j,m}}$. In addition, recall that $\Delta_m = \min_{j\neq j'}\|\theta_{T_m,j}-\theta_{T_m,j'}\|_2$. Let $(j_1,j_2) = \argmin_{j_1\neq j_2}\|\theta_{T_m,j_1}-\theta_{T_m,j_2}\|_2$, $u = \frac{1}{\sqrt{2}}(e_{j_1} - e_{j_2})$, where $e_{j_1},\,e_{j_2}$ are canonical vectors. Then we have $\|u\|_2=1$, and $\|\Theta^*_{T_m,:}u\|_2 = \frac{1}{\sqrt{2}}\Delta_m$, which implies $\|\Theta^*_{T_m,:}\|\geq \frac{1}{\sqrt{2}}\Delta_m$. Therefore, we can lower bound the $r$th singular value of $X_{T_m,I_m}^*$ as follows:
    \begin{equation*}
        \begin{split}
            \sigma_r(X^*_{T_m,I_m})\geq &\sigma_{\min}(F_{I_m,:}^*) \sigma_r(\Theta^*_{T_m,:})\\
            \geq &\sqrt{\min_{j}n_{j,m}}\|\Theta^*_{T_m,:}\|\kappa_r(\Theta^*_{T_m,:})\\
            \geq&\sqrt{\frac{\min_{j}n_{j,m}}{2}}\frac{\Delta_m}{\kappa(\Theta^*_{T_m,:})}.
        \end{split}
    \end{equation*}
    Thus Assumption \ref{assump:block_separation} implies Assumption \ref{assump:block_eigmin}. Combining the lower bound above with Theorem \ref{thm:main} leads to Theorem \ref{thm:main2}.
\end{proof}
\begin{proof}[Proof of Theorem \ref{thm:main2}]
    Recall that the noise matrix $E = A_rZA_c^\top$, where $A_r\in \bR^{n\times n'}$, $A_c\in \bR^{p\times p'}$, $Z\in \bR^{n'\times p'}$ has independent zero-mean sub-Gaussian($\sigma$) entries. We can first bound the spectral norm of each submatrix $E_{T_m,I_m} = (A_r)_{T_m,:}Z(A_c)_{I_m,:}^\top$ by a function of $\sigma$, spectral norms $\|A_r\|$, $\|A_c\|$, and its dimensions. In particular, we apply a technical lemma from \cite{zhou2022optimal} (Lemma 8.2), with $A=\frac{1}{\|(A_r)_{T_m,:}\|}(A_r)_{T_m,:}$, $B=\frac{1}{\|(A_c)_{I_m,:}\|}(A_c)_{I_m,:}^\top$, $t=C(n_m + p_m)$ for a sufficiently large universal constant. Then one can write
    \begin{equation*}
    \begin{split}
        \bbP(\|E_{T_m,I_m}\|\geq &2\sigma\|(A_r)_{T_m,:}\|\|(A_c)_{I_m,:}\|\sqrt{2p_m + n_m})\\
        \leq &2\cdot 5^{n_m}\exp\{-c_1\min\{\frac{(n_m+p_m)^2}{p_m},n_m+p_m\}\}\\
        \leq &2\exp\{-c_2(n_m+p_m)\}.
    \end{split}
    \end{equation*}
    Therefore, with probability at least $1-2\exp\{-c(n_m+p_m)\}$, $\|E_{T_m,I_m}\|\leq C\sigma\|A_r\|\|A_c\|(\sqrt{p_m}+\sqrt{n_m})$ for some universal constants $c,\,C>0$. We then note that \eqref{eq:SNR_subGaussian} together with Assumption \ref{assump:BndHeteroBalanced} implies Assumption \ref{assump:block_separation}; the misclustering rate in \eqref{eq:cluster_err_separation} and the bound above for $\|E_{T_m,I_m}\|$ immediately implies Theorem \ref{thm:main2}.
\end{proof}

\begin{figure}[b]
\centering
    \includegraphics[width=0.8\linewidth]{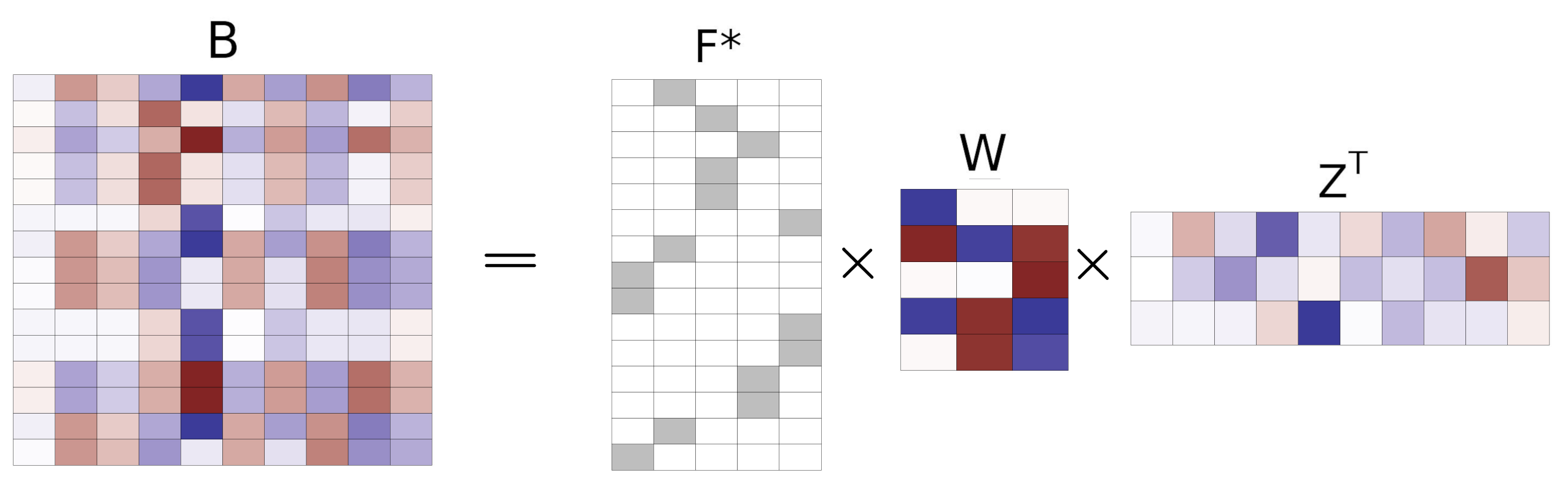}
        \caption{Illustration of data generation process for Gaussian mixture model simulations.}
    \label{fig:dg}
\end{figure}

\subsection{Proof of Theorem \ref{thm:datadriven}}\label{sec:datadrivenproof}
We first prove that when Algorithm \ref{alg:bsvdcq} is applied on data patches $\{X_{T_{\hat{\pi}_m}, I_{\hat{\pi}_m}}\}_{m=1}^M$ with data-driven patch ordering $\hat{\pi}$, and when the factor $\prod_{m=2}^M\left(1.1\gamma_m^{-1}+1\right)$ in Assumption \ref{assump:block_eigmin} and Theorem \ref{thm:main} is substituted by $\prod_{m=2}^M(1.1\gamma_m^{-1}(\pi^*) + 1)$, Theorem \ref{thm:main} still holds. In fact, it suffices to show that the matching error factor for the data-driven selected ordering $\hat{\pi}$ satisfies
\begin{equation}\label{eq:datadriven_factorbnd}
    \prod_{m=2}^M\left(1.1\gamma_m^{-1}(\hat{\pi}) + 1\right)\leq e^2\prod_{m=2}^M \left(1.1\gamma_m^{-1}(\pi^*) + 1\right);
\end{equation}
then we can simply invoke the original Theorem \ref{thm:main} on data patches $\{X_{T_{\hat{\pi}(m)}, I_{\hat{\pi}(m)}}\}_{m=1}^M$ to finish the proof.

Let $\hat{S}(\pi) = \prod_{m=2}^M s(\pi(m), \cup_{l<m}\pi(l))$ be the finite-sample score for a given ordering $\pi$; and let $S^*(\pi) =  \prod_{m=2}^M (1.1\gamma_m^{-1}(\pi) + 1)^{-1}$ be its population counterpart. To show \eqref{eq:datadriven_factorbnd}, it might be tempting to establish an error bound between the finite sample score function $\hat{S}(\pi)$ and its population counterpart $S^*(\pi)$ for all $\pi$. However, there possibly exist some patch orderings $\pi$ such that the overlapping set $J_m^{(1)}(\pi)$ is tiny; and for these terrible orderings, $\hat{S}(\pi)$ is not close to $S^*(\pi)$ at all. Therefore, one key step in our proof is to show that the selected $\hat{\pi}$ cannot be too bad under the given assumptions. For notational convenience, we let $\overline{\mathrm{SNR}}(m,\pi) = \frac{\big\|X^*_{T_{\pi(m)}, I_{\pi(m)}}\big\|}{\big\|E_{T_{\pi(m)}, I_{\pi(m)}}\big\|}$, $\underline{\mathrm{SNR}}(m,\pi) = \frac{\sigma_r\big(X^*_{T_{\pi(m)},J_m^{(1)}(\pi)}\big)}{\big\|E_{T_{\pi(m)},J_m^{(1)}(\pi)}\big\|}$. We first note that Assumption \ref{assump:block_eigmin} on the oracle patch ordering implies that
\begin{equation}\label{eq:assump2'}
    \frac{\sigma_r(X^*_{T_m,I_m})}{\|E_{T_m,I_m}\|}\geq \frac{C\alpha_{\max}(\beta_{\max}+1)}{\beta_{\min}^2}\prod_{l=2}^M\left(1.1\gamma_l^{-1}(\pi^*)+1\right)\sqrt{\frac{rK\max_j n_j}{\min_j n_j}},
\end{equation}
which then implies that for any $1\leq m\leq M$, 
\begin{equation}\label{eq:uppSNRbnd}
\overline{\mathrm{SNR}}(m,\pi^*) \geq 2C\times 2.1^{M-1}\sqrt{rK}\geq 4.4CM - 2.9C,    
\end{equation}
where we have utilized the fact that $\alpha_{\max},\,\beta_{\max}\geq 1$, $0<\beta_{\min},\gamma_m(\pi^*)\leq 1$, $K\geq 2$, and the second inequality holds due to the property of the exponential function of $M\geq 2$. Similarly, we have
\begin{equation}\label{eq:lowerSNRbnd}
\underline{\mathrm{SNR}}(m,\pi^*) \geq \gamma_m(\pi^*)\overline{\mathrm{SNR}}(m,\pi^*)\geq 2.2C\times 2.1^{M-2}\sqrt{rK}\geq 2.3CM-1.6C.
\end{equation}
By Weyl's inequality, we have 
\begin{equation}\label{eq:block_SNRbnd}
    \begin{split}
    	\frac{\|X_{T_{\pi^*(m)}, I_{\pi^*(m)}}\|}{\sigma_r(X_{T_{\pi^*(m)},J_m^{(1)}(\pi^*)})}\leq &\frac{\|X^*_{T_{\pi^*(m)}, I_{\pi^*(m)}}\|+\|E_{T_{\pi^*(m)}, I_{\pi^*(m)}}\|}{\sigma_r(X^*_{T_{\pi^*(m)},J_m^{(1)}(\pi^*)})-\|E_{T_{\pi^*(m)},J_m^{(1)}(\pi^*)}\|}\\
     \leq &\frac{\|X^*_{T_{\pi^*(m)}, I_{\pi^*(m)}}\|}{\sigma_r(X^*_{T_{\pi^*(m)},J_m^{(1)}(\pi^*)})}\left(1+\frac{1}{\overline{\mathrm{SNR}}(m,\pi^*)}\right)\left(1-\frac{1}{\underline{\mathrm{SNR}}(m,\pi^*)}\right)^{-1},
    \end{split}
\end{equation}
and the inflated error factor can be further written as 
\begin{equation}\label{eq:inflationFactorOracle}
    \begin{split}
        \left(1+\frac{1}{\overline{\mathrm{SNR}}(m,\pi^*)}\right)\left(1-\frac{1}{\underline{\mathrm{SNR}}(m,\pi^*)}\right)^{-1}=1 + \left(\frac{1-\overline{\mathrm{SNR}}^{-1}(m,\hat{\pi})}{\overline{\mathrm{SNR}}^{-1}(m,\pi^*)+\underline{\mathrm{SNR}}^{-1}(m,\pi^*)}\right)^{-1}.
    \end{split}
\end{equation}
By setting the constant $C>0$ sufficiently large in Assumption \ref{assump:block_eigmin}, we can guarantee that $\overline{\mathrm{SNR}}(m,\pi^*),\,\underline{\mathrm{SNR}}(m,\pi^*)\geq 10$, and hence \eqref{eq:uppSNRbnd} and \eqref{eq:lowerSNRbnd} implies
\begin{equation*}
    \begin{split}
        \left(1+\frac{1}{\overline{\mathrm{SNR}}(m,\pi^*)}\right)\left(1-\frac{1}{\underline{\mathrm{SNR}}(m,\pi^*)}\right)^{-1}
        \leq &1 + \frac{1}{M-1}.
    \end{split}
\end{equation*}
Therefore, the finite score of the oracle ordering satisfies
\begin{equation*}
    \begin{split}
       \hat{S}(\pi^*) \geq \prod_{m=2}^M\left(\frac{1.1\|X^*_{T_{\pi^*(m)}, I_{\pi^*(m)}}\|}{\sigma_r(X^*_{T_{\pi^*(m)},J_m^{(1)}(\pi^*)})} + 1\right)^{-1}\big(1+\frac{1}{M-1}\big)^{M-1} \leq e^{-1}\prod_{m=2}^M(1.1\gamma_m^{-1}(\pi^*)+ 1)^{-1}.
    \end{split}
\end{equation*}
Since the selected patch ordering maximizes the overlapping score function, we have %and applying Assumption \ref{assump:block_eigmin} with the new factor, we have
\begin{equation}\label{eq:pihat_helper1}
	\begin{split}
		\hat{S}(\hat{\pi})\geq \hat{S}(\pi^*)\geq e^{-1}\prod_{m=2}^M(1.1\gamma_m^{-1}(\pi^*)+ 1)^{-1}.
	\end{split}
\end{equation}
As mentioned earlier, in order to show \eqref{eq:datadriven_factorbnd}, we would like to establish a connection between $\hat{S}(\hat{\pi})$ and $S^*(\hat{\pi})$; to achieve this, we also need to ensure sufficiently large $\underline{\mathrm{SNR}}(m,\hat{\pi})$. We plug in the definition of $\hat{S}(\hat{\pi})$ into \eqref{eq:pihat_helper1} and apply Assumption \ref{assump:block_eigmin}:
\begin{equation}\label{eq:pihat_helper2}
	\begin{split}
	\hat{S}(\hat{\pi}) = \prod_{l=2}^M\left(\frac{1.1\|X_{T_{\hat{\pi}(l)},I_{\hat{\pi}(l)}}\|}{\sigma_r(X_{T_{\hat{\pi}(l)},J_l^{(1)}(\hat{\pi})})}+1\right)^{-1}\geq \frac{2C\sqrt{rK}\|E_{T_m,I_m}\|}{\sigma_r(X^*_{T_m,I_m})},
	\end{split}
\end{equation}
which holds for all $1\leq m\leq M$. Since $\frac{\|X_{T_{\hat{\pi}(l)},I_{\hat{\pi}(l)}}\|}{\sigma_r(X_{T_{\hat{\pi}(l)},J_l^{(1)}(\hat{\pi})})}\geq 1$ for all $2\leq l\leq M$, \eqref{eq:pihat_helper2} implies that
\begin{equation*}
	\begin{split}
	&\frac{1.1\|X_{T_{\hat{\pi}(m)},I_{\hat{\pi}(m)}}\|}{\sigma_r(X_{T_{\hat{\pi}(m)},J_m^{(1)}(\hat{\pi})})}+1\leq \frac{\sigma_r(X^*_{T_{\hat{\pi}(m)},I_{\hat{\pi}(m)}})}{2C\times 2.1^{M-2}\sqrt{rK}\|E_{T_{\hat{\pi}(m)},I_{\hat{\pi}(m)}}\|},\\
	\Rightarrow &\frac{1.1\left(\|X^*_{T_{\hat{\pi}(m)},I_{\hat{\pi}(m)}}\| - \|E_{T_{\hat{\pi}(m)},I_{\hat{\pi}(m)}}\|\right)}{\sigma_r(X_{T_{\hat{\pi}(m)},J_m^{(1)}(\hat{\pi})})}\leq \frac{\sigma_r(X^*_{T_{\hat{\pi}(m)},I_{\hat{\pi}(m)}}) - \|E_{T_{\hat{\pi}(m)},I_{\hat{\pi}(m)}}\|}{2C\times 2.1^{M-2}\sqrt{rK}\|E_{T_{\hat{\pi}(m)},I_{\hat{\pi}(m)}}\|},\\
   \Rightarrow &\sigma_r(X_{T_{\hat{\pi}(m)},J_m^{(1)}(\hat{\pi})})\geq 2.2C\times 2.1^{M-2}\sqrt{rK}\|E_{T_{\hat{\pi}(m)},I_{\hat{\pi}(m)}}\|,\\
   \Rightarrow &\sigma_r(X^*_{T_{\hat{\pi}(m)},J_m^{(1)}(\hat{\pi})})\geq \left(2.2C\times 2.1^{M-2}\sqrt{rK} - 1\right)\|E_{T_{\hat{\pi}(m)},J_m^{(1)}(\hat{\pi})}\|,
	\end{split}
\end{equation*}
where we have utilized $M\geq 2$ in the second line, and the last line is due to the Weyl's inequality. Therefore, we have $\underline{\mathrm{SNR}}(m,\hat{\pi})\geq 2.2\sqrt{2}C\times 2.1^{M-2} - 1\geq 2.3CM - 1.5C - 1$. Since $\overline{\mathrm{SNR}}(m,\hat{\pi})$ also satisfies the same upper bound as $\overline{\mathrm{SNR}}(m,\pi^*)$ in \eqref{eq:uppSNRbnd}, we can follow similar arguments to bounding \eqref{eq:inflationFactorOracle} and obtain:
%%opposite of this? underline, overline?
\begin{equation}\label{eq:inflationFactorTuned}
    \left(1+\frac{1}{\overline{\mathrm{SNR}}(m,\hat{\pi})}\right)\left(1-\frac{1}{\underline{\mathrm{SNR}}(m,\hat{\pi})}\right)^{-1}
        \leq 1 + \frac{1}{M-1}.
\end{equation}
Therefore, we also have
\begin{equation*}
\begin{split}
    S^*(\hat{\pi}) =&\prod_{m=2}^M\left(\frac{1.1\|X^*_{T_{\pi^*(m)}, I_{\pi^*(m)}}\|}{\sigma_r(X^*_{T_{\pi^*(m)},J_m^{(1)}(\pi^*)})} + 1\right)^{-1}\\
    \geq &\prod_{m=2}^M\left(\frac{1.1\|X_{T_{\pi^*(m)}, I_{\pi^*(m)}}\|}{\sigma_r(X_{T_{\pi^*(m)},J_m^{(1)}(\pi^*)})} + 1\right)^{-1}\left(1+\frac{1}{\overline{\mathrm{SNR}}(m,\hat{\pi})}\right)^{-1}\left(1-\frac{1}{\underline{\mathrm{SNR}}(m,\hat{\pi})}\right)\\
    \geq &\hat{S}(\hat{\pi})\left(1 + \frac{1}{M-1}\right)^{-(M-1)}\\
    \geq &e^{-1}\hat{S}(\hat{\pi})\\
    \geq &e^{-2}\prod_{m=2}^M(1.1\gamma_m^{-1}(\pi^*)+ 1)^{-1},
\end{split}
\end{equation*}
where the second line is due to the Weyl's inequality, the third line is due to \eqref{eq:inflationFactorTuned}, and the last two lines are due to \eqref{eq:pihat_helper1}. The definition of $S^*(\pi)$ then immediately implies \eqref{eq:datadriven_factorbnd}, and hence our proof of Theorem \ref{thm:main} is complete. 

For the proof of Corollary \ref{cor:main_delta} and Theorem \ref{thm:main2}, we notice that their proofs do not depend on the patch ordering, and can directly follow given the proof above of \ref{thm:main}.
\section{Additional Simulation Studies} \label{sec:ass}

In this section, we show additional simulation studies from the Gaussian mixture model in Section 3.1, comparing Cluster Quilting with the same comparison incomplete spectral clustering methods that we use in our empirical studies in the main text: MSC-AGL \citep{wen2018incomplete}, IMG \citep{zhao2016incomplete}, DAIMC \citep{hu2019doubly}, and OPIMC \citep{hu2019onepass}. Performances are compared via the Adjusted Rand index with respect to the true underlying cluster labels of the underlying mixture model.

\subsection{Sequential Patch Observed Data, Oracle Tuning} \label{sub:bmgmmo}

We first investigate the efficacy of Cluster Quilting in the sequential patch observed data setting. Here, we generate data using the procedure described in Section \ref{sub:datagen}, with one single cluster centroid matrix used for the entire data set. We then produce the sequential patchwork observation structure by delineating $M$ sequential blocks in the data matrix each containing $o$ rows, with overlap between the rows included in consecutive observation blocks, and $p/M$ non-overlapping columns. These blocks are treated as the observed portion of the data, while the entries of the data matrix which fall outside of the designated blocks are masked. In this step, we ensure that all mixture clusters are represented at least once in each patch. We apply Cluster Quilting and other incomplete spectral clustering methods to the masked data set, with the goal of recovering the true underlying cluster labels of the mixture model. In the simulation study results shown below, we set the baseline parameters of the simulation to be such that the full data sets contain 710 observations of 100 features, with each observation belonging to 1 of 3 true underlying clusters. The cluster centroid matrix is rank 2, with $d = 4.5$ used to generate the elements of the $W$ matrix. From this, 4 blocks are created from the data, each containing 210 observations, which represent the simulated observed portion of the data. We then alter different elements of the simulation parameters in order to study how the performance of each method changes under different conditions; in particular, we present results for changing the number of observed blocks, the number of observations in each block, the number of total features in the data, the number of underlying clusters, the distance between clusters, and the correlation between features. For each simulation setting, we perform 50 repetitions and show the aveestimated cluster labels returned by the Cluster Quilting procedure. Columns are orderedrage and standard deviation of the adjusted Rand index.

\begin{figure}[t]
\begin{center}
    \begin{subfigure}[t]{0.4\linewidth}
    \centering
        \includegraphics[width=\linewidth]{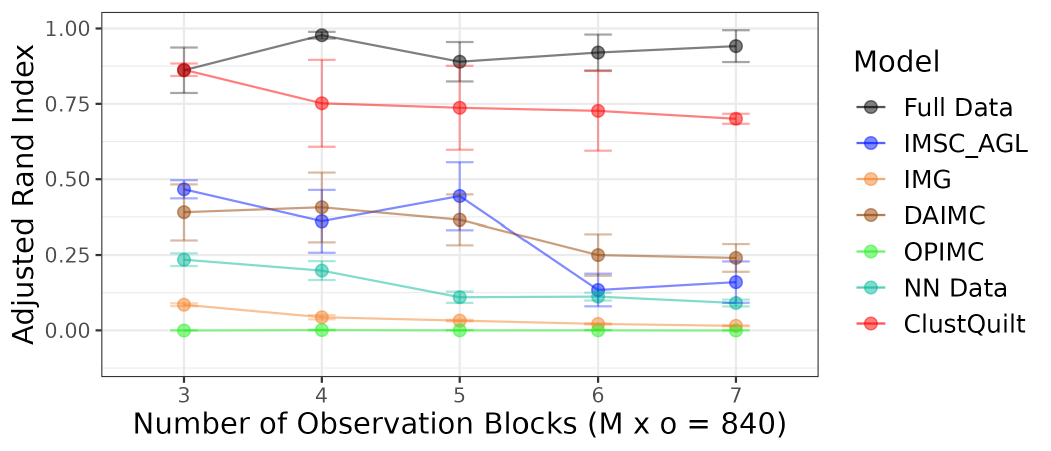}
        \caption{Number of blocks.}
    \label{fig:supp_gmm_m}
    \end{subfigure}%
    \begin{subfigure}[t]{0.4\linewidth}
    \centering
        \includegraphics[width=\linewidth]{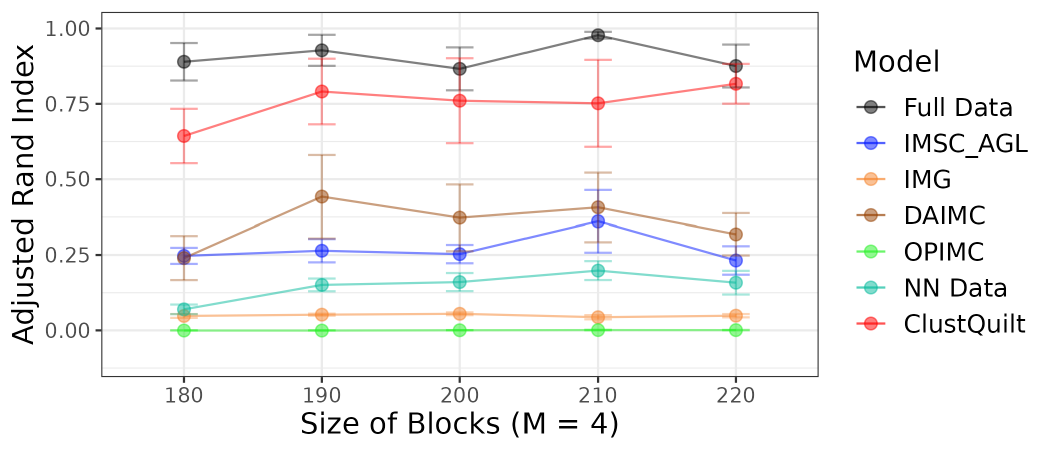}
        \caption{Size of blocks.}
    \label{fig:supp_gmm_o}
    \end{subfigure}
    \begin{subfigure}[t]{0.4\linewidth}
    \centering
        \includegraphics[width=\linewidth]{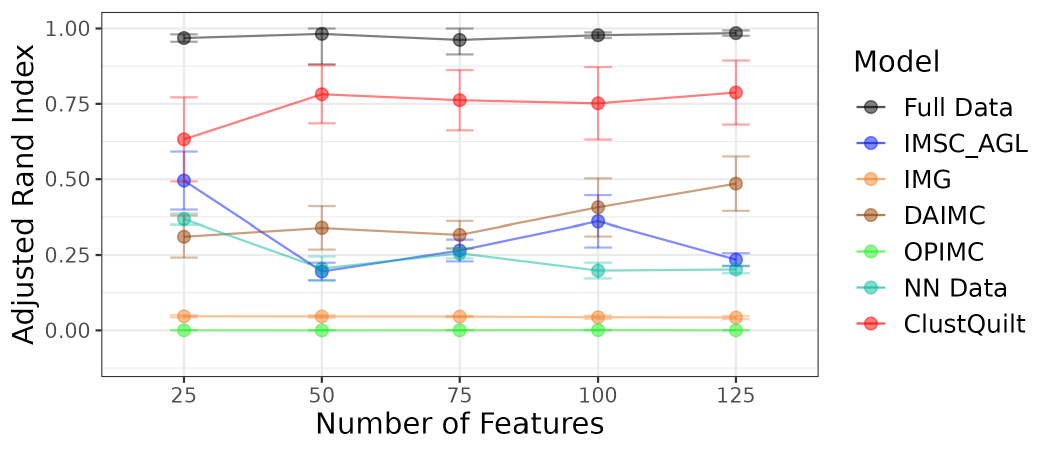}
        \caption{Number of total features.}
    \label{fig:supp_gmm_p}
    \end{subfigure}%
    \begin{subfigure}[t]{0.4\linewidth}
    \centering
        \includegraphics[width=\linewidth]{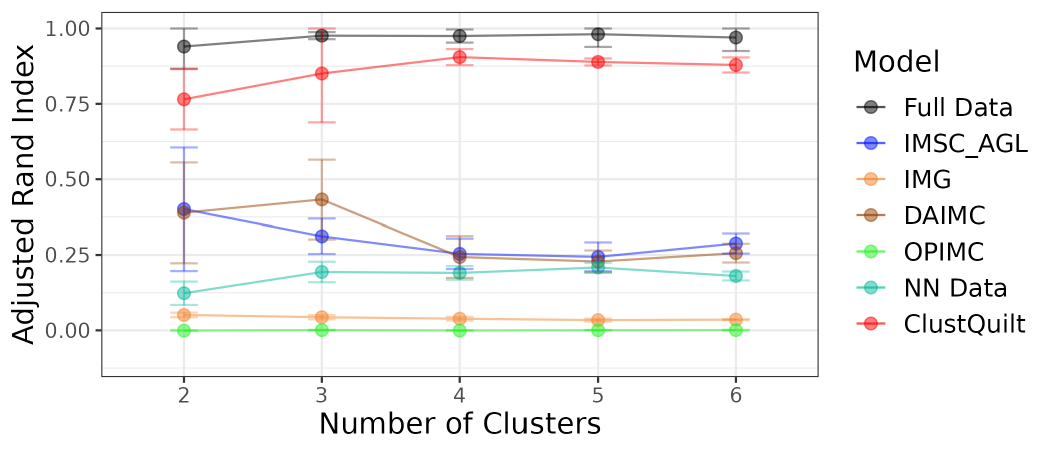}
        \caption{Number of clusters.}
    \label{fig:supp_gmm_k}
    \end{subfigure}
    \begin{subfigure}[t]{0.4\linewidth}
    \centering
        \includegraphics[width=\linewidth]{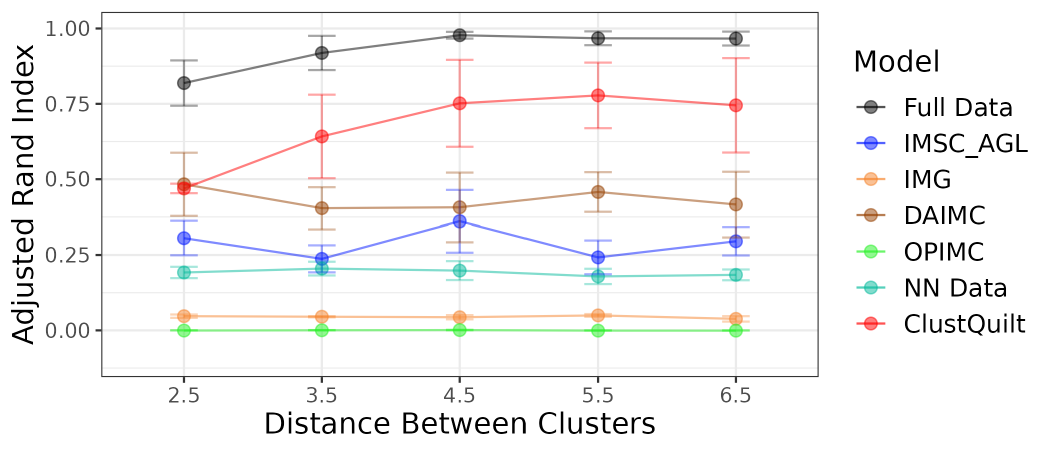}
        \caption{Cluster distance.}
    \label{fig:supp_gmm_mu}
    \end{subfigure}%
    \begin{subfigure}[t]{0.4\linewidth}
    \centering
        \includegraphics[width=\linewidth]{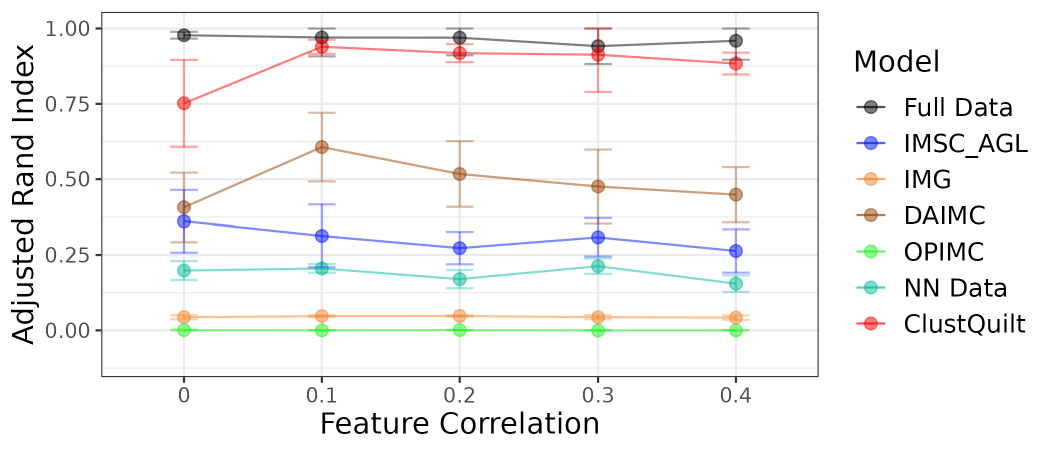}
        \caption{Feature autocorrelation.}
    \label{fig:supp_gmm_rho}
    \end{subfigure}
\end{center}
    \caption{Performance of Cluster Quilting and comparison cluster imputation methods with oracle tuning on Gaussian mixture model data in the sequential patch setting for different simulation parameters, measured by adjusted Rand index.}
    \label{fig:supp_gmm}
\end{figure}

The results of applying Cluster Quilting and the comparative cluster imputation methods are shown in Figure \ref{fig:supp_gmm}. In general, we see that the Cluster Quilting algorithm outperforms the other competing clustering imputation methods across all of the different variations of the simulation parameters, both when using oracle as well as data-driven hyperparameter tuning. Clustering on the similarity matrix with missing values imputed by the block singular value decomposition method yields fairly accurate results as well, close to the performance level of Cluster Quilting, while the other comparison methods do not recover the underlying clusters well in any scenario. Across the different sets of simulation parameters, we observe several expected trends in terms of the impact on accuracy. When the size of each observation block is decreased, either through directly decreasing the number of rows per block (Figure \ref{fig:supp_gmm_m}) or increasing the number of blocks while holding the number of blocks multiplied by the number of rows per block constant (Figure \ref{fig:supp_gmm_o}). On the other hand, as the number of features is increased, all methods perform better (Figure \ref{fig:supp_gmm_p}); this can likely be attributed to having more information from having more data to use for estimation, including an increase in the size of the overlap between each pair of observation blocks. When we increase the number of clusters (Figure \ref{fig:supp_gmm_k}) or decrease the distance between clusters (Figure \ref{fig:supp_gmm_mu}) as measured by the maximum difference between feature means across different clusters, we see that all methods become less accurate, which matches what we would expect as this should make the clustering problem more difficult in general. Also, we find that accuracy of all methods is relatively consistent across all levels of autocorrelation (Figure \ref{fig:supp_gmm_rho}). In Section 2 in the Supporting Information, we include additional simulation results for patchwork observed data from a Gaussian mixture model for the high-dimensional case, as well as a non-Gaussian simulation study.

\subsection{Data Driven Tuning}

We next present results for the same simulated data as was used for our empirical studies in Sections 3.2 in the main text for the mosaic patch observation setting and in Section \label{sub:bmgmm} for the sequential patch observation setting, but with data-driven tuning of the rank of the data matrix and the number of clusters. Here, hyperparameter selection is done using the prediction validation method of \cite{tibshirani2005cluster}, as outlined in Section 2.1 of the main text. We see that relative performance of all methods is consistent with the oracle tuning results, i.e. the Cluster Quilting method performs better than the comparison incomplete spectral clustering methods in most simulation settings. Additionally, in general, adjusted Rand Index scores across all methods decreases when compared to performance, showing that rank and cluster misspecification has a modest effect on cluster estimation accuracy.

\begin{figure}[t]
\begin{center}
    \begin{subfigure}[t]{0.4\linewidth}
    \centering
        \includegraphics[width=\linewidth]{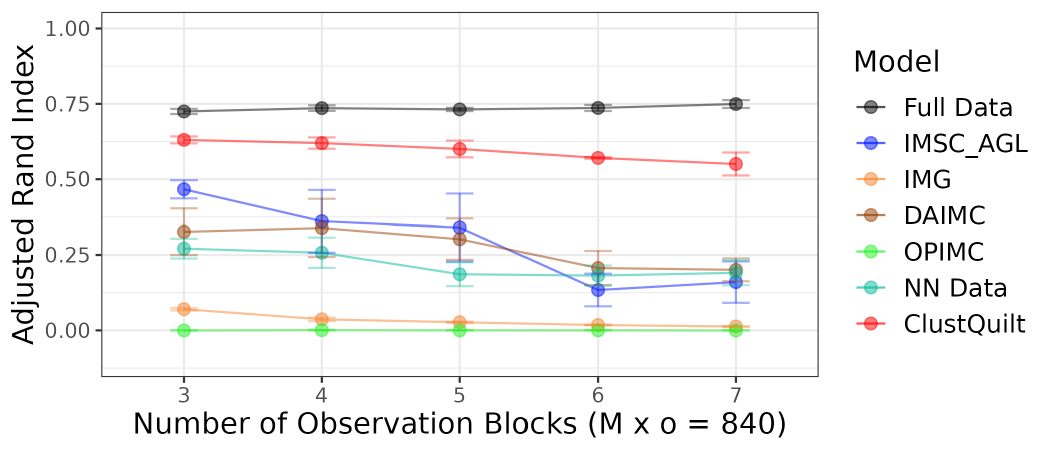}
        \caption{Changing number of blocks.}
    \label{fig:sup_gmm_m}
    \end{subfigure}%
    \begin{subfigure}[t]{0.4\linewidth}
    \centering
        \includegraphics[width=\linewidth]{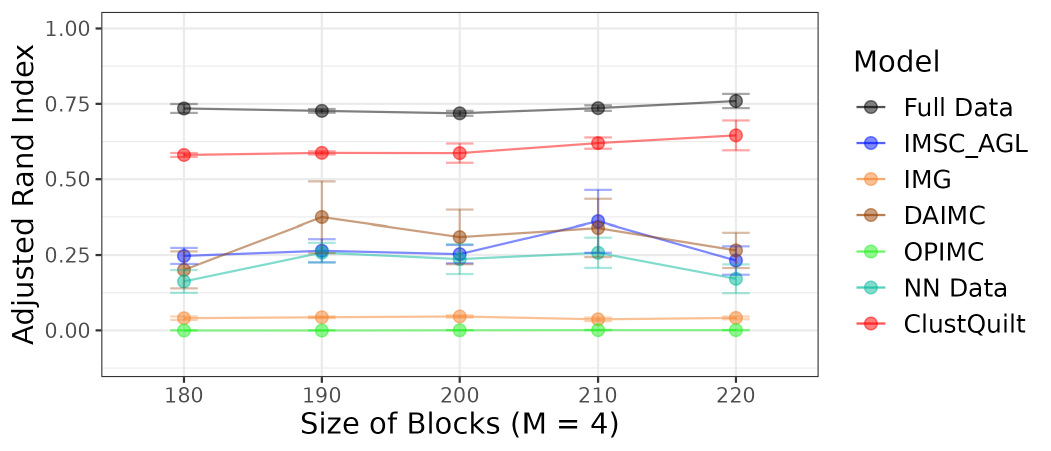}
        \caption{Changing size of blocks.}
    \label{fig:sup_gmm_o}
    \end{subfigure}
    \begin{subfigure}[t]{0.4\linewidth}
    \centering
        \includegraphics[width=\linewidth]{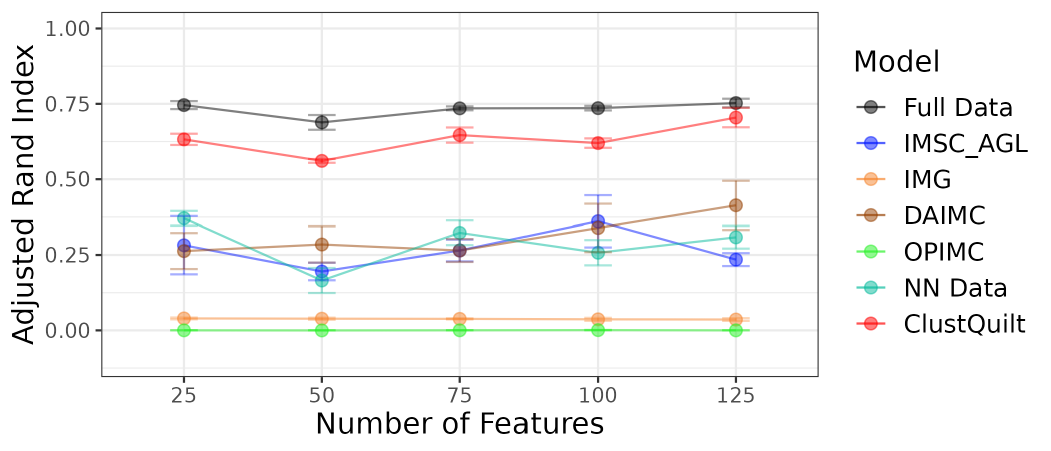}
        \caption{Changing number of total features.}
    \label{fig:sup_gmm_p}
    \end{subfigure}%
    \begin{subfigure}[t]{0.4\linewidth}
    \centering
        \includegraphics[width=\linewidth]{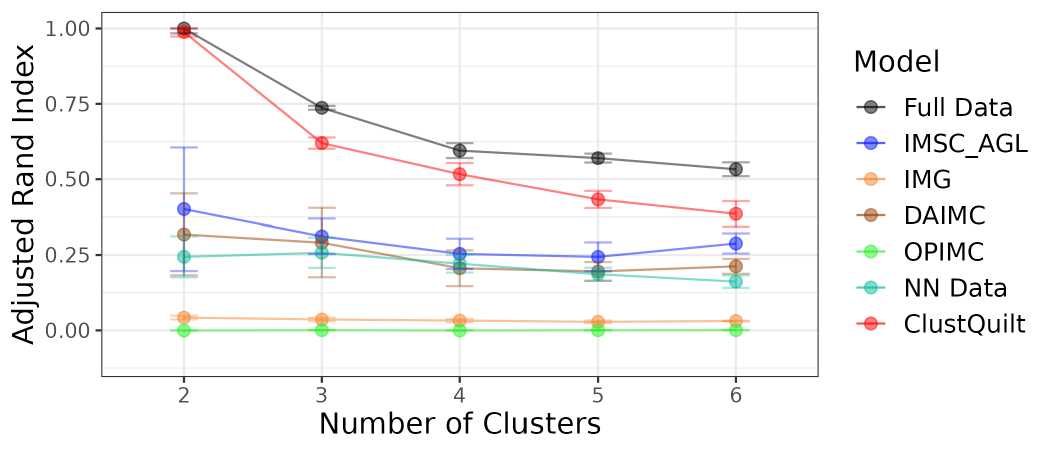}
        \caption{Changing number of  clusters.}
    \label{fig:sup_gmm_k}
    \end{subfigure}
    \begin{subfigure}[t]{0.4\linewidth}
    \centering
        \includegraphics[width=\linewidth]{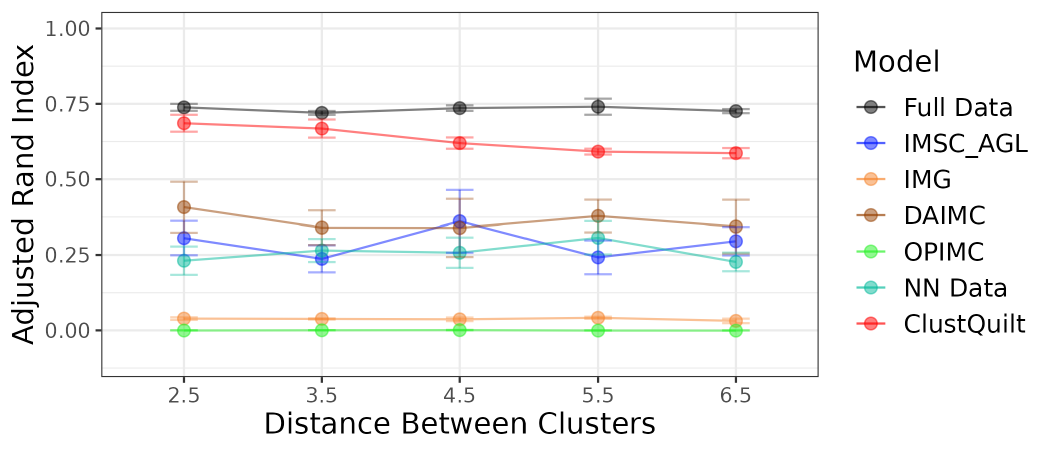}
        \caption{Changing cluster distance.}
    \label{fig:sup_gmm_mu}
    \end{subfigure}%
    \begin{subfigure}[t]{0.4\linewidth}
    \centering
        \includegraphics[width=\linewidth]{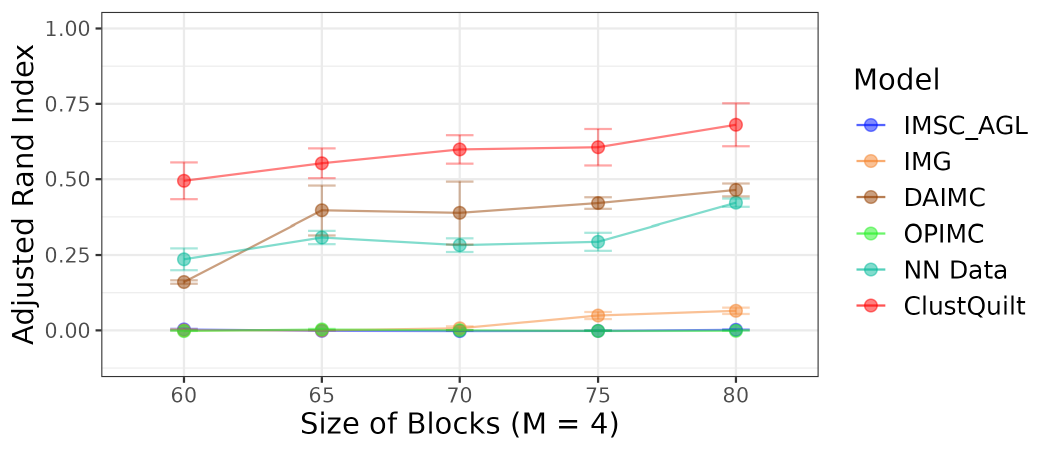}
        \caption{Autocorrelation between features.}
    \label{fig:sup_gmm_rho}
    \end{subfigure}
\end{center}
    \caption{Performance of Cluster Quilting and comparison cluster imputation methods with data-driven tuning on Gaussian mixture model data in the sequential patch setting, measured by adjusted Rand index.}
    \label{fig:sup_gmm}
\end{figure}

\begin{figure}[t]
\begin{center}
    \begin{subfigure}[t]{0.4\linewidth}
    \centering
        \includegraphics[width=\linewidth]{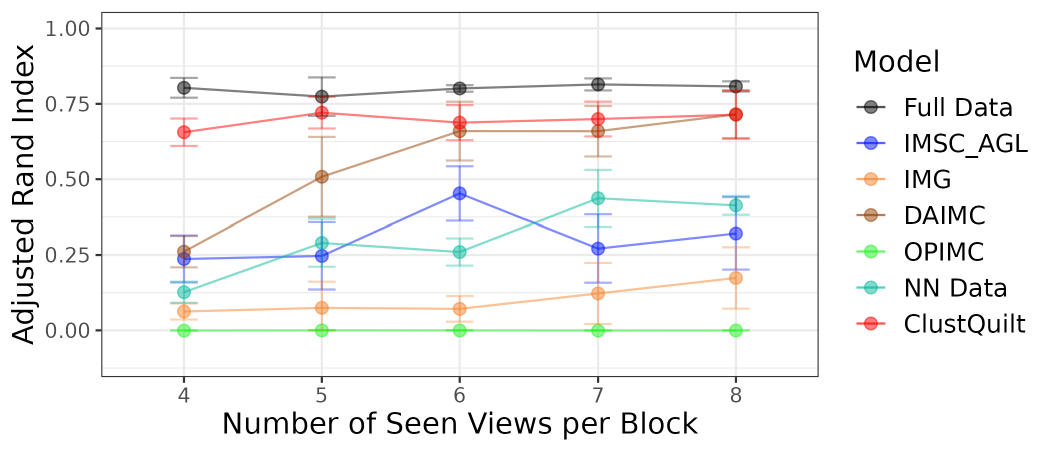}
        \caption{Changing views observed per block.}
    \label{fig:sup_pmv_h}
    \end{subfigure} \hspace{0.5cm}%
    \begin{subfigure}[t]{0.4\linewidth}
    \centering
        \includegraphics[width=\linewidth]{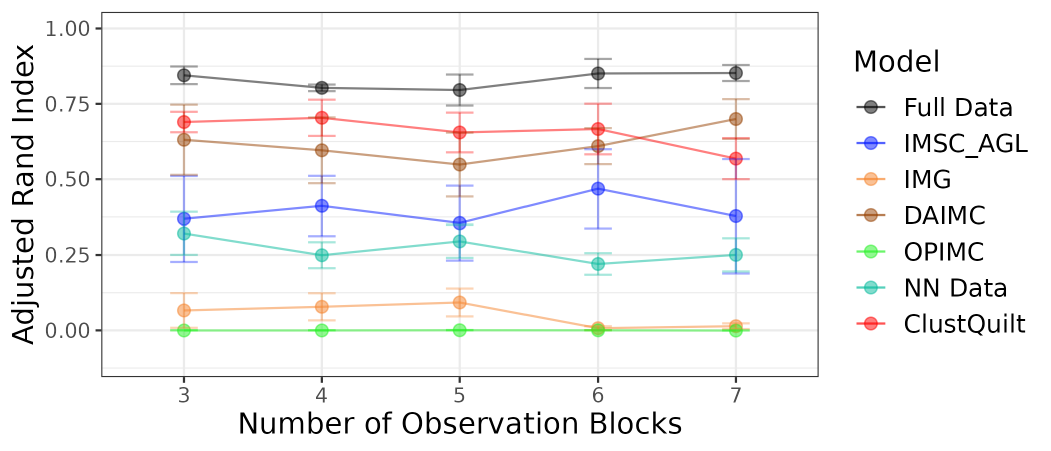}
        \caption{Changing number of blocks.}
    \label{fig:sup_pmv_b}
    \end{subfigure}\hspace{0.5cm} %
    \begin{subfigure}[t]{0.4\linewidth}
    \centering
        \includegraphics[width=\linewidth]{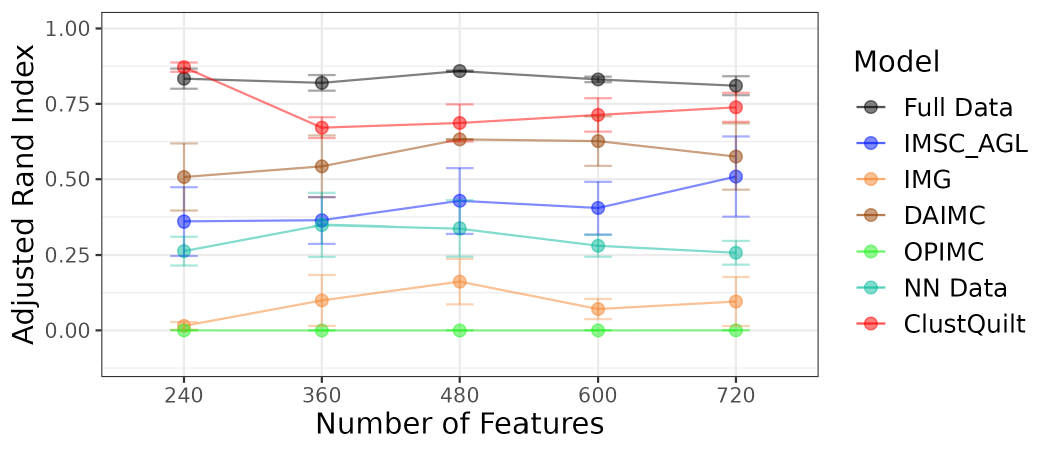}
        \caption{Changing number of total features.}
    \label{fig:sup_pmv_p}
    \end{subfigure}\hspace{0.5cm}%
    \begin{subfigure}[t]{0.4\linewidth}
    \centering
        \includegraphics[width=\linewidth]{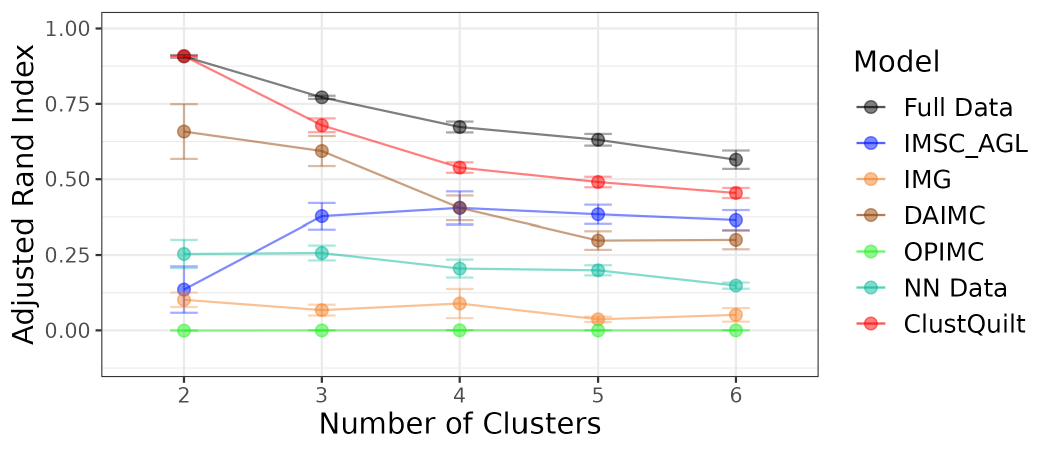}
        \caption{Changing number of true clusters.}
    \label{fig:sup_pmv_k}
    \end{subfigure}\hspace{0.5cm}%
    \begin{subfigure}[t]{0.4\linewidth}
    \centering
        \includegraphics[width=\linewidth]{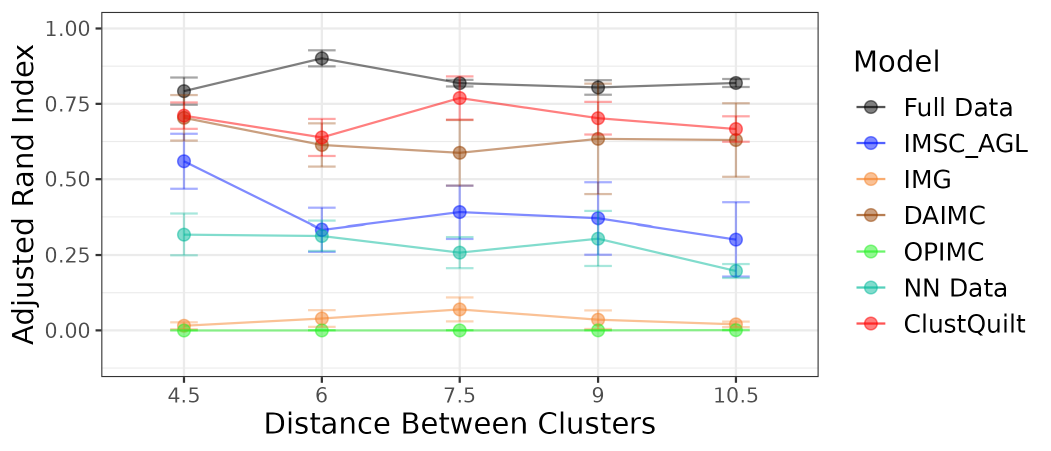}
        \caption{Changing cluster distance.}
    \label{fig:sup_pmv_d}
    \end{subfigure}\hspace{0.5cm}%
    \begin{subfigure}[t]{0.4\linewidth}
    \centering
        \includegraphics[width=\linewidth]{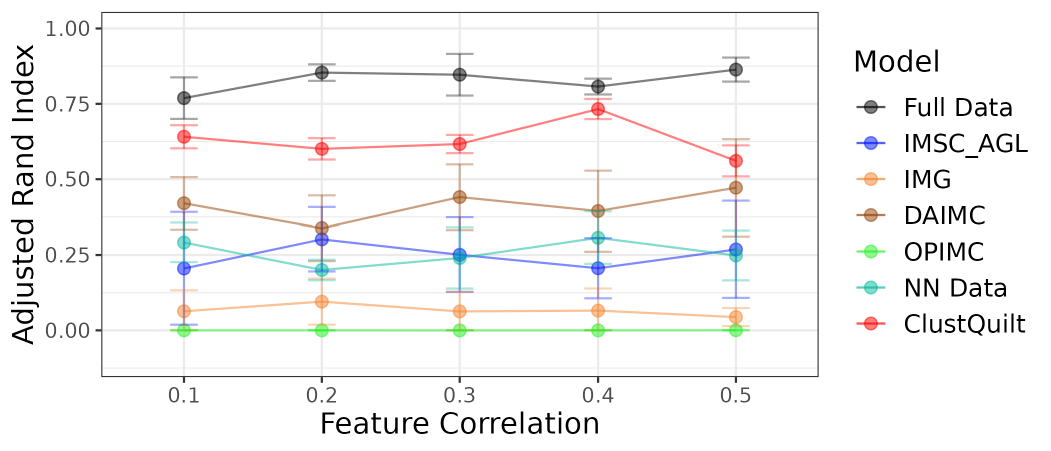}
        \caption{Autocorrelation between features.}
    \label{fig:sup_pmv_rho}
    \end{subfigure}
\end{center}
    \caption{Performance of Cluster Quilting and comparison cluster imputation methods with data-driven tuning on Gaussian mixture model data in the mosaic patch setting, measured by adjusted Rand index.}
    \label{fig:sup_pmv}
\end{figure}

\clearpage

\subsection{High Dimensional Studies}

Here, we provide additional simulation studies in a high-dimensional setting. For the sequential patch observation setting, the simulated data contain 710 observations spread across three true underlying clusters; 4 synthetic observation blocks of 210 observations are then created from the data. In Figures \ref{fig:sup_pmv_high} and \ref{fig:sup_pmv_high_dd} we present results when changing the number of features in the data between 1000 and 2000, as opposed to the 100 features in the simulations in the main text. As we would expect from the pattern from the main text simulations, we find that the increasing the number of features is correlated with increased accuracy of cluster estimates, although there seems to be an upper limit of the performance overall at an adjusted Rand Index of around 0.7.

For the mosaic patch observation setting in the high-dimensional setting, we generate a full data matrix of 840 observations split into 3 underlying cluster, and 12 views of 50 features each. The cluster centroid matrix used to generate the data for each view is a rank 2 matrix with $d = 7.5$. We then create 4 observation blocks, for which 6 views per observation block are selected to be part of the set of observed entries. Figures \ref{fig:sup_gmm_high} and \ref{fig:sup_gmm_high_dd} show the results for varying the number of features per view between 100 and 500. Overall, we see that the Cluster Quilting method performs better than any of the other incomplete spectral clustering methods in the simulation study across the varying simulation parameters.

\begin{figure}[t]
\begin{center}
    \begin{subfigure}[t]{0.4\linewidth}
    \centering
        \includegraphics[width=\linewidth]{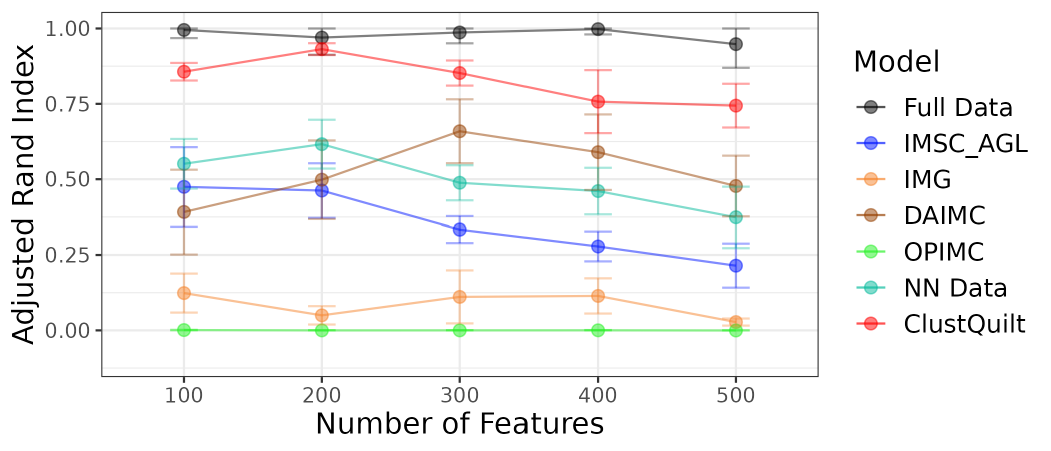}
        \caption{Sequential patch data, \\ oracle tuning.}
    \label{fig:sup_pmv_high}
    \end{subfigure}%
    \begin{subfigure}[t]{0.4\linewidth}
    \centering
        \includegraphics[width=\linewidth]{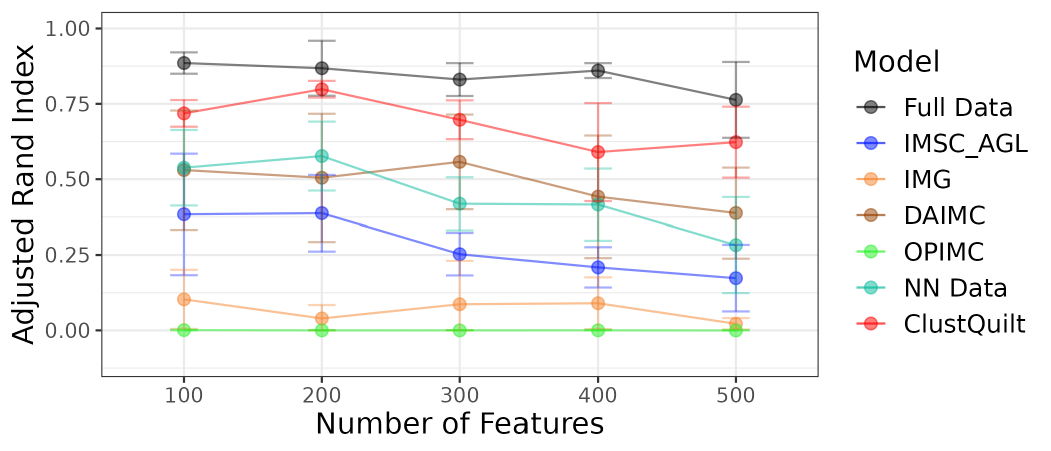}
        \caption{Sequential patch data, \\  data-driven tuning.}
    \label{fig:sup_pmv_high_dd}
    \end{subfigure} 
    \begin{subfigure}[t]{0.4\linewidth}
    \centering
        \includegraphics[width=\linewidth]{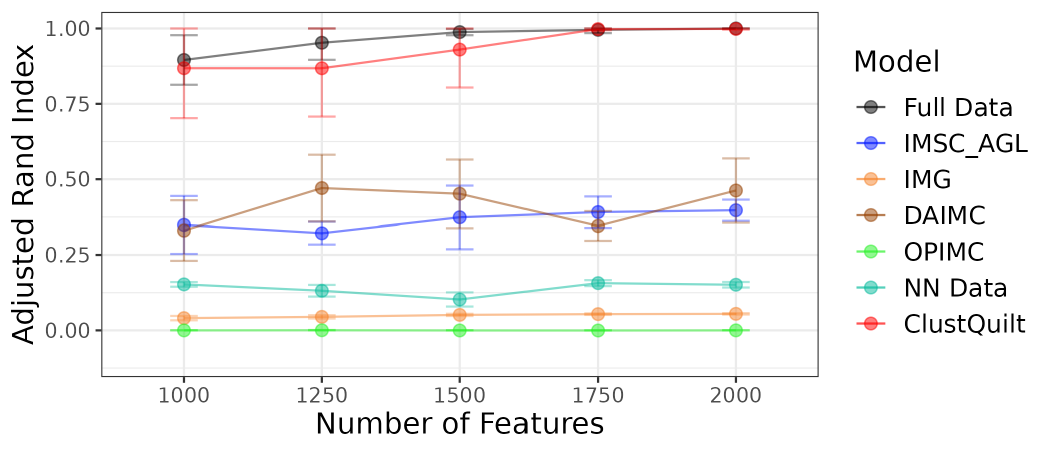}
        \caption{General patchwork data, \\ oracle tuning.}
    \label{fig:sup_gmm_high}
    \end{subfigure}%
    \begin{subfigure}[t]{0.4\linewidth}
    \centering
        \includegraphics[width=\linewidth]{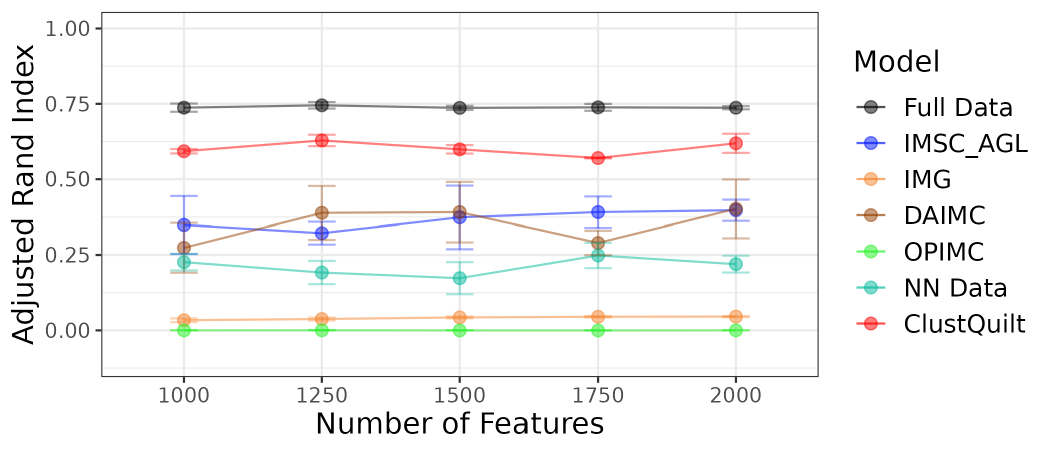}
        \caption{General patchwork data, \\  data-driven tuning.}
    \label{fig:sup_gmm_high_dd}
    \end{subfigure}
\end{center}
    \caption{Performance of Cluster Quilting and comparison cluster imputation methods on Gaussian mixture model data with different missingness patterns in a high-dimensional setting, measured by Adjusted Rand index with respect to the underlying cluster membership.}
    \label{fig:sup_gmm2}
\end{figure}

\clearpage
\subsection{t-Distribution Mixture Model, Sequential Patch Observation}

We now test the Cluster Quilting and comparison methods in the sequential patch observation setting on a mixture model with a heavy tailed non-sub-Gaussian distribution. To do this, we generate data from the Gaussian mixture model as described in Section 5a of the main text, then apply a copula transform to each univariate marginal distribution to a t-distribution with 5 degrees of freedom. We then create the sequential patch observation pattern by partitioning the data matrix into $M$ blocks with $o$ rows, with overlap between the rows included in consecutive observation blocks, and $p/M$ non-overlapping columns. We then estimate clusters from partially observed data matrix and measure performance against the objective of recovering the cluster labels used to generate the original Gaussian mixture model. As baseline simulation parameters, we set mixture model to have 710 observations of 100 features in 3 clusters. 4 synthetic observation blocks are then created from the data, each containing 210 observations. As in the main text, we study the effects of changing the number of observed blocks, the number of observations in each block, the number of features, the number of underlying clusters, and the distance between clusters. 

The simulation results are shown in Figure \ref{fig:t}. Similar to the Gaussian mixture model simulations, the Cluster Quilting method demonstrates superior performance across almost all of the different simulation parameter settings, while most of the other methods tend to be substantially less accurate. In terms of the effects of varying simulation parameters on performance, we also observe similar patterns as in the Gaussian mixture model results, in which increasing the number of features, size of blocks, and distance between clusters improves the accuracy of all methods, while increasing the number of blocks between which the observations are divided negatively impacts recovery of cluster labels. One interesting deviation from previous results is that, for the t-distribution mixture model, increasing the number of clusters while holding the number of total rows in the data constant often increases the accuracy of all methods. This may be due to the heavy-tailed shape of the t-distribution, which increases the probability of an observations in one cluster to be incorrectly classified empirically due to being closer to points in a different cluster which does not match the true label.

\begin{figure}[t]
\begin{center}
    \begin{subfigure}[t]{0.4\linewidth}
    \centering
        \includegraphics[width=\linewidth]{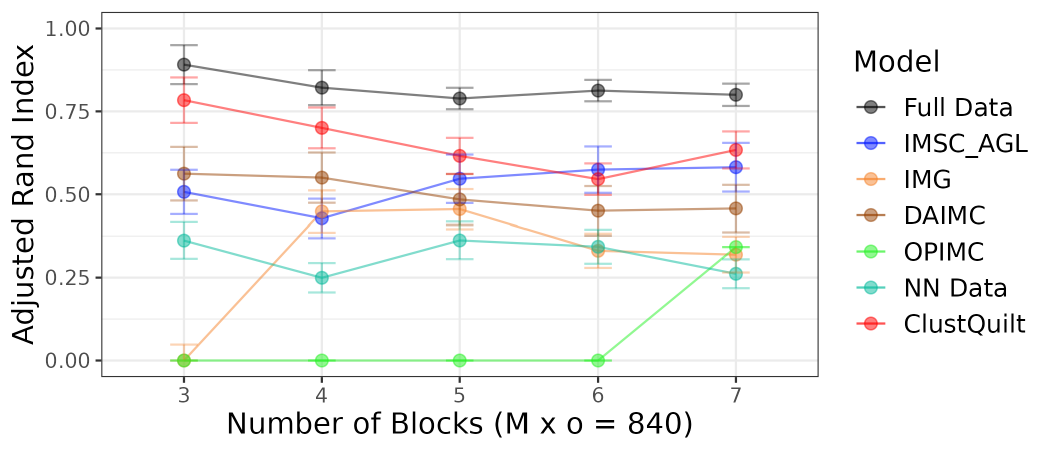}
        \caption{Changing number of blocks; \\ oracle tuning.}
    \label{fig:t_m}
    \end{subfigure}%
    \begin{subfigure}[t]{0.4\linewidth}
    \centering
        \includegraphics[width=\linewidth]{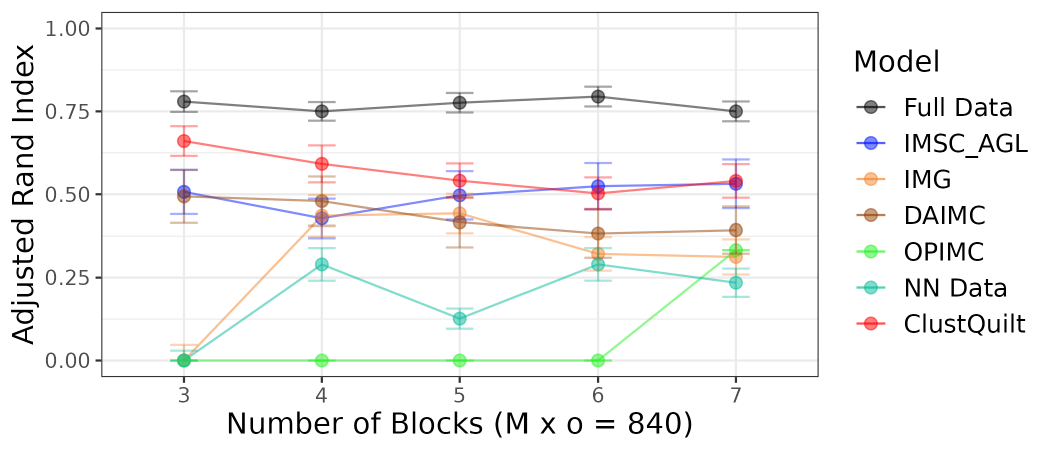}
        \caption{Changing number of blocks; \\ data-driven tuning.}
    \label{fig:t_m_dd}
    \end{subfigure}
    \begin{subfigure}[t]{0.4\linewidth}
    \centering
        \includegraphics[width=\linewidth]{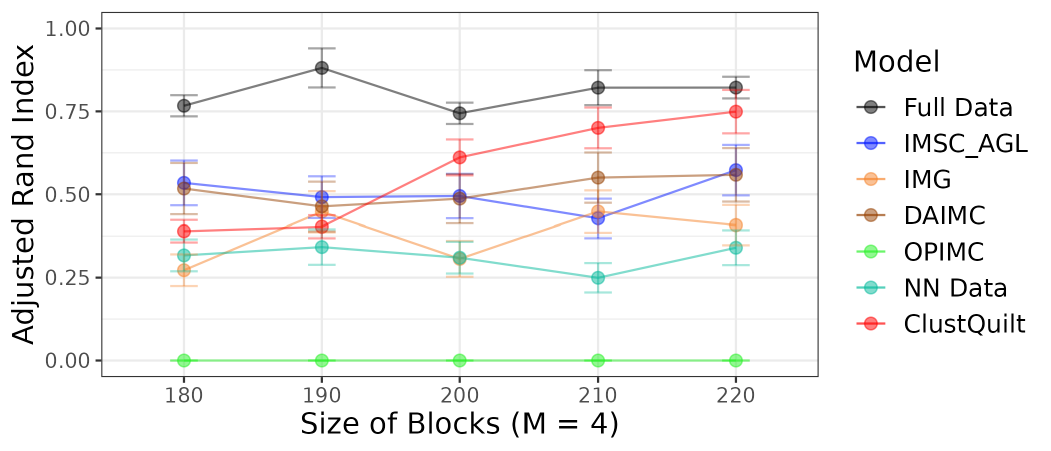}
        \caption{Changing size of observed blocks; \\ oracle tuning.}
    \label{fig:t_o}
    \end{subfigure}%
    \begin{subfigure}[t]{0.4\linewidth}
    \centering
        \includegraphics[width=\linewidth]{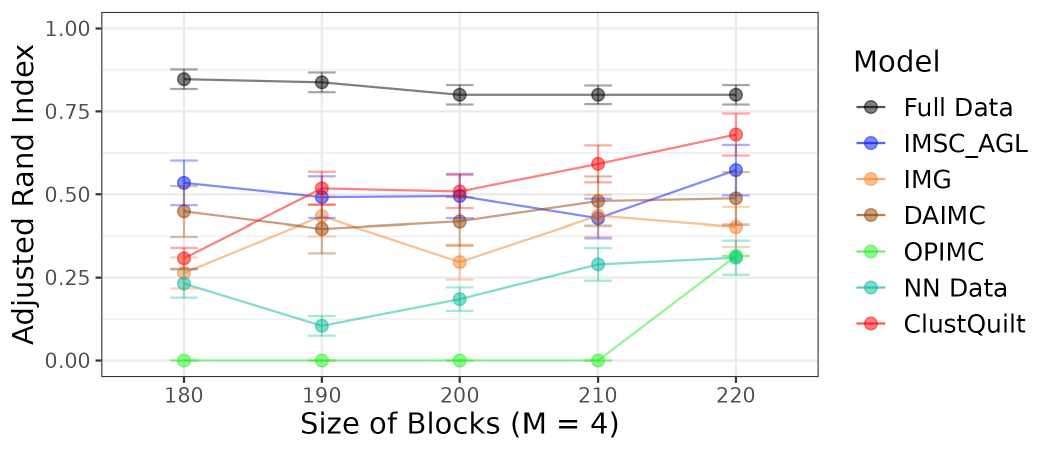}
        \caption{Changing size of observed blocks; \\ data-driven tuning.}
    \label{fig:t_o_dd}
    \end{subfigure}
    \begin{subfigure}[t]{0.4\linewidth}
    \centering
        \includegraphics[width=\linewidth]{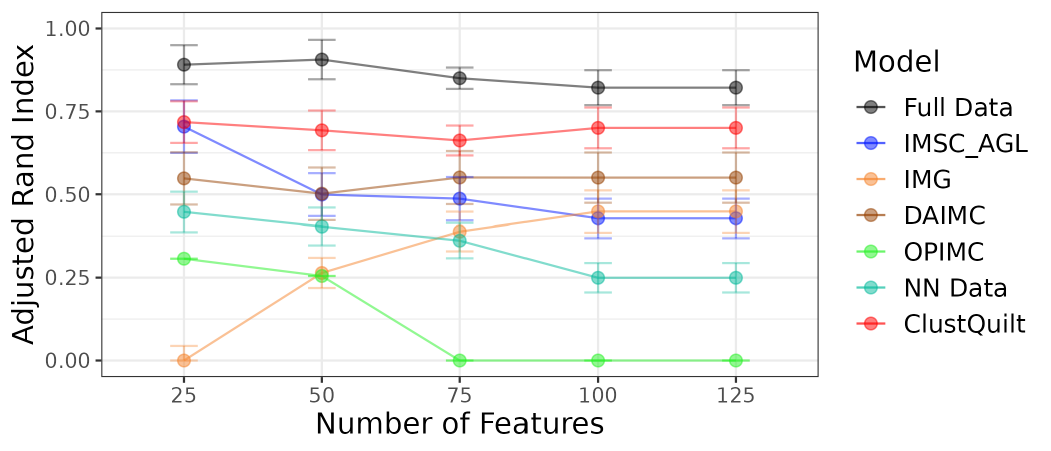}
        \caption{Changing number of total features; \\ oracle tuning.}
    \label{fig:t_p}
    \end{subfigure}%
    \begin{subfigure}[t]{0.4\linewidth}
    \centering
        \includegraphics[width=\linewidth]{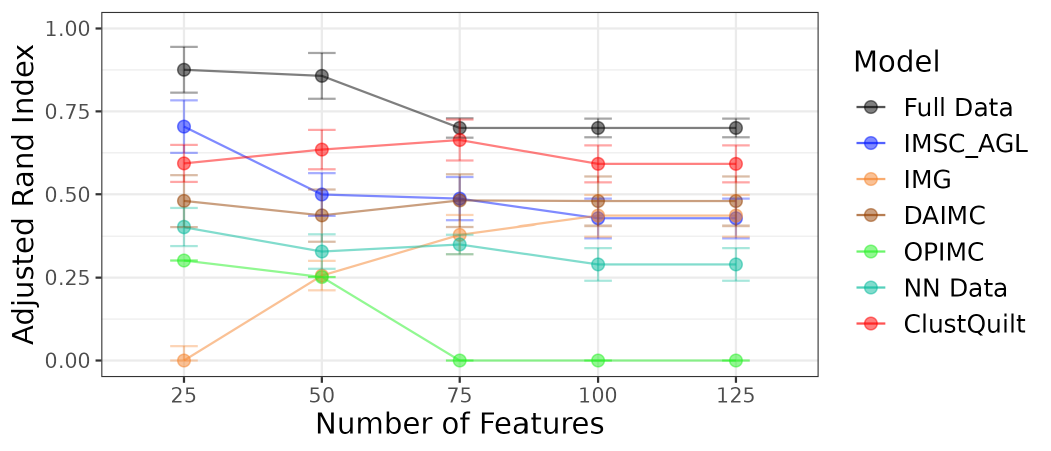}
        \caption{Changing number of total features; \\ data-driven tuning.}
    \label{fig:t_p_dd}
    \end{subfigure}
    \begin{subfigure}[t]{0.4\linewidth}
    \centering
        \includegraphics[width=\linewidth]{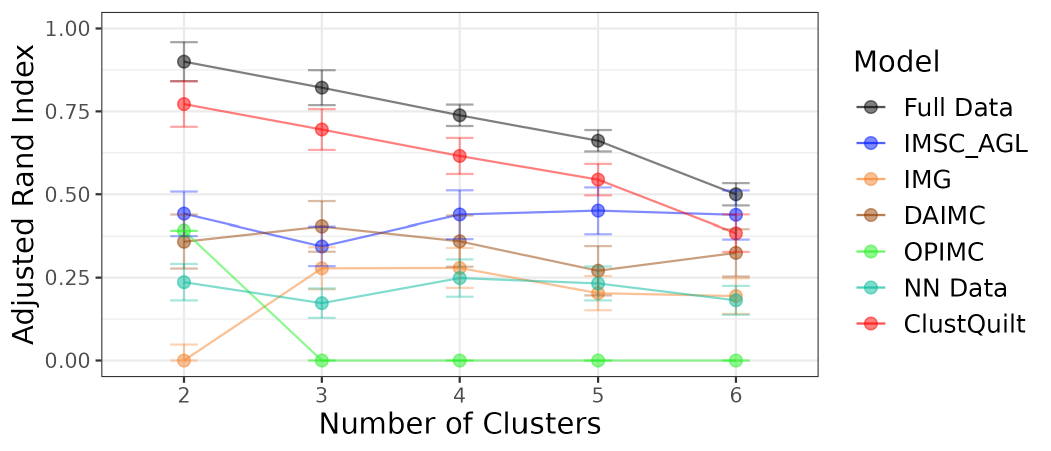}
        \caption{Changing number of underlying clusters; \\ oracle tuning.}
    \label{fig:t_k}
    \end{subfigure}%
    \begin{subfigure}[t]{0.4\linewidth}
    \centering
        \includegraphics[width=\linewidth]{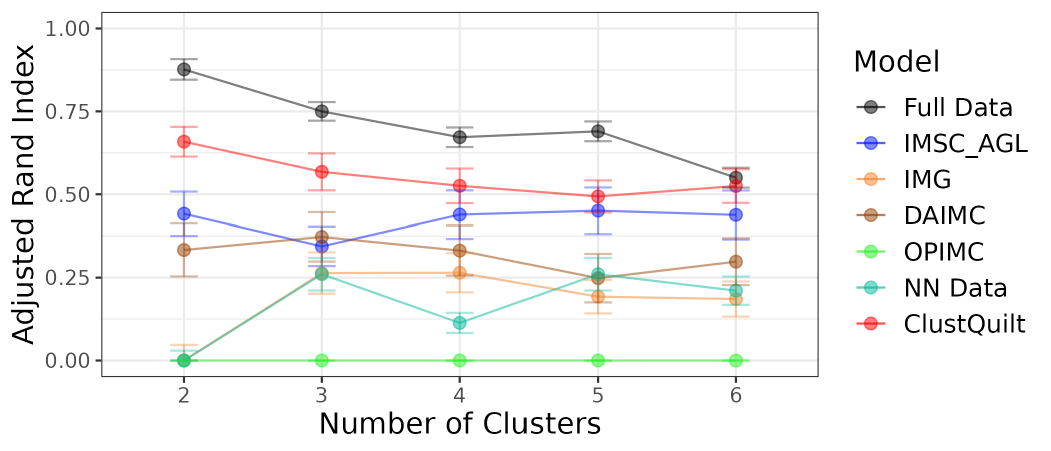}
        \caption{Changing number of underlying clusters; \\ data-driven tuning.}
    \label{fig:t_k_dd}
    \end{subfigure}
    \begin{subfigure}[t]{0.4\linewidth}
    \centering
        \includegraphics[width=\linewidth]{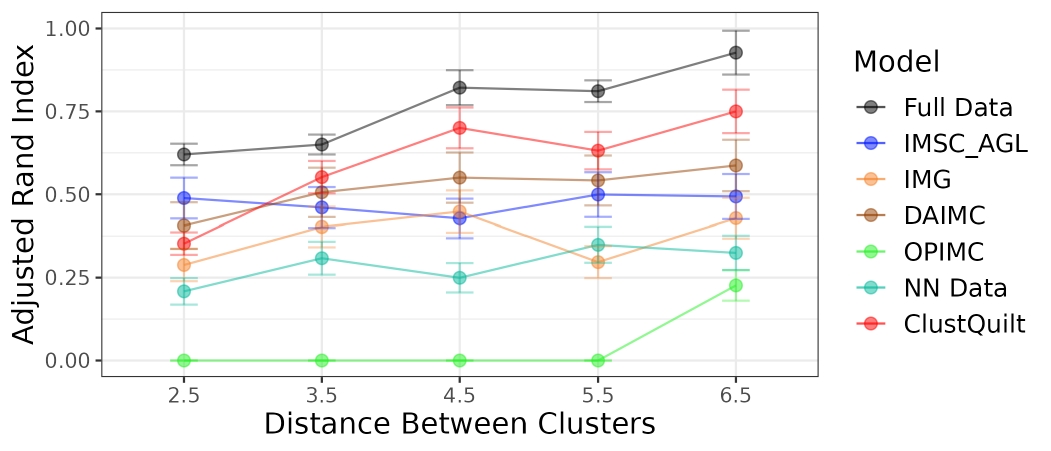}
        \caption{Changing cluster distance; \\ oracle tuning.}
    \label{fig:t_mu}
    \end{subfigure}%
    \begin{subfigure}[t]{0.4\linewidth}
    \centering
        \includegraphics[width=\linewidth]{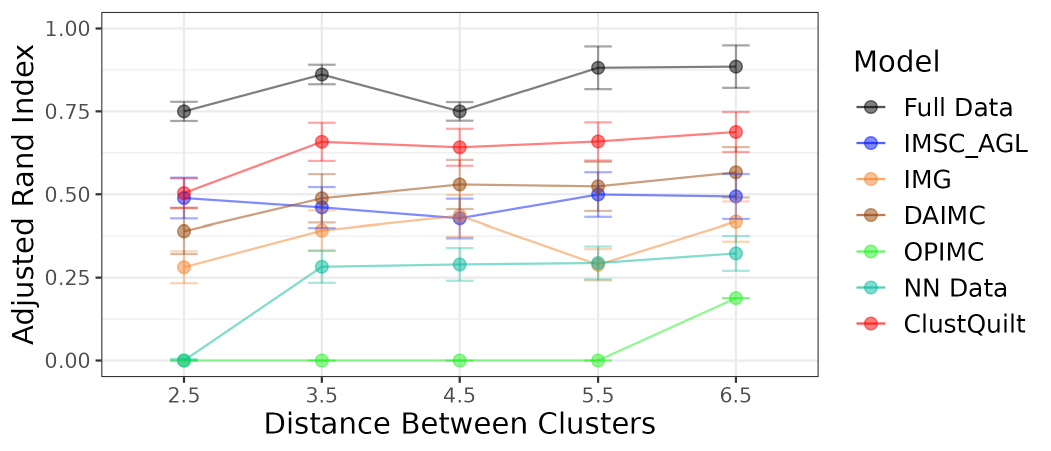}
        \caption{Changing cluster distance; \\ data-driven tuning.}
    \label{fig:t_mu_dd}
    \end{subfigure}
\end{center}
    \caption{Performance of Cluster Quilting and comparison cluster imputation methods on sequential patch observe data from a multivariate t-distribution with 5 degrees of freedom, measured by Adjusted Rand index with respect to the underlying cluster membership.}
    \label{fig:t}
\end{figure}

\clearpage

\subsection{Copula Mixture Model, Mosaic Patch Observation}

We also evaluate Cluster Quilting and other incomplete spectral clustering methods in the mosaic patch observation setting for data generated from a copula mixture model in which we set each view to follow a different multivariate distribution. To do this, we generate data from Gaussian mixture models from each view as described in Section 5b of the main text, then apply a copula transform to each column. We then create the mosaic patch observation pattern by randomly segmenting the rows of the full data matrix into $M$ blocks of observations and then randomly selecting $h$ views for each observation block to leave in the set of observed entries. We use default parameters of 840 observations, 3 underlying clusters, 12 views, and 50 features per view. The cluster centroid matrix is a rank 2 matrix with $d = 7.5$. After data generation, we randomly divide the rows into 4 blocks of equal size and select 6 views to be part of the set of observed entries. For this simulation, each view is randomly selected to follow either a skew normal, Gaussian, generalized normal, or beta distribution, with randomized shape parameters; all columns within a view follow the same distribution. 

We show the results of the simulation study in Figure \ref{fig:sup_pmvc}. In general, the Cluster Quilting method is able to more accurately recover the true underlying cluster labels used to generate the copula mixture model compared to the other incomplete spectral clustering methods for both oracle and data-driven hyperparameter tuning. The notable exceptions to this are when the number of observation blocks is relatively large, which follows what is presented in the theoretical results of the main text, and when the proportion of the data matrix that is observed is relatively small, when other methods tend to achieve better cluster recovery. When varying the simulation parameters, we find that all methods tend to perform better when the cluster centroids are further apart and when more data is observed; this matches the results from the Gaussian mixture model simulations in Section 5.

\begin{figure}[t]
\begin{center}
    \begin{subfigure}[t]{0.4\linewidth}
    \centering
        \includegraphics[width=\linewidth]{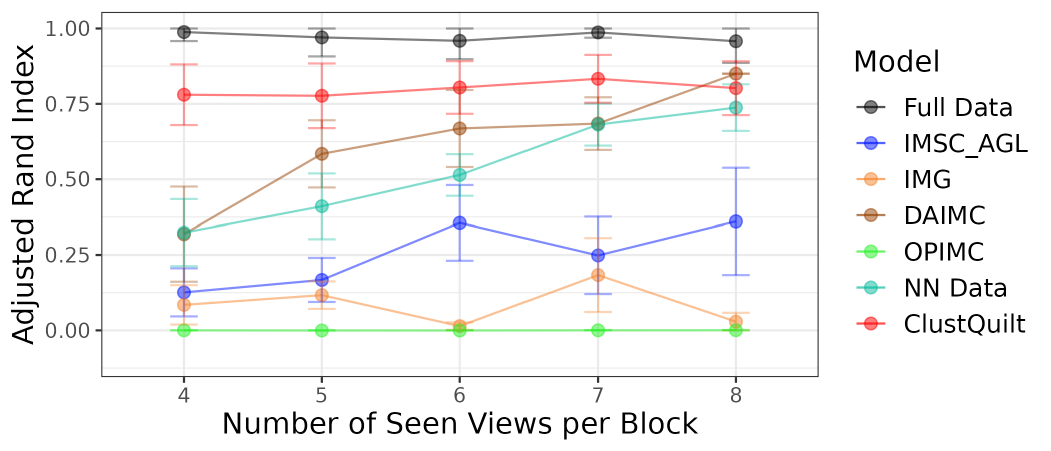}
        \caption{Changing views observed per block; \\ oracle tuning.}
    \label{fig:sup_pmvc_h}
    \end{subfigure} \hspace{0.5cm}%
    \begin{subfigure}[t]{0.4\linewidth}
    \centering
        \includegraphics[width=\linewidth]{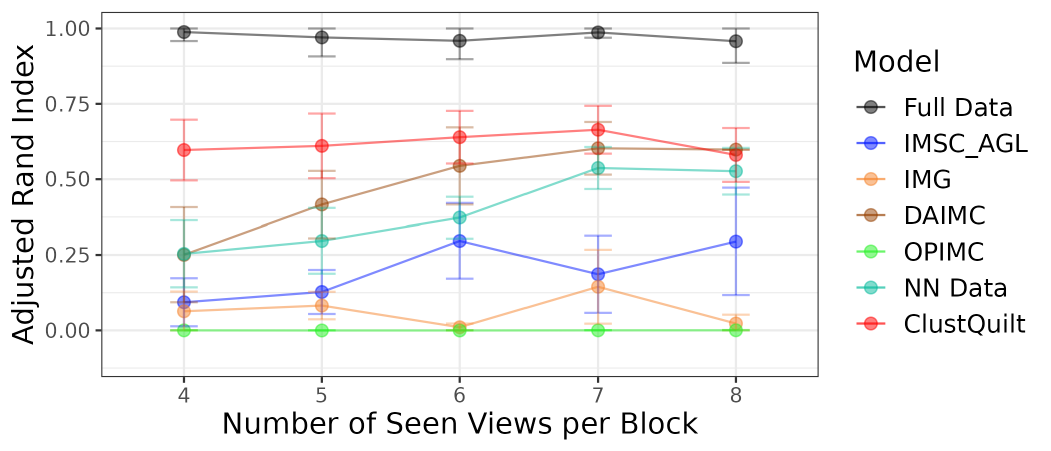}
        \caption{Changing views observed per block; \\ data-driven tuning.}
    \label{fig:sup_pmvc_h_dd}
    \end{subfigure}
    \begin{subfigure}[t]{0.4\linewidth}
    \centering
        \includegraphics[width=\linewidth]{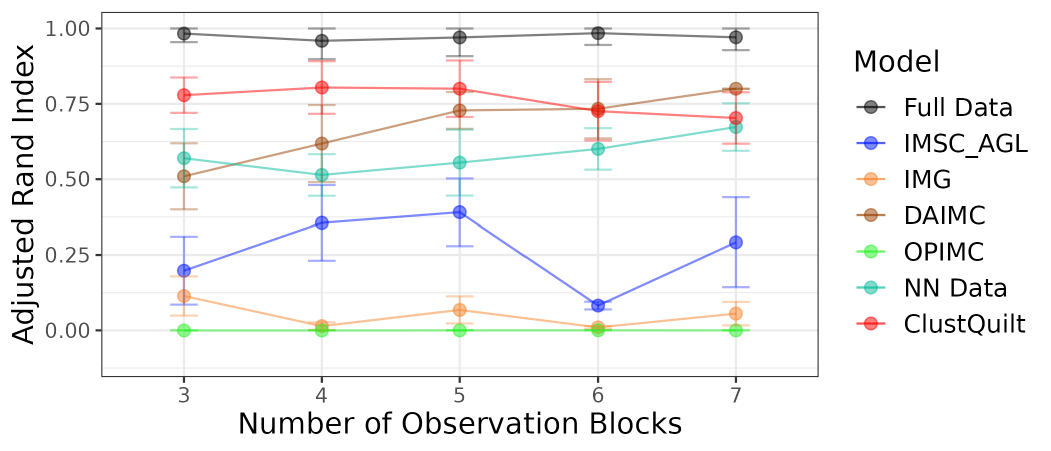}
        \caption{Changing number of blocks; \\ oracle tuning.}
    \label{fig:sup_pmvc_b}
    \end{subfigure}\hspace{0.5cm} %
    \begin{subfigure}[t]{0.4\linewidth}
    \centering
        \includegraphics[width=\linewidth]{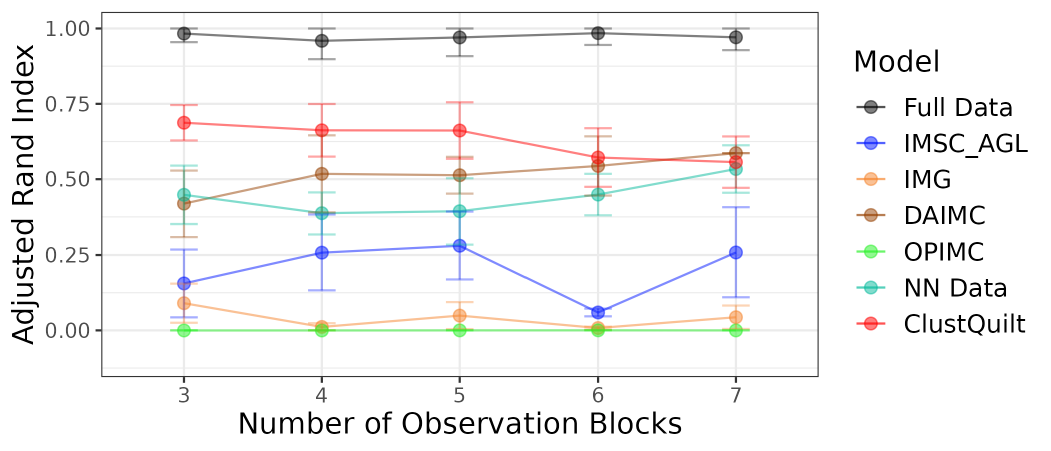}
        \caption{Changing number of blocks; \\ data-driven tuning.}
    \label{fig:sup_pmvc_b_dd}
    \end{subfigure}
    \begin{subfigure}[t]{0.4\linewidth}
    \centering
        \includegraphics[width=\linewidth]{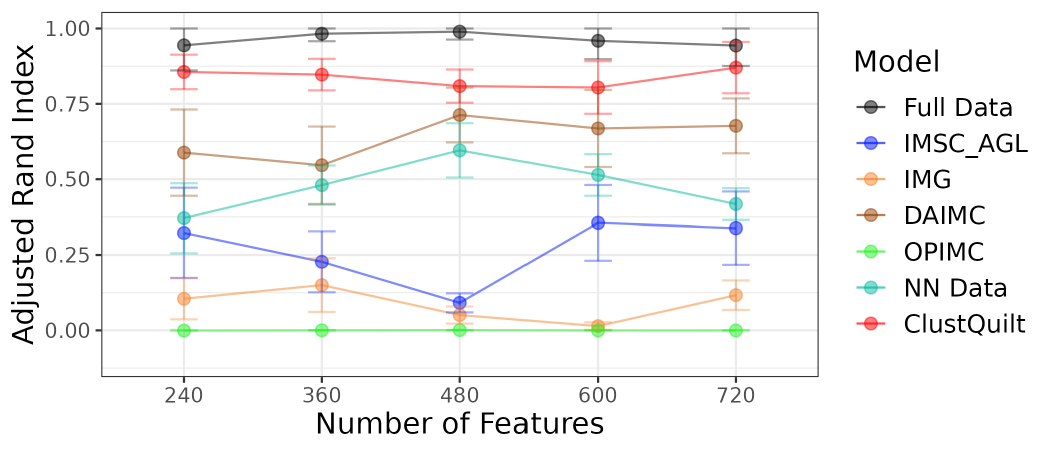}
        \caption{Changing number of total features; \\ oracle tuning.}
    \label{fig:sup_pmvc_p}
    \end{subfigure}\hspace{0.5cm}%
    \begin{subfigure}[t]{0.4\linewidth}
    \centering
        \includegraphics[width=\linewidth]{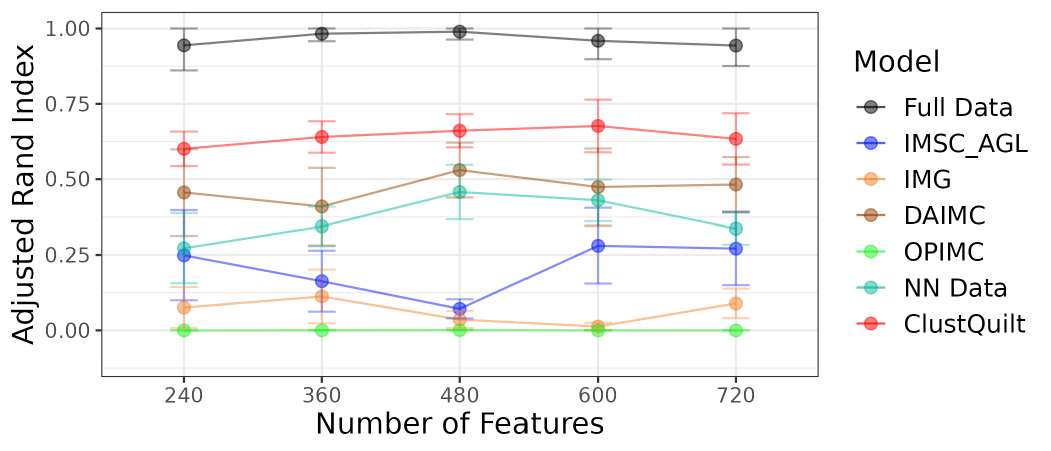}
        \caption{Changing number of total features; \\ data-driven tuning.}
    \label{fig:sup_pmvc_p_dd}
    \end{subfigure}
    \begin{subfigure}[t]{0.4\linewidth}
    \centering
        \includegraphics[width=\linewidth]{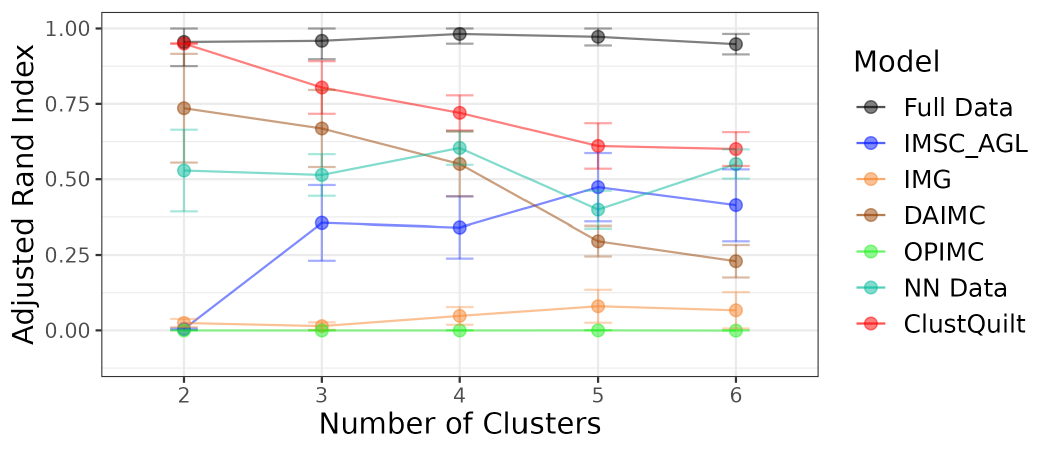}
        \caption{Changing number of true clusters; \\ oracle tuning.}
    \label{fig:sup_pmvc_k}
    \end{subfigure}\hspace{0.5cm}%
    \begin{subfigure}[t]{0.4\linewidth}
    \centering
        \includegraphics[width=\linewidth]{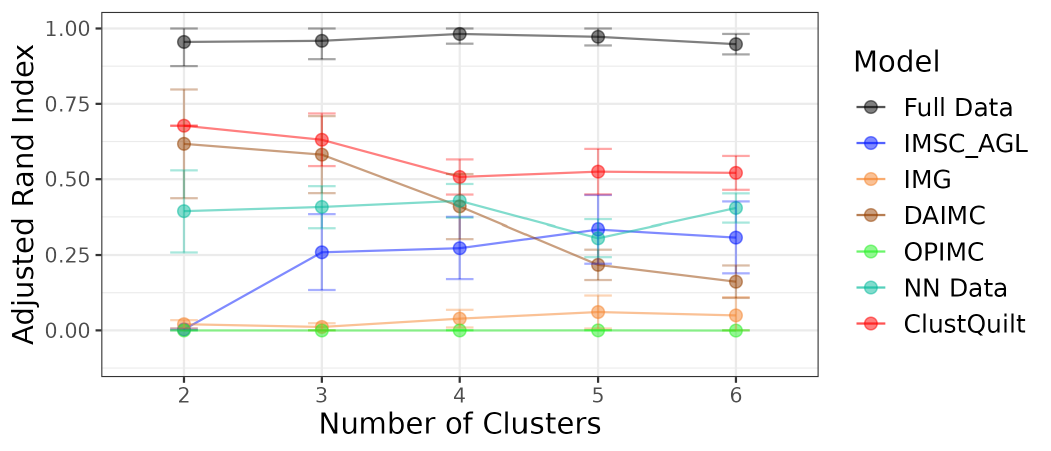}
        \caption{Changing number of true clusters; \\ data-driven tuning.}
    \label{fig:sup_pmvc_k_dd}
    \end{subfigure}
    \begin{subfigure}[t]{0.4\linewidth}
    \centering
        \includegraphics[width=\linewidth]{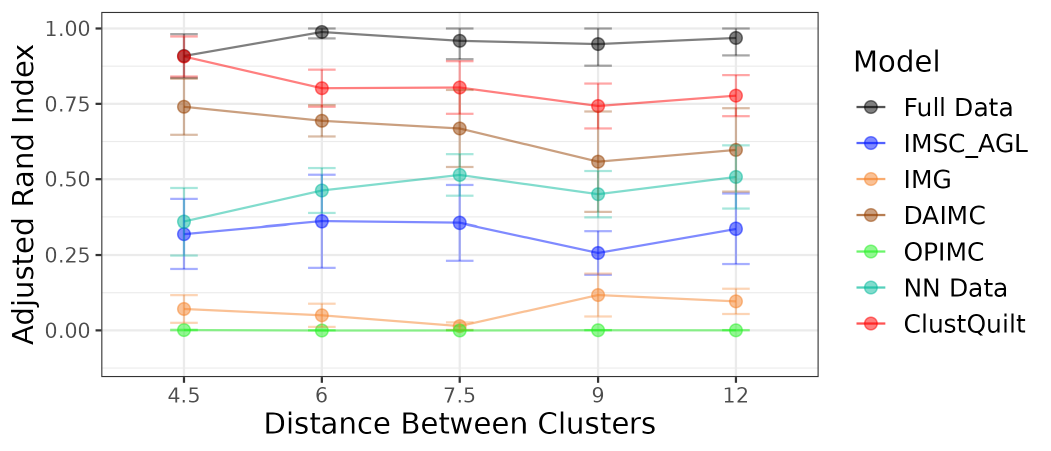}
        \caption{Changing cluster distance; \\ oracle tuning.}
    \label{fig:sup_pmvc_d}
    \end{subfigure}\hspace{0.5cm}%
    \begin{subfigure}[t]{0.4\linewidth}
    \centering
        \includegraphics[width=\linewidth]{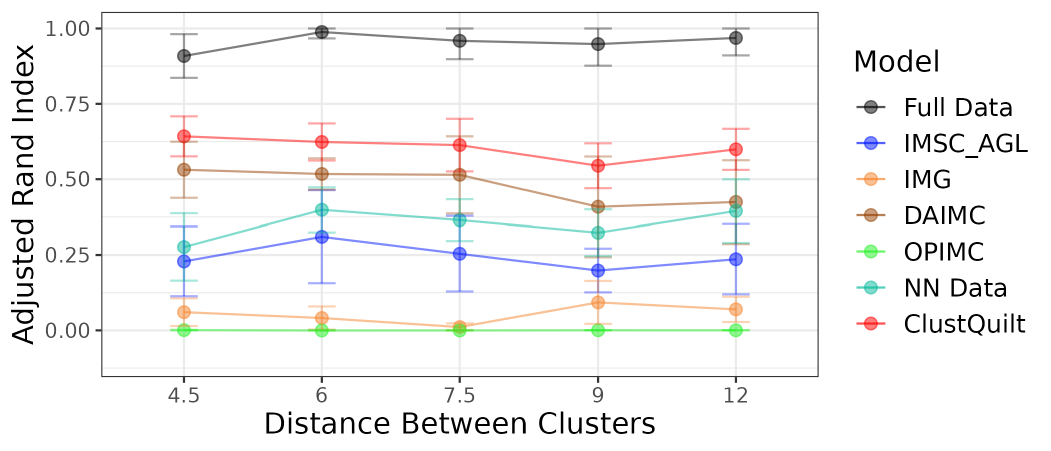}
        \caption{Changing cluster distance; \\ data-driven tuning.}
    \label{fig:sup_pmvc_d_dd}
    \end{subfigure}
\end{center}
    \caption{Performance of Cluster Quilting and comparison cluster imputation methods on copula mixture model data in the mosaic patch observation setting, measured by adjusted Rand index.}
    \label{fig:sup_pmvc}
\end{figure}

\clearpage
\section{Calcium Imaging Case Study}\label{ssec:ci}

Here, we apply Cluster Quilting to the problem of estimating functional neuronal clusters, i.e. how individual neurons can be organized into broad groups based on correlations in firing activity, using data collected via two-photon calcium imaging. Due to current technological limitations, only a small fraction of a brain region can be recorded with enough granularity to differentiate individual neurons. Neural activity data from calcium imaging for entire experiment is therefore generally collected via a series of scans at sequentially increasing depths, with some overlap between the neurons recorded between consecutive scans. This leads to the sequential patch observation scheme as described in Section \ref{s:intro}, and thus requires the implementation of an incomplete data clustering technique in order to estimate function clustering of neurons for a full brain region of interest. In this section, we study the applicability of the Cluster Quilting algorithm to calcium imaging data. Below, we analyze a data set from the Allen Institute containing fluorescence traces for 227 neurons measured across approximately 8000 observations during spontaneous activity. For this case study, we first apply spectral clustering to the complete, unmasked data set, with data-driven tuning of the rank and number of clusters using the prediction validation method described in Section \ref{sec:model}; we treat this as the ground truth clustering to which we compare the results of Cluster Quilting and competing methods. We then simulate the missingness pattern found in calcium imaging data collected in multiple patches by designating synthetic sequential patch observation blocks with overlap in the rows but not the columns as part of the empirical observed set of entries in the data matrix, and we use this as the input to the Cluster Quilting procedure and the comparison incomplete spectral clustering methods. We compare the procedures based on how well they recover the cluster assignments estimated by spectral clustering on the fully observed data for varying block sizes and for different numbers of sequential blocks while holding the total size of the set of observed entries in the data matrix constant. For each combination of the number of blocks and block size, we perform 50 replications and report the mean and standard deviation of the adjusted Rand index with respect to the cluster labels found using the fully observed data set for each method. We consider both oracle and data-driven tuning of the rank and the number of clusters; for this particular case study, we treat the rank and number of clusters found via data driven tuning with spectral clustering applied to the fully observed data set as the oracle rank and number of clusters to use as input to the incomplete spectral clustering methods.

\begin{figure}[t]
\begin{center}
    \begin{subfigure}[t]{0.4\linewidth}
    \centering
        \includegraphics[width=\linewidth]{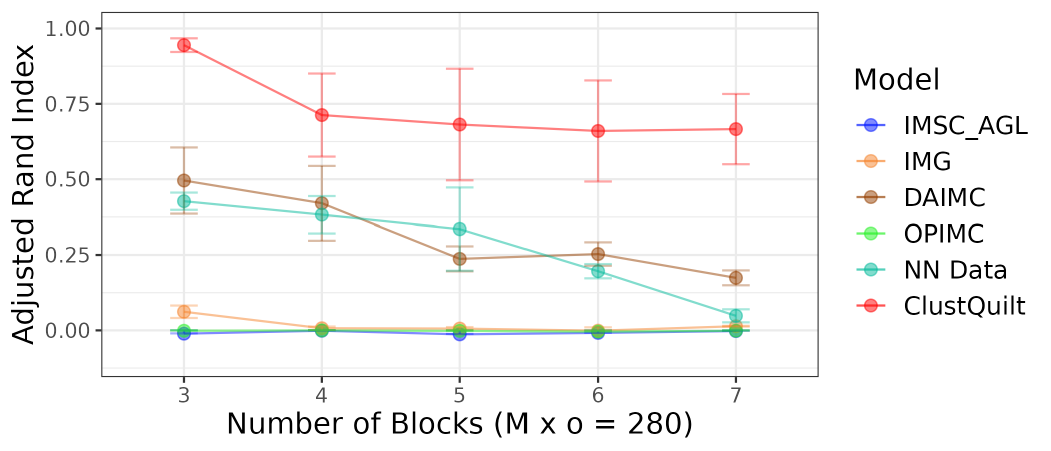}
        \caption{Changing number of blocks; \\oracle tuning.}
    \label{fig:aba_m}
    \end{subfigure}%
    \begin{subfigure}[t]{0.4\linewidth}
    \centering
        \includegraphics[width=\linewidth]{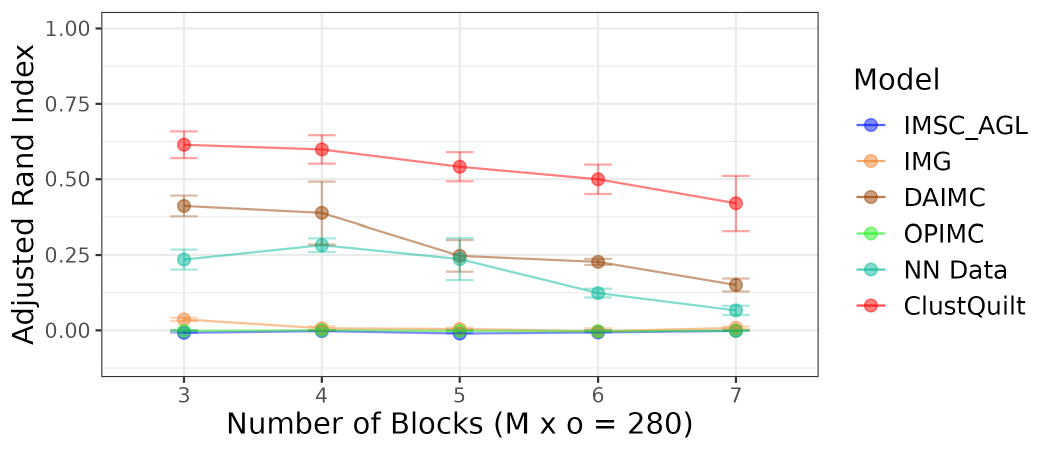}
        \caption{Changing number of blocks; \\ data-driven tuning.}
    \label{fig:aba_m_dd}
    \end{subfigure}
    \begin{subfigure}[t]{0.4\linewidth}
    \centering
        \includegraphics[width=\linewidth]{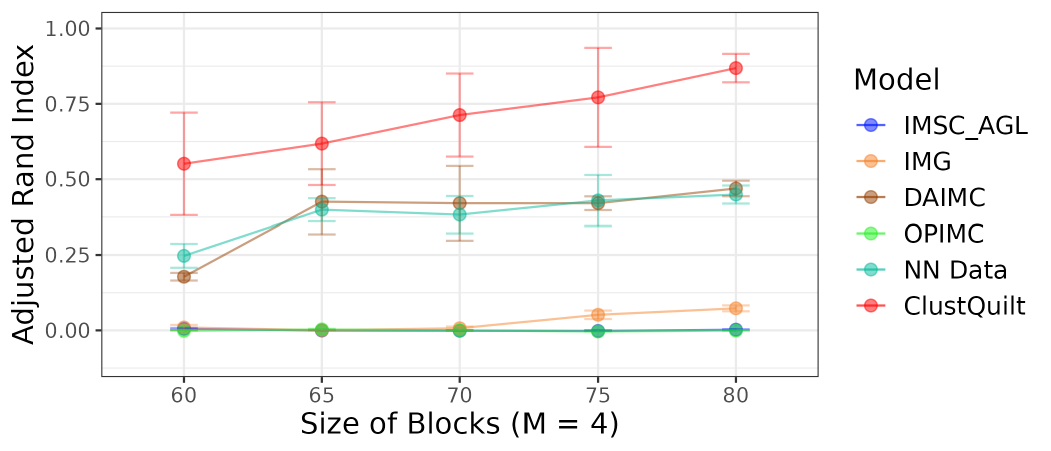}
        \caption{Changing size of observed blocks; \\ oracle tuning.}
    \label{fig:aba_o}
    \end{subfigure}%
    \begin{subfigure}[t]{0.4\linewidth}
    \centering
        \includegraphics[width=\linewidth]{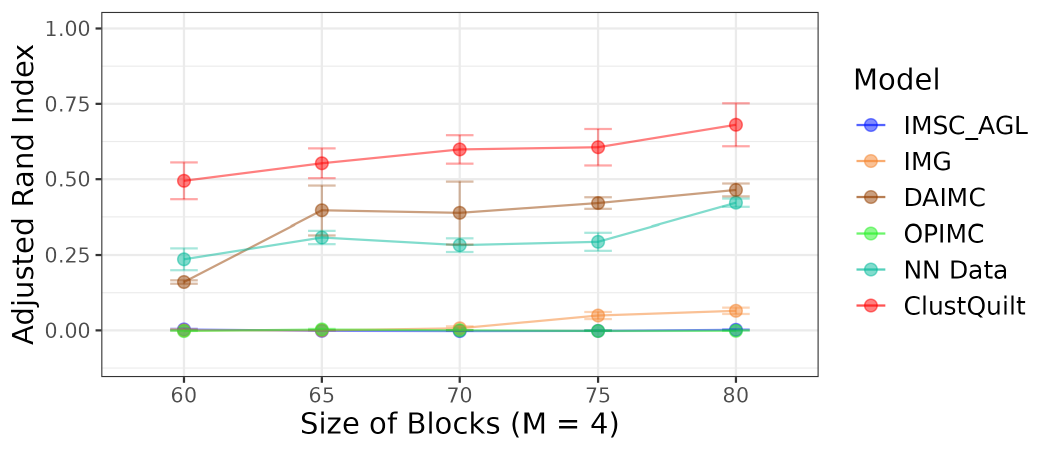}
        \caption{Changing size of observed blocks; \\ data-driven tuning.}
    \label{fig:aba_o_dd}
    \end{subfigure}
\end{center}
    \caption{Performance of Cluster Quilting and comparison cluster imputation methods on Allen Institute data with synthetic simulated block-missingness, measured by adjusted Rand index with respect to the cluster labels found with the fully observed data. In general, Cluster Quilting recovers cluster labels that most closely match those from the full data compared to other methods.}
    \label{fig:aba}
\end{figure}

The performance of Cluster Quilting and the comparison incomplete spectral clustering methods in recovering the cluster labels found with the fully observed data in Figure \ref{fig:aba}, with oracle hyperparameter tuning results in the left column and data-driven hyperparameter tuning in the right column. Across the varying block sizes and number of blocks, we find that the Cluster Quilting is the best performing method in terms of recovering the same clustering assignments that are assigned by applying spectral clustering to the fully observed data. Cluster labels from the imputed similarity matrix from the block singular value decomposition also matches relatively well on average with those estimated from the fully observed data. On the other hand, cluster assignments from applying spectral clustering after zero imputation of the incomplete similarity matrix is by far the worst performing method. These results show that Cluster Quilting may be the most appropriate method in order to estimate functional clusters in neurons recorded across multiple calcium imaging scans. As in the simulation studies in the main text, we see that increasing both the proportion of the data matrix that is observed and the size of the intersections between each pair of consecutive blocks, achieved by increasing the number of neurons in each observation block (Figures \ref{fig:aba_m} and \ref{fig:aba_m_dd}) or decreasing the total number of observation blocks amongst which the neurons are divided (Figures \ref{fig:aba_o} and \ref{fig:aba_o_dd}), increases the concordance between the clusters predicted by the imputation methods and the cluster labels found from applying spectral clustering on the full data.

\begin{figure}[t]
\begin{center}
    \begin{subfigure}[t]{0.4\linewidth}
    \centering
        \includegraphics[width=\linewidth]{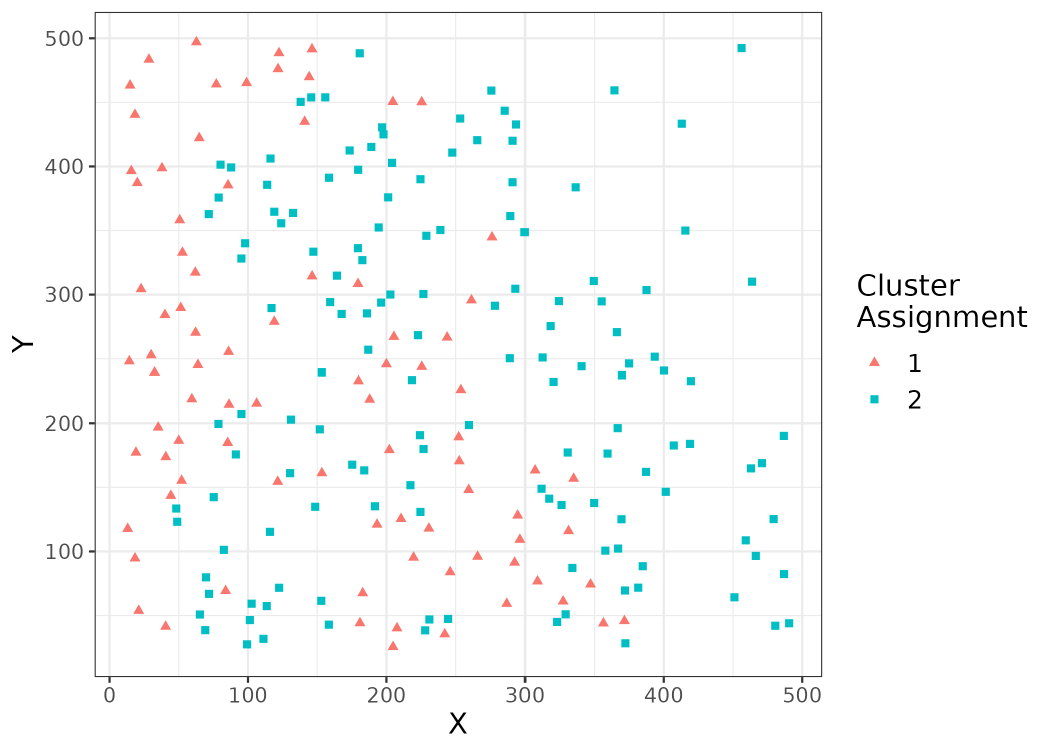}
        \caption{Full data.}
    \label{fig:full_loc}
    \end{subfigure}%
    \hspace{0.5cm}
    \begin{subfigure}[t]{0.4\linewidth}
    \centering
        \includegraphics[width=\linewidth]{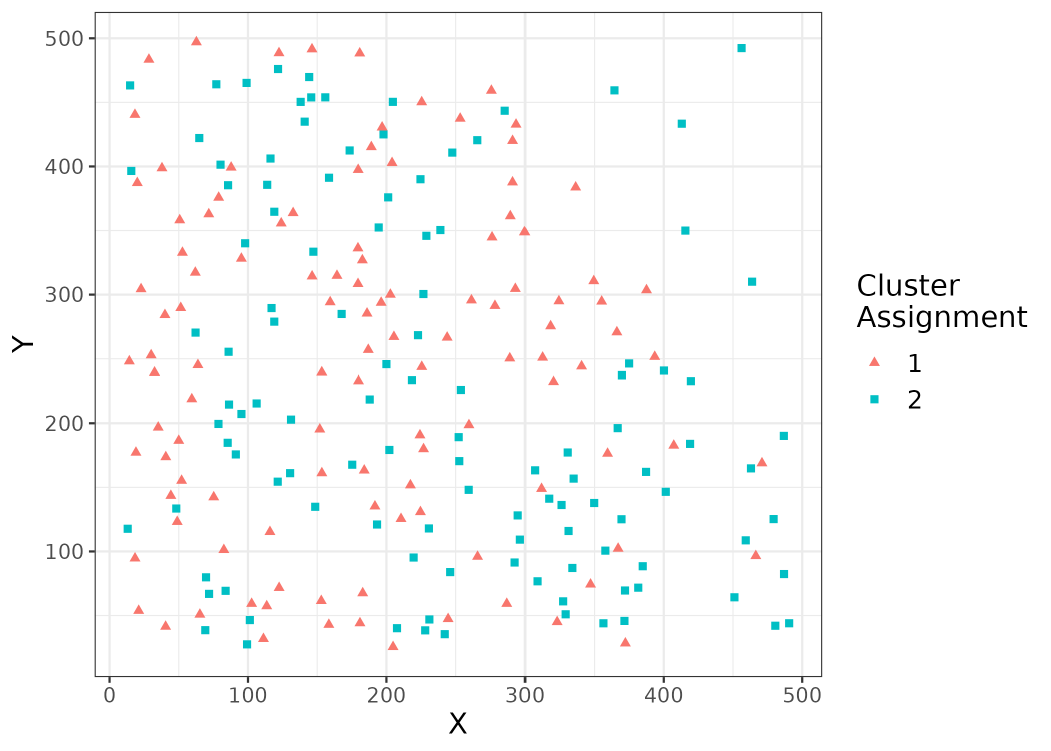}
        \caption{DAIMC.}
    \label{fig:daimc_loc}
    \end{subfigure}
    \begin{subfigure}[t]{0.4\linewidth}
    \centering
        \includegraphics[width=\linewidth]{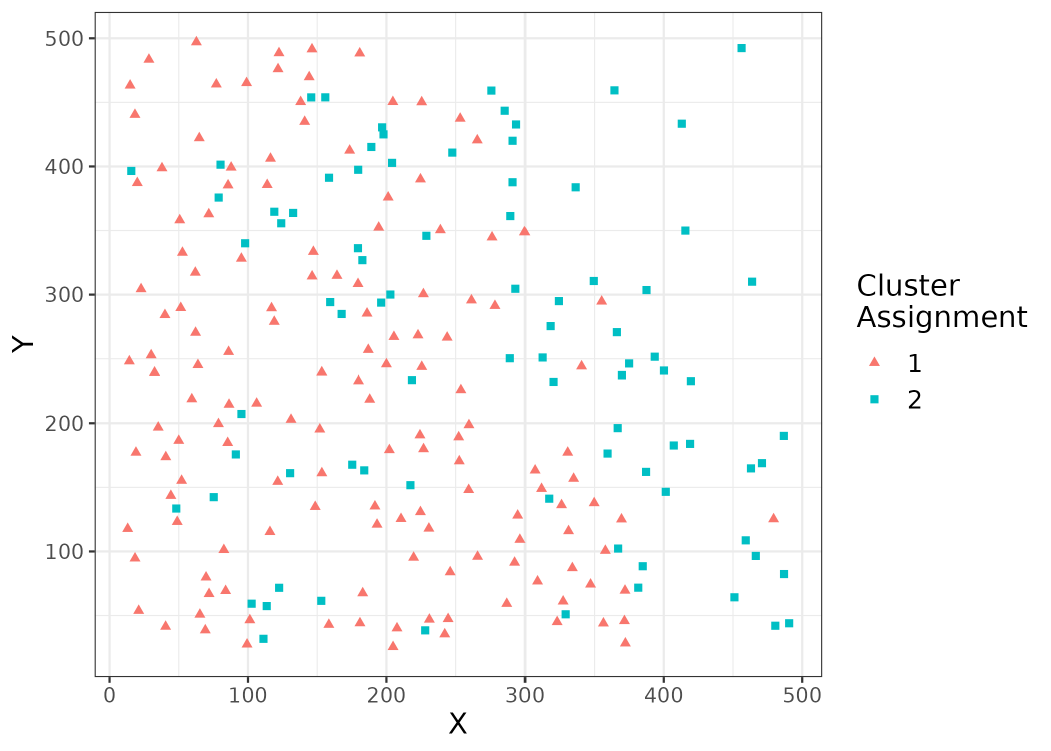}
        \caption{Likelihood data imputation.}
    \label{fig:imp_loc}
    \end{subfigure}%
    \hspace{0.5cm}
    \begin{subfigure}[t]{0.4\linewidth}
    \centering
        \includegraphics[width=\linewidth]{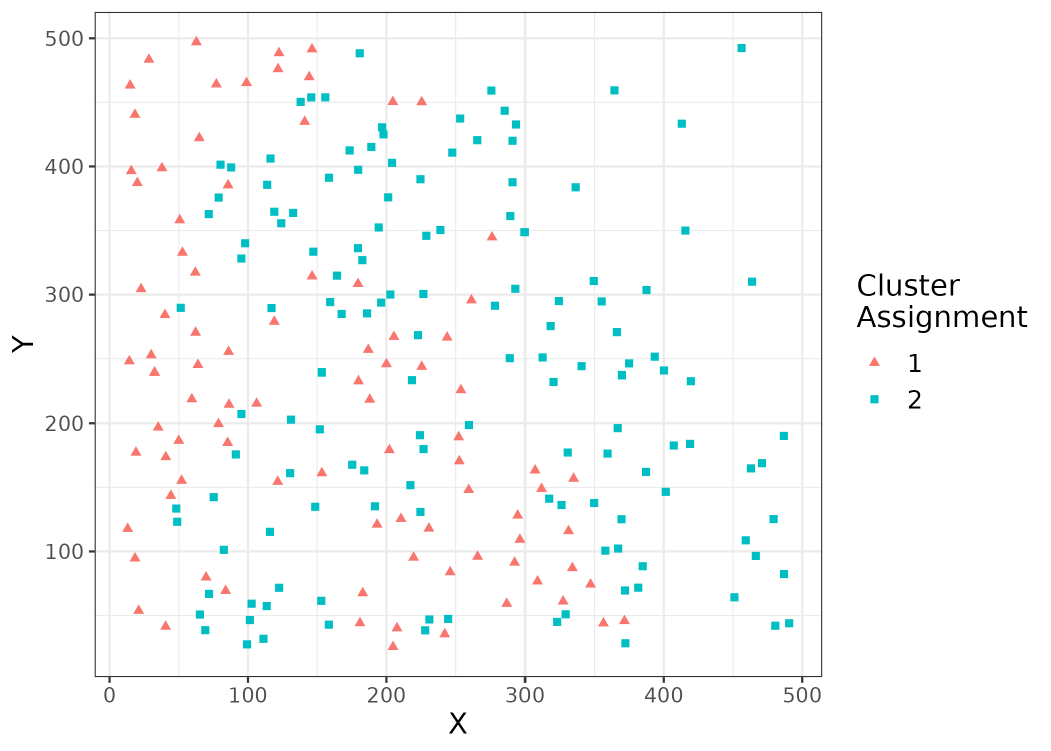}
        \caption{Cluster Quilting.}
    \label{fig:cqq_loc}
    \end{subfigure}
\end{center}
    \caption{Spatial distribution of functional clusters estimated from neuronal activity data collected via two-photon calcium imaging. We find that the spatial patterns of the functional clusters found via Cluster Quilting \textit{(d)} are most closely aligned with those derived from the full data \textit{(a)}, as compared to the functional clusters estimated with other incomplete spectral clustering techniques \textit{(b)} and nuclear-norm-regularized-likelihood imputation \textit{(c)} on the missing data. Calcium imaging provided by the Allen Institute.}
    \label{fig:aba_loc}
\end{figure}

We also analyze the spatial distribution of the functional neuronal clusters found with full data to those estimated via Cluster Quilting, alongside the estimates from the data with missing values imputed using nuclear-norm regularization in Figure \ref{fig:aba_loc}. In order to visualize the spatial patterns of the different clustering estimates, we show the cluster labels with respect to the location of the neurons in the brain slice. Figure \ref{fig:full_loc} shows the functional clusters estimated on the full data set. Here, we see that the estimated functional clusters appear to have a relationship with the spatial location of the neurons, as the neurons in the same function clusters tend to be grouped spatially into several distinct spatial areas; in particular, the clusters seem to be divided spatially along the horizontal axis of the brain slice. This finding follows what has been found previously in the study of functional neuronal clustering from calcium imaging data \citep{poliakov1999limited, dombeck2009functional, wilson2016orientation}. In Figures \ref{fig:daimc_loc} and \ref{fig:imp_loc}, we show an example of the functional clusters estimated from imputation of the masked data via the DAIMC and nuclear-norm-regularized likelihood methods, respectively. We see in these cases that the spatial patterns that underlie the functional clusters estimated from the full data are much less apparent in the functional clustering assignments from the data imputed using the low-rank factorization method and zero imputation, instead having much less structured spatial distributions. On the other hand, in Figure \ref{fig:cqq_loc}, which shows the functional clusters identified by applying Cluster Quilting to the same masked data, we see more separable spatial patterns as well as cluster assignments that more closely match the cluster labels found from fully observed data. From this case study, we find that Cluster Quilting provides imputed functional neuronal clustering estimates from calcium imaging data that better match what would be found if the full data were to be observed and which correspond more closely with what we would expect in the specific neuroscience context. specifically with respect to the relationship to the spatial distributions of the functional clusters.

\bibliographystyle{apalike}
\bibliography{reference}

\begin{thebibliography}{}

\bibitem[Abbe et~al., 2022]{abbe2022l}
Abbe, E., Fan, J., and Wang, K. (2022).
\newblock An $\ell_p$ theory of pca and spectral clustering.
\newblock {\em The Annals of Statistics}, 50(4):2359--2385.

\bibitem[Aggarwal and Zhai, 2012]{aggarwal2012survey}
Aggarwal, C. and Zhai, C. (2012).
\newblock {\em A survey of text clustering algorithms}.
\newblock Springer.

\bibitem[Aktaş et~al., 2019]{aktas2019data}
Aktaş, M., Kaplan, S., Abacı, H., Kalipsiz, O., Ketenci, U., and Turgut, U.
  (2019).
\newblock Data imputation methods for missing values in the context of
  clustering.
\newblock In {\em Big Data and Knowledge Sharing in Virtual Organizations},
  pages 240--274. IGI Global.

\bibitem[Baghfalaki et~al., 2016]{baghfalaki2016missing}
Baghfalaki, T., Ganjali, M., and Berridge, D. (2016).
\newblock Missing value imputation for rna-sequencing data using statistical
  models: a comparative study.
\newblock {\em Journal of Statistical Theory and Applications}, 15(3):221--236.

\bibitem[Bishop and Yu, 2014]{bishop2014deterministic}
Bishop, W.~E. and Yu, B.~M. (2014).
\newblock Deterministic symmetric positive semidefinite matrix completion.
\newblock {\em Advances in Neural Information Processing Systems}, 27.

\bibitem[Boluki et~al., 2019]{boluki2019optimal}
Boluki, S., Zamani~Dadaneh, S., Qian, X., and Dougherty, E.~R. (2019).
\newblock Optimal clustering with missing values.
\newblock {\em BMC bioinformatics}, 20:1--10.

\bibitem[Bridges, 1966]{bridges1966hierarchical}
Bridges, C. (1966).
\newblock Hierarchical cluster analysis.
\newblock {\em Psychological Reports}, 18(3):851--854.

\bibitem[Browning, 2008]{browning2008missing}
Browning, S. (2008).
\newblock Missing data imputation and haplotype phase inference for genome-wide
  association studies.
\newblock {\em Human Genetics}, 124:439--450.

\bibitem[Cai et~al., 2019]{cai2019survey}
Cai, Q., Wang, H., Li, Z., and Liu, X. (2019).
\newblock A survey on multimodal data-driven smart healthcare systems:
  approaches and applications.
\newblock {\em IEEE Access}, 7:133583--133599.

\bibitem[Chang et~al., 2022]{chang2022lowrank}
Chang, A., Zheng, L., and Allen, G. (2022).
\newblock Low-rank covariance completion for graph quilting with applications
  to functional connectivity.
\newblock {\em arXiv preprint arXiv:2209.08273}.

\bibitem[Chatterjee, 2020]{chatterjee2020deterministic}
Chatterjee, S. (2020).
\newblock A deterministic theory of low rank matrix completion.
\newblock {\em IEEE Transactions on Information Theory}, 66(12):8046--8055.

\bibitem[Chen et~al., 2015]{chen2015convex}
Chen, G., Chi, E., Ranola, J., and Lange, K. (2015).
\newblock Convex clustering: An attractive alternative to hierarchical
  clustering.
\newblock {\em PLoS Computational Biology}, 11(5):e1004228.

\bibitem[Chen et~al., 2021]{chen2021spectral}
Chen, Y., Chi, Y., Fan, J., Ma, C., et~al. (2021).
\newblock Spectral methods for data science: A statistical perspective.
\newblock {\em Foundations and Trends{\textregistered} in Machine Learning},
  14(5):566--806.

\bibitem[Chi et~al., 2016]{chi2016k}
Chi, J.~T., Chi, E.~C., and Baraniuk, R.~G. (2016).
\newblock k-pod: A method for k-means clustering of missing data.
\newblock {\em The American Statistician}, 70(1):91--99.

\bibitem[Choi and Yuan, 2023]{choi2023matrix}
Choi, J. and Yuan, M. (2023).
\newblock Matrix completion when missing is not at random and its applications
  in causal panel data models.
\newblock {\em arXiv preprint arXiv:2308.02364}.

\bibitem[Dalton et~al., 2009]{dalton2009clustering}
Dalton, L., Ballarin, V., and Brun, M. (2009).
\newblock Clustering algorithms: on learning, validation, performance, and
  applications to genomics.
\newblock {\em Current Genomics}, 10(6):430--445.

\bibitem[Dombeck et~al., 2009]{dombeck2009functional}
Dombeck, D.~A., Graziano, M.~S., and Tank, D.~W. (2009).
\newblock Functional clustering of neurons in motor cortex determined by
  cellular resolution imaging in awake behaving mice.
\newblock {\em Journal of Neuroscience}, 29(44):13751--13760.

\bibitem[Donath and Hoffman, 1972]{donath1972algorithms}
Donath, W. and Hoffman, A. (1972).
\newblock Algorithms for partitioning of graphs and computer logic based on
  eigenvectors of connection matrices.
\newblock {\em IBM Technical Disclosure Bulletin}, 15(3):938--944.

\bibitem[Gemma et~al., 2001]{gemma2001genomic}
Gemma, A., Hosoya, Y., Seike, M., Uematsu, K., Kurimoto, F., Hibino, S.,
  Yoshimura, A., Shibuya, M., Kudoh, S., and Emi, M. (2001).
\newblock Genomic structure of the human mad2 gene and mutation analysis in
  human lung and breast cancers.
\newblock {\em Lung Cancer}, 32(3):289--295.

\bibitem[Hathaway and Bezdek, 2001]{hathaway2001fuzzy}
Hathaway, R.~J. and Bezdek, J.~C. (2001).
\newblock Fuzzy c-means clustering of incomplete data.
\newblock {\em IEEE Transactions on Systems, Man, and Cybernetics, Part B
  (Cybernetics)}, 31(5):735--744.

\bibitem[Hu and Chen, 2019a]{hu2019doubly}
Hu, M. and Chen, S. (2019a).
\newblock Doubly aligned incomplete multi-view clustering.
\newblock {\em arXiv preprint arXiv:1903.02785}.

\bibitem[Hu and Chen, 2019b]{hu2019onepass}
Hu, M. and Chen, S. (2019b).
\newblock One-pass incomplete multi-view clustering.
\newblock In {\em Proceedings of the AAAI Conference on Artificial
  Intelligence}, volume~33, pages 3838--3845.

\bibitem[Hunt and Jorgensen, 2003]{hunt2003mixture}
Hunt, L. and Jorgensen, M. (2003).
\newblock Mixture model clustering for mixed data with missing information.
\newblock {\em Computational statistics \& data analysis}, 41(3-4):429--440.

\bibitem[Inman et~al., 2015]{inman2015case}
Inman, D., Elmore, R., and Bush, B. (2015).
\newblock A case study to examine the imputation of missing data to improve
  clustering analysis of building electrical demand.
\newblock {\em Building Services Engineering Research and Technology},
  36(5):628--637.

\bibitem[Jin, 2015]{jin2015fast}
Jin, J. (2015).
\newblock Fast community detection by score.
\newblock {\em THE ANNALS}, 43(1):57--89.

\bibitem[Jin et~al., 2018]{jin2018breast}
Jin, L., Han, B., Siegel, E., Cui, Y., Giuliano, A., and Cui, X. (2018).
\newblock Breast cancer lung metastasis: Molecular biology and therapeutic
  implications.
\newblock {\em Cancer biology \& therapy}, 19(10):858--868.

\bibitem[Kline et~al., 2022]{kline2022multimodal}
Kline, A., Wang, H., Li, Y., Dennis, S., Hutch, M., Xu, Z., Wang, F., Cheng,
  F., and Luo, Y. (2022).
\newblock Multimodal machine learning in precision health: A scoping review.
\newblock {\em npj Digital Medicine}, 5(1):171.

\bibitem[Landemaine et~al., 2008]{landemaine2008six}
Landemaine, T., Jackson, A., Bellahcene, A., Rucci, N., Sin, S., Abad, B.~M.,
  Sierra, A., Boudinet, A., Guinebretiere, J.-M., Ricevuto, E., and Nogues, C.
  (2008).
\newblock A six-gene signature predicting breast cancer lung metastasis.
\newblock {\em Cancer research}, 68(15):6092--6099.

\bibitem[Lei and Rinaldo, 2015]{lei2015consistency}
Lei, J. and Rinaldo, A. (2015).
\newblock Consistency of spectral clustering in stochastic block models.
\newblock {\em The Annals of Statistics}, 43(1):215--237.

\bibitem[Lin et~al., 2006]{lin2006fast}
Lin, T.~I., Lee, J.~C., and Ho, H.~J. (2006).
\newblock On fast supervised learning for normal mixture models with missing
  information.
\newblock {\em Pattern Recognition}, 39(6):1177--1187.

\bibitem[Lock and Dunson, 2013]{lock2013bayesian}
Lock, E.~F. and Dunson, D.~B. (2013).
\newblock Bayesian consensus clustering.
\newblock {\em Bioinformatics}, 29(20):2610--2616.

\bibitem[L{\"o}ffler et~al., 2021]{loffler2021optimality}
L{\"o}ffler, M., Zhang, A.~Y., and Zhou, H.~H. (2021).
\newblock Optimality of spectral clustering in the gaussian mixture model.
\newblock {\em The Annals of Statistics}, 49(5):2506--2530.

\bibitem[MacQueen, 1967]{macqueen1967some}
MacQueen, J. (1967).
\newblock Some methods for classification and analysis of multivariate
  observations.
\newblock In {\em Proceedings of the Fifth Berkeley Symposium on Mathematical
  Statistics and Probability, Volume 1}, pages 281--297. University of
  California Press.

\bibitem[{MICrONS Consortium}, 2021]{microns2021functional}
{MICrONS Consortium} (2021).
\newblock Functional connectomics spanning multiple areas of mouse visual
  cortex.
\newblock {\em BioRxiv}, pages 2021--07.

\bibitem[Mishra et~al., 2007]{mishra2007clustering}
Mishra, N., Schreiber, R., Stanton, I., and Tarjan, R. (2007).
\newblock Clustering social networks.
\newblock In {\em International Workshop on Algorithms and Models for the
  Web-Graph}, pages 56--67, Berlin, Heidelberg. Springer Berlin Heidelberg.

\bibitem[Montel et~al., 2005]{montel2005expression}
Montel, V., Huang, T.-Y., Mose, E., Pestonjamasp, K., and Tarin, D. (2005).
\newblock Expression profiling of primary tumors and matched lymphatic and lung
  metastases in a xenogeneic breast cancer model.
\newblock {\em The American journal of pathology}, 166(5):1565--1579.

\bibitem[Poddar and Jacob, 2019]{poddar2019clustering}
Poddar, S. and Jacob, M. (2019).
\newblock Clustering of data with missing entries using non-convex fusion
  penalties.
\newblock {\em IEEE transactions on signal processing}, 67(22):5865--5880.

\bibitem[Poliakov and Schieber, 1999]{poliakov1999limited}
Poliakov, A.~V. and Schieber, M.~H. (1999).
\newblock Limited functional grouping of neurons in the motor cortex hand area
  during individuated finger movements: a cluster analysis.
\newblock {\em Journal of Neurophysiology}, 82(6):3488--3505.

\bibitem[Rajendran et~al., 2023]{rajendran2023patchwork}
Rajendran, S., Pan, W., Sabuncu, M.~R., Zhou, J., and Wang, F. (2023).
\newblock Patchwork learning: A paradigm towards integrative analysis across
  diverse biomedical data sources.
\newblock {\em arXiv preprint arXiv:2305.06217}.

\bibitem[Rohe et~al., 2011]{rohe2011spectral}
Rohe, K., Chatterjee, S., and Yu, B. (2011).
\newblock Spectral clustering and the high-dimensional stochastic blockmodel.
\newblock {\em Annals of statistics}, 39(4):1878--1915.

\bibitem[Santos and Embrechts, 2009]{santos2009use}
Santos, J. and Embrechts, M. (2009).
\newblock On the use of the adjusted rand index as a metric for evaluating
  supervised classification.
\newblock In {\em International Conference on Artificial Neural Networks},
  pages 175--184, Berlin, Heidelberg. Springer Berlin Heidelberg.

\bibitem[Schwartz et~al., 1999]{schwartz1999familial}
Schwartz, A.~G., Siegfried, J.~M., and Weiss, L. (1999).
\newblock Familial aggregation of breast cancer with early onset lung cancer.
\newblock {\em Genetic Epidemiology: The Official Publication of the
  International Genetic Epidemiology Society}, 17(4):274--284.

\bibitem[Shen et~al., 2009]{shen2009integrative}
Shen, R., Olshen, A., and Ladanyi, M. (2009).
\newblock Integrative clustering of multiple genomic data types using a joint
  latent variable model with application to breast and lung cancer subtype
  analysis.
\newblock {\em Bioinformatics}, 25(22):2906--2912.

\bibitem[Tibshirani and Walther, 2005]{tibshirani2005cluster}
Tibshirani, R. and Walther, G. (2005).
\newblock Cluster validation by prediction strength.
\newblock {\em Journal of Computational and Graphical Statistics},
  14(3):511--528.

\bibitem[Vinci et~al., 2019]{vinci2019graph}
Vinci, G., Dasarathy, G., and Allen, G.~I. (2019).
\newblock Graph quilting: graphical model selection from partially observed
  covariances.
\newblock {\em arXiv preprint arXiv:1912.05573}.

\bibitem[Wang and Allen, 2021]{wang2021integrative}
Wang, M. and Allen, G.~I. (2021).
\newblock Integrative generalized convex clustering optimization and feature
  selection for mixed multi-view data.
\newblock {\em Journal of Machine Learning Research}, 22(55):1--73.

\bibitem[Wang et~al., 2019]{wang2019k}
Wang, S., Li, M., Hu, N., Zhu, E., Hu, J., Liu, X., and Yin, J. (2019).
\newblock K-means clustering with incomplete data.
\newblock {\em IEEE Access}, 7:69162--69171.

\bibitem[{Weinstein, J.N. et. al.}, 2013]{weinstein2013cancer}
{Weinstein, J.N. et. al.} (2013).
\newblock The cancer genome atlas pan-cancer analysis project.
\newblock {\em Nature Genetics}, 45(10):1113--1120.

\bibitem[Wen et~al., 2018]{wen2018incomplete}
Wen, J., Xu, Y., and Liu, H. (2018).
\newblock Incomplete multiview spectral clustering with adaptive graph
  learning.
\newblock {\em IEEE Transactions on Cybernetics}, 50(4):1418--1429.

\bibitem[Wen et~al., 2022]{wen2022survey}
Wen, J., Zhang, Z., Fei, L., Zhang, B., Xu, Y., Zhang, Z., and Li, J. (2022).
\newblock A survey on incomplete multiview clustering.
\newblock {\em IEEE Transactions on Systems, Man, and Cybernetics: Systems},
  53(2):1136--1149.

\bibitem[Wilson et~al., 2016]{wilson2016orientation}
Wilson, D.~E., Whitney, D.~E., Scholl, B., and Fitzpatrick, D. (2016).
\newblock Orientation selectivity and the functional clustering of synaptic
  inputs in primary visual cortex.
\newblock {\em Nature neuroscience}, 19(8):1003--1009.

\bibitem[Yan and Wainwright, 2024]{yan2024entrywise}
Yan, Y. and Wainwright, M.~J. (2024).
\newblock Entrywise inference for causal panel data: A simple and
  instance-optimal approach.
\newblock {\em arXiv preprint arXiv:2401.13665}.

\bibitem[Yu et~al., 2020]{yu2020optimal}
Yu, G., Li, Q., Shen, D., and Liu, Y. (2020).
\newblock Optimal sparse linear prediction for block-missing multi-modality
  data without imputation.
\newblock {\em Journal of the American Statistical Association},
  115(531):1406--1419.

\bibitem[Zhao et~al., 2016]{zhao2016incomplete}
Zhao, H., Liu, H., and Fu, Y. (2016).
\newblock Incomplete multi-modal visual data grouping.
\newblock In {\em IJCAI}, pages 2392--2398.

\bibitem[Zhou and Chen, 2023]{zhou2023heteroskedastic}
Zhou, Y. and Chen, Y. (2023).
\newblock Heteroskedastic tensor clustering.
\newblock {\em arXiv preprint arXiv:2311.02306}, 3.

\bibitem[Zhou et~al., 2022]{zhou2022optimal}
Zhou, Y., Zhang, A.~R., Zheng, L., and Wang, Y. (2022).
\newblock Optimal high-order tensor svd via tensor-train orthogonal iteration.
\newblock {\em IEEE transactions on information theory}, 68(6):3991--4019.

\end{thebibliography}
	
\end{document}